\documentclass[a4paper,english]{article}
\usepackage[latin9]{inputenc}
\usepackage{babel}
\usepackage{amsmath}
\usepackage{amssymb}
\usepackage[unicode=true,
 bookmarks=true,bookmarksnumbered=false,bookmarksopen=false,
 breaklinks=false,pdfborder={0 0 1},backref=false,colorlinks=false]
 {hyperref}
\hypersetup{pdftitle={On projections to the pure spinor space},
 pdfauthor={Grassi / Guttenberg}}
\usepackage{breakurl}

\makeatletter

\providecommand{\LyX}{\texorpdfstring%
  {L\kern-.1667em\lower.25em\hbox{Y}\kern-.125emX\@}
  {LyX}}

\usepackage{amssymb} 
\usepackage{amsthm} 
\usepackage{ifthen}
\newcounter{draft}
\ifthenelse{\thedraft=0}{}{\usepackage[notref,notcite]{showkeys}}
\newcounter{alteFormelNr}\newcounter{alteFootnoteNr}

\newcounter{remark}
\newcommand{\rem}[1]{%
\ifthenelse{\thedraft=2}{%
\addtocounter{remark}{1}%
\setcounter{alteFormelNr}{\value{equation}}%
\setcounter{alteFootnoteNr}{\value{footnote}}%
\marginpar[$\boldsymbol{{}_{\theremark}==>}$]{$\boldsymbol{<==_{\theremark}}$}%
{\sc\tt  <<~#1~>>$_{\theremark}$ }%
\setcounter{equation}{\value{alteFormelNr}}%
\setcounter{footnote}{\value{alteFootnoteNr}}%
}{}%
}
\newcommand{\rembreak}{\ifthenelse{\thedraft=2}{\\}{}}
\newcommand{\frem}[1]{
\addtocounter{remark}{1}%
\ifthenelse{\thedraft=2}{%
\setcounter{alteFormelNr}{\value{equation}}%
{\sc\tt  <<~#1~>>\ensuremath{_{\theremark}} }%
\setcounter{equation}{\value{alteFormelNr}}%
}{}%
}

\usepackage{hyperref}

\numberwithin{equation}{section}

\makeatother

\begin{document}
\global\long\def\bs#1{\boldsymbol{#1}}
\global\long\def\mf#1{\mathfrak{#1} }
\global\long\def\mc#1{\mathcal{#1}}
\global\long\def\ip{\imath}
\global\long\def\partiell#1#2{\frac{\partial#1}{\partial#2}}
\global\long\def\partl#1{\frac{\partial}{\partial#1}}
\global\long\def\total#1#2{\frac{d#1}{d#2}}
\global\long\def\totl#1{\frac{d}{d#1}}
\global\long\def\funktional#1#2{\frac{\delta#1}{\delta#2}}
\global\long\def\funktl#1{\frac{\delta}{\delta#1}}
\global\long\def\de{\bs{{\rm d}}\hspace{-0.01cm}}
\global\long\def\De{\bs{{\rm D}}\!}
\global\long\def\cov{\textrm{D}\!}
\global\long\def\es{\bs{{\rm s}}\hspace{-0.01cm}}
\global\long\def\pe{\bs{\partial}}
\global\long\def\ola#1{\overleftarrow{#1}}
\global\long\def\partialr{\overleftarrow{\partial}}
\global\long\def\partr#1{\frac{\partialr}{\partial#1}}
\global\long\def\totr#1{\frac{\ola d}{d#1}}
\global\long\def\funktr#1{\frac{\ola{\delta}}{\delta#1}}
\global\long\def\ket#1{|#1 \rangle}
\global\long\def\bra#1{\langle#1 |}
\global\long\def\Ket#1{\left| #1 \right\rangle }
\global\long\def\Bra#1{\left\langle #1 \right| }
 \global\long\def\braket#1#2{\langle#1 |#2 \rangle}
\global\long\def\Braket#1#2{\Bra{#1 }\left. #2 \right\rangle }
\global\long\def\abs#1{\mid#1 \mid}
\global\long\def\Abs#1{\left| #1 \right| }
\global\long\def\erw#1{\langle#1\rangle}
\global\long\def\Erw#1{\left\langle #1 \right\rangle }
\global\long\def\bei#1#2{\left. #1 \right| _{#2 }}
\global\long\def\dann{\Rightarrow}
\global\long\def\q#1{\underline{#1 }}
\global\long\def\hoch#1{^{#1 }}
\global\long\def\tief#1{_{#1 }}
\renewcommand{\hoch}[1]{{}^{#1}}\renewcommand{\tief}[1]{{}_{#1}}\global\long\def\lqn#1{\lefteqn{#1}}
\global\long\def\slqn#1{\lefteqn{{\scriptstyle #1}}}
\global\long\def\os#1#2{\overset{#1}{#2}}
\global\long\def\us#1#2{\underset{#2}{#1}}
\global\long\def\ous#1#2#3{\underset{#3}{\os{#1}{#2}}}
\global\long\def\zwek#1#2{\begin{array}{c}
 #1\\
#2 
\end{array}}
\global\long\def\drek#1#2#3{\begin{array}{c}
 #1\\
#2\\
#3 
\end{array}}
\global\long\def\tr{\textrm{tr}\,}
\global\long\def\Tr{\textrm{Tr}\,}
\global\long\def\Det{\textrm{Det}\,}
\global\long\def\diag{\textrm{diag}\,}
\global\long\def\Diag{\textrm{Diag}\,}
\global\long\def\sgn{\textrm{sgn\,}}
\global\long\def\one{1\!\!1}
\global\long\def\fussend{\diamond}
\global\long\def\eps{\varepsilon}
\global\long\def\feps{{\bs{\eps}}}
\global\long\def\tet{\bs{\theta}}
\global\long\def\Tet{\bs{\Theta}}
\global\long\def\dali{\Box}

\global\long\def\ws#1{\q{#1}}
\global\long\def\wsm{\ws m}
\global\long\def\wsn{\ws n}
\global\long\def\wsl{\ws l}
\global\long\def\wsk{\ws k}
\global\long\def\wsa{\ws a}
\global\long\def\wsb{\ws b}
\global\long\def\wsc{\ws c}
\global\long\def\wsd{\ws d}
\global\long\def\sgn{\textrm{sgn\,}}
\global\long\def\one{1\!\!1}
\global\long\def\fussend{\diamond}
\global\long\def\eps{\varepsilon}
\global\long\def\dali{\Box}
\global\long\def\lcconst{c}
\global\long\def\allfields#1{\phi_{\textrm{all}}^{\mc{#1}}}
\global\long\def\lc#1{\us{#1}{\llcorner}}
\global\long\def\rc#1{\us{#1}{\lrcorner}}
\global\long\def\asympartial{\partial}
\global\long\def\asymD{{\rm D}}
\global\long\def\syml{\langle}
\global\long\def\symr{\rangle}
\global\long\def\asyml{\lfloor}
\global\long\def\asymr{\rfloor}
\global\long\def\asymdelta{\check{\delta}}
\global\long\def\norm#1{{\parallel#1 \parallel}}
\global\long\def\Norm#1{\left\Vert #1 \right\Vert }
\newcommand{\monthof}[1]{\protect\ifthenelse{#1=12}{December}{\ifthenelse{#1=01}{January}{\ifthenelse{#1=02}{February}{\ifthenelse{#1=03}{March}{\ifthenelse{#1=04}{April}{\ifthenelse{#1=05}{May}{\ifthenelse{#1=06}{June}{\ifthenelse{#1=07}{July}{\ifthenelse{#1=08}{August}{\ifthenelse{#1=09}{September}{\ifthenelse{#1=10}{October}{\ifthenelse{#1=11}{November}{??}}}}}}}}}}}}}
\def\parsedate #1:20#2#3#4#5#6#7#8\empty{\monthof{#4#5} #6#7, 20#2#3} \def\moddate#1{\expandafter\parsedate\pdffilemoddate{#1}\empty} 

\global\long\def\auxi{\mf{\xi}}
\global\long\def\auxii{\mf{\zeta}}
\global\long\def\auxiii{\mf{\eta}}
\global\long\def\rest{\pi}
\global\long\def\Auxi{\tilde{\auxi}}
\global\long\def\Auxii{\tilde{\auxii}}
\global\long\def\fa{\tilde{f}}
\global\long\def\Hauxii{\check{\zeta}}
\global\long\def\Hauxi{\check{\xi}}
\global\long\def\fah{\check{f}}

\title{On Projections to the Pure Spinor Space}

\author{\begin{picture}(0,0)\unitlength=1mm\put(80,50){DISTA-2011}\put(80,45){DFTT 25/2011}
\end{picture} P.A. Grassi$^{*}$ and S. Guttenberg$^{\dagger,\ddagger}$\vspace{0.5cm}
\\
{\footnotesize }{\small $^{*}$DISTA, Università del Piemonte Orientale,\vspace{-.15cm}}\\
{\small Via Teresa Michel 11, Alessandria, 15121, Italy, \vspace{-.15cm}}\\
{\small and INFN - Sezione di Torino.\vspace{.2cm}}\\
{\small $^{\dagger}$Dipartimento di Fisica Teorica, Università di
Torino,\vspace{-.15cm}}\\
{\small Via P. Giuria 1, I-10125 Torino, Italy \vspace{-.15cm}}\\
{\small{} and INFN - Sezione di Torino.\vspace{.2cm}}\\
{\small $^{\ddagger}$CAMGSD, Departamento de Matemática, \vspace{-.15cm}}\\
{\small Instituto Superior Técnico,\vspace{-.15cm}}\\
{\small Av. Rovisco Pais 1,\vspace{-.15cm}}\\
{\small 1049\textendash{}001 Lisboa, Portugal. \vspace{.2cm}}\texttt{\small }~\\
\texttt{\small pietro.grassi@mfn.unipmn.it}\texttt{\emph{\small ,
}}\texttt{\small sgutten@math.ist.utl.pt}}

\date{{}}
\maketitle
\begin{abstract}
A family of covariant non-linear projections from the space of SO(10)
Weyl spinors onto the space of pure SO(10) Weyl spinors is presented.
The Jacobian matrices of these projections are related to a linear
projector which was previously discussed in pure spinor string literature
and which maps the antighost to its gauge invariant part. Only one
representative of the family leads to a Hermitian Jacobian matrix
and can itself be derived from a scalar potential. Comments on the
SO(1,9) case are given as well as on the non-covariant version of
the projection map. The insight is applied to the ghost action of
pure spinor string theory, where the constraints on the fields can
be removed using the projection, while introducing new gauge symmetries.
This opens the possibility of choosing different gauges which might
help to clarify the origin of the pure spinor ghosts. Also the measure
of the pure spinor space is discussed from the projection point of
view. The appendix contains the discussion of a toy model which served
as a guideline for the pure spinor case.
\end{abstract}
\newpage \tableofcontents{}\vspace{.5cm}
\newtheorem{prop}{Proposition}

\section{Introduction}

\label{sec:Introduction} In recent years, pure spinors in 10 dimensions
have become very important for string theory. This is mainly due to
the invention by Berkovits \cite{Berkovits:2000fe} of a string theory
sigma-model where a pure spinor ghost field $\lambda^{\alpha}$ plays
a prominent role. But already before, pure spinors where used to describe
10-dimensional supersymmetric quantum field theories \cite{Nilsson:1985cm,Howe:1991mf,Howe:1991bx}.

In spite of its remarkable success, the pure spinor sigma model is
still conceptually not completely understood. Due to the presence
of the quadratic pure spinor constraint on the ghost fields, it is
highly challenging to derive the model from first principles as a
gauge-fixing of a classical theory. In particular the diffeomorphism
b-ghost appears only as a non-trivial composite field. 

There have been many efforts to remove the ghost-constraints by either
adding more ghosts like in \cite{Nh:2001ug} or \cite{Aisaka:2002sd}
(with finitely many ghosts) or like in \cite{Chesterman:2002ey} or
\cite{Lee:2005jy} (with infinitely many ghosts). Other efforts were
relating the pure spinor string to the superembedding-formalism \cite{SorokinMatone:2002ft}
or on the operator level directly to the Green Schwarz formalism \cite{Berkovits:2004tw,Berkovits:2007wz}.
In \cite{Gaona:2005yw} a connection to the pure spinor string was
obtained by working out the BFT conversion of second class constraints
(of the Green Schwarz action) into first class constraints. In addition
in \cite{Aisaka:2005vn} as well as recently in \cite{Oda:2011kh}
the pure spinor string was derived from classical ghost-free actions,
and an attempt to derive the b-ghost from first principles by coupling
to worldsheet-gravity was given in \cite{Hoogeveen:2007tu}. Although
all these descriptions gave important insight and might eventually
give a complete understanding, the picture so far remains a bit unsatisfactory.

At classical level, one natural way to remove constraints is to use
a projection. To be more precise, replace in the action the constrained
variables by projections of unconstrained ones and consider the result
to be the new action of the unconstrained variables. A priori it is
not clear whether this can solve some of the problems, since the free
action would become non-free and the constraint would be replaced
by another gauge symmetry. However, the latter might open the way
to switch to different formulations by choosing different gauges.
This was one of the motivations to look for a projection.

A linear projector has already appeared in the pure spinor string
literature, in particular within the so-called Y-formalism in \cite{Oda:2001zm}
and recently also in a covariant version in \cite{Berkovits:2010zz}.
This projector is extracting the gauge invariant part of the antighost
field. The transpose of this projector maps generic spinor variations
to variations which are consistent with the pure spinor constraint.
Therefore a projection of a general spinor to a pure spinor corresponds
to the integration of this linearized projection. In this article
we will present the full non-linear projection. 

The article is organized as follows. In section \ref{sec:nonlin}
we introduce in equation (\ref{ProjGeneral}) a family of nonlinear
projection maps from the space of SO(10) Weyl spinors to SO(10) pure
Weyl spinors and study some of its properties. Also the viewpoint
of the projection being part of a variable transformation $\rho^{\alpha}\mapsto(\lambda^{\alpha},\auxii^{a})$
is presented on page \pageref{Proj-inv} and the non-covariant version
of the projection using a reference spinor is discussed on page \pageref{Pnoncov}.
In section \ref{sec:linearized} on page \pageref{sec:linearized}
we calculate the Jacobian matrix of the projection which provides
the push-forward map for the vectors of the tangent spaces of the
previously mentioned spaces. In particular, the variations of spinors
are mapped with this linearized map. It is shown that on the constraint
surface it reduces to a linear projector whose transpose is known
to extract the gauge invariant part from the antighost of pure spinor
string theory. In section \ref{sec:projectorsOnSurface} on page \pageref{sec:projectorsOnSurface}
we collect some properties of this and a few other projection matrices
on the constraint surface, mainly to have them available as a toolbox
for the subsequent section. In section \ref{sec:hermitean} on page
\pageref{sec:hermitean} a case is discussed where the Jacobian matrix
is Hermitian which leads to several appealing properties. In particular
it turns out that the projection then can be derived from a potential.
In section \ref{sec:U5} on page \pageref{sec:U5} we discuss the
non-covariant version of the projection map for a particular reference
spinor within the U(5)-covariant parametrization of SO(10) spinors.

In section \ref{sec:ghost-action} on page \pageref{sec:ghost-action}
we finally apply the mathematical insight to the pure spinor string.
In subsection \ref{sec:ghost-action-1} we review the non-minimal
pure spinor string and in particular the projection of the antighost
field $\omega_{z\alpha}$ to its gauge invariant part. In addition
we also introduce gauge invariant projections for the non-minimal
fields $\bar{\omega}_{z}^{\alpha}$ and $\bs s_{z}^{\alpha}$ which
had not yet been presented in the literature to our knowledge. In
subsection \ref{sec:ghost-action-2} on page \pageref{sec:ghost-action-2}
we replace the pure spinor $\lambda^{\alpha}$ by the projection image
of an unconstrained spinor $\rho^{\alpha}$. The resulting constraint-free
action is not very appealing by itself, but it is conceptionally interesting
as it comes with an additional gauge symmetry which allows to look
for different gauges than the pure spinor constraint. And in subsection
\ref{sec:ghost-action-3} on page \pageref{sec:ghost-action-3} we
quickly review the form of the action in the U(5) formalism in order
to see how the resulting gauge invariant antighost-combination corresponds
to the previously discussed projection. 

In section \ref{sec:integration-measure} on page \pageref{sec:integration-measure}
finally we regard the pure spinor space as being embedded in $\mathbb{C}^{16}$
and calculate the transformation of the holomorphic volume form of
this ambient space under the aforementioned variable transformation
$\rho^{\alpha}\mapsto(\lambda^{\alpha},\auxii^{a})$. We discuss the
relation of the result to the pure spinor holomorphic volume form
known from the literature. 

The appendix \ref{app:toy-model} on page \pageref{app:toy-model}
contains the discussion of a toy model that served as a guide line
to derive our projection map. Appendices \ref{app:proof1}-\ref{app:proof3}
starting from page \pageref{app:proof1} contain the detailed proofs
of three propositions of the main part.

\section{Nonlinear projection to the pure spinor space}

\label{sec:nonlin}An SO(10) or SO(1,9) Weyl spinor with complex components
$\lambda^{\alpha}$ and ${\scriptstyle \alpha\in\{1,\ldots,16\}}$
is called a pure spinor%
\footnote{In fact, in general dimensions $d$ a pure spinor in even dimensions
is defined to be annihilated by a maximally isotropic subspace of
the Clifford vector space\frem{correct name?} spanned by the Dirac
gamma matrices. So roughly speaking it is annihilated by {}``half
of the gamma matrices'', meaning by $d/2$ linear combinations of
gamma matrices. If one thinks of pure spinors as being a Clifford
vacuum, then these $d/2$ generators are just the annihilators while
the remaining ones will be the generators in the Clifford representation
of spinors. In other words pure spinors provide possible vacua for
a Clifford representation. It is well known that in 10 dimensions
the above definition of pure spinors is equivalent to the quadratic
constraint (\ref{ps-constraint}).\frem{reference}$\quad\fussend$ %
} if it obeys the quadratic constraints 
\begin{equation}
\lambda^{\alpha}\gamma_{\alpha\beta}^{c}\lambda^{\beta}=0\label{ps-constraint}
\end{equation}
where the Latin index ${\scriptstyle c\in\{1,\ldots,10\}}$ for SO(10)
or ${\scriptstyle c\in\{0,1\ldots,9\}}$ for $SO(1,9)$. The matrices
$\gamma_{\alpha\beta}^{c}$ are the off-diagonal chiral blocks of
the SO(10) or SO(1,9) Dirac gamma matrices in the standard Weyl representation%
\footnote{\label{fn:10dgamma}In the standard Weyl-representation of SO(10)
spinors one has 
\begin{eqnarray*}
\Gamma^{a\,\q{\alpha}}\tief{\q{\beta}} & = & \left(\begin{array}{cc}
0 & \gamma^{a\,\alpha\beta}\\
\gamma_{\alpha\beta}^{a} & 0
\end{array}\right),\quad\q{\alpha}\in\{1,\ldots,32\},\quad\alpha\in\{1,\ldots,16\},\quad a\in\{1,\ldots,10\}
\end{eqnarray*}
with numerically $\gamma_{\alpha\beta}^{a}=\gamma^{a\,\alpha\beta}$
real and symmetric for $a\in\{1,\ldots,9\}$ and $\gamma^{10\,\alpha\beta}=-\gamma_{\alpha\beta}^{10}=i\delta_{\alpha\beta}$
symmetric and imaginary. The index $a$ is pulled with the SO(10)
metric $\delta_{ab}$. If instead we think of SO(1,9) spinors, we
have to use 
\[
\Gamma^{0}\equiv-i\Gamma^{10}=\left(\begin{array}{cc}
0 & \one\\
-\one & 0
\end{array}\right)
\]
Then the index is pulled with $\eta_{ab}=\diag(-1,1,\ldots,1)$. In
any case (so for either all $a\in\{1,\ldots,10\}$ or $a\in\{0,\ldots,9\}$)
we thus have 
\begin{eqnarray*}
(\gamma_{\alpha\beta}^{a})^{*} & = & \gamma_{a}^{\alpha\beta}=\gamma_{a}^{\beta\alpha}
\end{eqnarray*}
So the notation is such that the $\gamma_{\alpha\beta}^{a}$ behave
in contractions as if they were all real (for both, SO(1,9) and SO(10)).
In particular we have 
\[
(\rho\gamma^{a}\rho)^{*}=(\bar{\rho}\gamma_{a}\bar{\rho})\quad\fussend
\]
\frem{Reference for the above representation: \cite[appendix]{Nh:2003cm},\cite[p.175ff]{Guttenberg:2008ic}?}%
}. The space of pure spinors is known to be a $\mathbb{C}^{*}$-fibration
of $SO(10)/U(5)$. Thus pure spinors are a non-linear representation
of SO(10). Therefore any projection to this space will necessarily
be non-linear. In the following proposition we will observe that 
\begin{equation}
P^{\alpha}(\rho,\bar{\rho})\equiv\rho^{\alpha}-\tfrac{1}{2}\frac{(\rho\gamma^{a}\rho)(\bar{\rho}\gamma_{a})^{\alpha}}{(\rho\bar{\rho})+\sqrt{(\rho\bar{\rho})^{2}-\frac{1}{2}(\rho\gamma^{b}\rho)(\bar{\rho}\gamma_{b}\bar{\rho})}}
\end{equation}
 projects a general Weyl spinor to a pure spinor in an SO(10) covariant
way. This fact will certainly not change if the projection is multiplied
by some scalar function $f$ which is 1 on the constraint surface.
As this function will play an important role later on for tuning the
properties of the Jacobian matrices, we will immediately include it
into the discussion and promote $P^{\alpha}$ to a family of projections
$P_{(f)}^{\alpha}$. In addition we will introduce some auxiliary
variables $\auxi$ and $\auxii^{a}$ in \eqref{zxDefandfzero} which
will simplify the expressions and save some space.

\begin{prop}[Covariant projection to the pure spinor space]\label{prop:covproj}Consider
an SO(10) Weyl spinor $\rho^{\alpha}$ with 16 complex components
and the family of $\mbox{maps }P_{(f)}$ 
\[
(\rho^{\alpha},\bar{\rho}_{\beta})\mapsto\left(P_{(f)}^{\alpha}(\rho,\bar{\rho}),\bar{P}_{(f)\beta}(\rho,\bar{\rho})\right)
\]
acting on Weyl spinors via%
\footnote{\label{fn:generality}About the generality of (\ref{ProjGeneral}):
The projection (\ref{ProjGeneral}) is covariant and homogeneous of
degree 1 for every function $f.$ However, it is certainly not the
most general covariant projection with this property. In a more general
ansatz one could imagine terms like $\frac{(\rho\gamma^{[5]}\rho)(\bar{\rho}\gamma_{[5]})^{\alpha}}{(\rho\bar{\rho})}$
(where $\gamma^{[5]}$ is shorthand for $\gamma^{a_{1}\ldots a_{5}}$)
multiplied with a coefficient that depends not only on $\auxi$, but
also on other scale-invariant variables like $\frac{(\rho\gamma^{[5]}\rho)(\bar{\rho}\gamma_{[5]}\bar{\rho})}{(\rho\bar{\rho})^{2}}$
and probably many others. However, there is little motivation to make
such a far more complicated ansatz. It might be interesting only if
one wants to achieve that the projection behaves like the identity
in quadratic contractions with $\gamma^{[5]}$, i.e. $(P(\rho,\bar{\rho})\gamma^{[5]}P(\rho,\bar{\rho}))=(\rho\gamma^{[5]}\rho).\qquad\fussend$ %
}
\begin{eqnarray}
P_{(f)}^{\alpha}(\rho,\bar{\rho}) & \equiv & f\left(\auxi\right)\Bigl(\rho^{\alpha}-\tfrac{1}{2}\frac{\auxii^{a}(\bar{\rho}\gamma_{a})^{\alpha}}{1+\sqrt{1-\auxi}}\Bigr)\label{ProjGeneral}\\
\mbox{with }f(0) & = & 1,\quad\auxii^{a}\equiv\frac{(\rho\gamma^{a}\rho)}{(\rho\bar{\rho})},\quad\auxi\equiv\tfrac{1}{2}\auxii^{a}\bar{\auxii}_{a}\label{zxDefandfzero}
\end{eqnarray}
 with $f$ any complex-valued function defined on the interval $[0,1]$
and obeying $f(0)=1$ and $\gamma_{\alpha\beta}^{a}$ and $\gamma_{a}^{\alpha\beta}$
being the chiral blocks of the Dirac-$\Gamma$-matrices in the 10d
standard Weyl-representation of footnote \ref{fn:10dgamma}. Then
the following statements hold:
\begin{enumerate}
\item $\forall f,\: P_{(f)}$ is a \textbf{projection map} from the space
of Weyl spinors to the space of pure Weyl spinors, i.e. it obeys the
following two properties 
\begin{eqnarray}
P_{(f)}^{\alpha}(\rho,\bar{\rho})\gamma_{\alpha\beta}^{a}P_{(f)}^{\beta}(\rho,\bar{\rho}) & = & 0\quad\forall\rho^{\alpha}\label{proj-prop1}\\
P_{(f)}^{\alpha}(\lambda,\bar{\lambda}) & = & \lambda^{\alpha}\quad\forall\lambda^{\alpha}\mbox{ with }\lambda^{\alpha}\gamma_{\alpha\beta}^{a}\lambda^{\beta}=0\qquad\label{proj-prop2}
\end{eqnarray}
The first one just states that the image of a general Weyl spinor
is a pure Weyl spinor, while the second property is the projection
property. It implies idempotency of the map $P_{(f)}$ (i.e. $P_{(f)}\circ P_{(f)}=P_{(f)}$),
but it is slightly stronger as it also guarantees surjectivity onto
the space of pure Weyl spinors. 
\item $P_{(f)}$ is homogeneous of degree (1,0) for any $f$
\begin{equation}
P_{(f)}^{\alpha}(c\rho,\bar{c}\bar{\rho})=cP_{(f)}^{\alpha}(\rho,\bar{\rho})\quad\forall c\in\mathbb{C}\label{homogenous}
\end{equation}
and its modulus square is given by 
\begin{equation}
P_{(f)}^{\alpha}(\rho,\bar{\rho})\bar{P}_{(f)\alpha}(\rho,\bar{\rho})=2(\rho\bar{\rho})\abs{f\left(\auxi\right)}^{2}\left(\frac{1-\auxi}{1+\sqrt{1-\auxi}}\right)\label{Proj-modulus}
\end{equation}
Although $\auxi$ and $\auxii^{a}$ are not well-defined at the origin
$\abs{\rho}=0$, the projection map is still well-defined there in
the sense that the limit exists if $f$ is continuous on $[0,1]$:
\begin{equation}
P_{(f)}^{\alpha}(0,0)\equiv\lim_{\abs{\rho}\to0}P_{(f)}^{\alpha}(\rho,\bar{\rho})=0\quad(\mbox{for cont. }f)\label{limitToZero}
\end{equation}

\item The zero-locus of the projection is given by 
\begin{equation}
P_{(f)}^{-1}(0)=\left\{ 0\right\} \cup\left\{ \rho^{\alpha}|\,\rho^{\alpha}=\tfrac{1}{2}\auxii^{b}(\gamma_{b}\bar{\rho})^{\alpha}\right\} \cup\left\{ \rho^{\alpha}|f(\auxi)=0\right\} \label{kerP}
\end{equation}
In particular real vectors $\bar{\rho}_{\alpha}=\rho^{\alpha}$ are
in the zero-locus. As this is not an $SO(10)$ invariant statement,
also all vectors obeying an $SO(10)$ rotated version of this reality
condition lie in the zero-locus.
\item The projection $P_{(f)}$ is continuous everywhere if $f$ is continuous,
and it is differentiable everywhere but at the zero-locus subset $\left\{ 0\right\} \cup\left\{ \rho^{\alpha}|\,\auxi=1\right\} $,
if $f$ is differentiable. One can choose $f$ such that it will even
become differentiable at $\auxi=1$, namely for $f(\auxi)=\tilde{f}(\auxi)\left(1-\auxi\right)^{1+r}$
with a differentiable $\tilde{f}$ with $\tilde{f}(0)=1$ and $r\geq0$.
\item The auxiliary variables $\auxi,\auxii^{a}$ obey
\begin{eqnarray}
\hspace{-1cm}\auxi & \in & [0,1]\label{xin01}\\
\hspace{-1cm}\auxi=1 & \!\!\!\!\iff\!\!\!\! & 0\neq\rho^{\alpha}=\alpha^{b}(\gamma_{b}\bar{\rho})^{\alpha}{\scriptstyle \, for\, some}\,\alpha^{b}\in\mathbb{C}\iff0\neq\rho^{\alpha}=\tfrac{1}{2}\auxii^{b}(\gamma_{b}\bar{\rho})^{\alpha}\qquad\quad\label{xoneifreal}\\
\hspace{-1cm}\auxii^{a}\auxii_{a} & = & 0,\quad\auxii^{a}(\gamma_{a}\rho)_{\alpha}=0\label{zsquarezero}\\
\hspace{-1cm}\bei{\auxi}{\rho=\lambda} & = & 0=\bei{\auxii^{a}}{\rho=\lambda}\quad\forall{\scriptstyle pure}\:\lambda^{\alpha}\neq0\label{xzvanishforpure}\\
\hspace{-1cm}\bei{\auxi}{c\rho,\bar{c}\bar{\rho}} & = & \bei{\auxi}{\rho,\bar{\rho}}\quad,\quad\bei{\auxii^{a}}{c\rho,\bar{c}\bar{\rho}}=\tfrac{c}{\stackrel{}{\bar{c}}}\bei{\auxii^{a}}{\rho,\bar{\rho}}\quad\forall c\in\mathbb{C}\label{xzscaling}
\end{eqnarray}

\end{enumerate}
\end{prop}

The proof of this proposition is given in appendix \ref{app:proof1}
on page \pageref{app:proof1}. Most of it contains only straightforward
calculations. Only the proof of $\auxi\leq1$ turns out to be quite
tricky.

\paragraph{Projection as part of a reparametrization}

We can regard the projection as part of a variable transformation
\begin{eqnarray}
(\rho^{\alpha},\bar{\rho}_{\alpha}) & \mapsto & (\lambda^{\alpha},\auxii^{a},\bar{\lambda}_{\alpha},\bar{\auxii}_{a})\label{var-trafo}\\
 &  & \mbox{with }\lambda^{\alpha}\equiv P_{(f)}^{\alpha}(\rho,\bar{\rho}),\quad\auxii^{a}\equiv\tfrac{(\rho\gamma^{a}\rho)}{(\rho\bar{\rho})}\nonumber 
\end{eqnarray}
where the variables on the righthand side are constrained by%
\footnote{Proof of $\auxii^{a}\left(\gamma_{a}\lambda\right)_{\alpha}=0$ in
(\ref{ProjZsquareZero}):\vspace{-.5cm} 
\[
\auxii^{a}\left(\gamma_{a}\lambda\right)_{\alpha}\equiv\auxii^{a}\left(\gamma_{a}P_{(f)}(\rho,\bar{\rho})\right)_{\alpha}\stackrel{\eqref{ProjGeneral}}{=}f\left(\auxi\right)\Bigl(\auxii^{a}(\gamma_{a}\rho)_{\alpha}-\tfrac{1}{2}\frac{\overbrace{\auxii^{a}\auxii^{c}(\gamma_{a}\gamma_{c}}^{\auxii^{a}\auxii_{a}}\bar{\rho})^{\alpha}}{1+\sqrt{1-\auxi}}\Bigr)\stackrel{\eqref{zsquarezero}}{=}0\qquad\fussend
\]
} 
\begin{equation}
(\lambda\gamma^{a}\lambda)=0=\auxii^{a}(\lambda\gamma_{a})_{\alpha}=\auxii^{a}\auxii_{a}\quad,\quad\tfrac{1}{2}\auxii^{a}\bar{\auxii}_{a}\leq1\label{ProjZsquareZero}
\end{equation}
One can see that they have effectively the same number of degrees
of freedom by observing that the above reparametrization is invertible
\begin{equation}
\rho^{\alpha}=\frac{1+\sqrt{1-\frac{1}{2}\auxii^{a}\bar{\auxii}_{a}}}{2f(\frac{1}{2}\auxii^{a}\bar{\auxii}_{a})\sqrt{1-\frac{1}{2}\auxii^{a}\bar{\auxii}_{a}}}\lambda^{\alpha}+\frac{1}{4\bar{f}(\frac{1}{2}\auxii^{a}\bar{\auxii}_{a})\sqrt{1-\frac{1}{2}\auxii^{a}\bar{\auxii}_{a}}}\auxii^{a}\left(\bar{\lambda}\gamma_{a}\right)^{\alpha}\label{Proj-inv}
\end{equation}
We have explicitly replaced here $\auxi$ by $\frac{1}{2}\auxii^{a}\bar{\auxii}_{a}$
in order to stress that we have a function of only $\auxii^{a}$ and
$\lambda^{\alpha}$ on the righthand side. Instead for the other direction
in (\ref{var-trafo}) $\lambda^{\alpha}$ and $\auxii^{a}$ have to
be seen as functions of $\rho$ only. In particular all the appearances
of $\auxi$ and $\auxii^{a}$ in the projection $P_{(f)}^{\alpha}(\rho,\bar{\rho})$
are mere placeholders for the $\rho$-expressions given in (\ref{zxDefandfzero}).
The validity of (\ref{Proj-inv}) can easily be checked by plugging
the explicit expression for the projection $P_{(f)}^{\alpha}(\rho,\bar{\rho})$
for $\lambda^{\alpha}$ and the same for its complex conjugate.%
\footnote{In order to prove (\ref{Proj-inv}) we start from the righthand side:
\begin{eqnarray*}
 &  & \frac{1+\sqrt{1-\auxi}}{2f(\auxi)\sqrt{1-\auxi}}P^{\alpha}(\rho,\bar{\rho})+\frac{1}{4\bar{f}(\auxi)\sqrt{1-\auxi}}\auxii^{a}\left(\bar{P}(\rho,\bar{\rho})\gamma_{a}\right)^{\alpha}=\\
 & \stackrel{\mbox{\tiny\eqref{ProjGeneral}}}{=} & \frac{1+\sqrt{1-\auxi}}{2f(\auxi)\sqrt{1-\auxi}}f\left(\auxi\right)\!\left(\!\rho^{\alpha}-\tfrac{1}{2}\frac{\auxii^{a}(\bar{\rho}\gamma_{a})^{\alpha}}{1+\sqrt{1-\auxi}}\!\right)\!+\!\frac{1}{4\bar{f}(\auxi)\sqrt{1-\auxi}}\auxii^{a}\bar{f}\left(\auxi\right)\biggl(\!(\bar{\rho}\gamma_{a})^{\alpha}\!-\!\tfrac{1}{2}\frac{\bar{\auxii}^{b}(\rho\overbrace{\gamma_{b}\gamma_{a}}^{-\gamma_{a}\gamma_{b}\lqn{{\scriptstyle +2\delta_{ab}}}})^{\alpha}}{1+\sqrt{1-\auxi}}\!\biggr)\!=\\
 & = & \frac{1}{2\sqrt{1-\auxi}}\left(1+\sqrt{1-\auxi}-\frac{\auxi}{1+\sqrt{1-\auxi}}\right)\rho^{\alpha}=\\
 & = & \rho^{\alpha}\qquad\fussend
\end{eqnarray*}
} Note that (\ref{Proj-inv}) is singular at $\auxi=1$ (if not cured
by an appropriately chosen $f(\auxi)$) and at the zeros of $f$. 

So to summarize, every Weyl spinor $\rho^{\alpha}$ can be parametrized
by a pure spinor $\lambda^{\alpha}$ and a constrained vector $\auxii^{a}$.
It turns out to be very useful in some calculations to use (\ref{Proj-inv})
and write $\rho^{\alpha}$ as a linear combination of $\lambda^{\alpha}=P_{(f)}^{\alpha}(\rho,\bar{\rho})$
and $\bar{\lambda}_{\alpha}$. 

We can use (\ref{Proj-modulus}) to quickly determine the inverse
transformation of the modulus 
\begin{equation}
(\rho\bar{\rho})=\frac{\left(1+\sqrt{1-\auxi}\right)}{2\abs{f(\auxi)}^{2}\left(1-\auxi\right)}(\lambda\bar{\lambda})\label{Proj-modulus-repeated}
\end{equation}
Let us also finally provide two more useful contractions
\begin{eqnarray}
(\lambda\bar{\rho})\equiv P_{(f)}^{\alpha}(\rho,\bar{\rho})\bar{\rho}_{\alpha} & \stackrel{\eqref{ProjGeneral}}{=} & (\rho\bar{\rho})f(\auxi)\sqrt{1-\auxi}\label{PbarRho}\\
(\lambda\gamma^{c}\rho)\equiv(P_{(f)}(\rho,\bar{\rho})\gamma^{c}\rho) & \ous{\eqref{ProjGeneral}}={\eqref{zsquarezero}} & (\rho\bar{\rho})f(\auxi)\frac{\auxii^{c}\sqrt{1-\auxi}}{1+\sqrt{1-\auxi}}\label{PgammaRho}
\end{eqnarray}

\paragraph{Alternative reparametrization}

Rewriting \eqref{Proj-inv} in the form \\
$\rho^{\alpha}=\frac{1+\sqrt{1-\auxi}}{2f(\auxi)\sqrt{1-\auxi}}\Bigl(\lambda^{\alpha}+\tfrac{1}{2}\frac{f(\auxi)\auxii^{a}}{\bar{f}(\auxi)\left(1+\sqrt{1-\auxi}\right)}\left(\bar{\lambda}\gamma_{a}\right)^{\alpha}\Bigr)$
shows that this inverse variable transformation becomes particularly
simple if one chooses 
\begin{equation}
\Auxii^{a}\equiv\frac{f(\auxi)\auxii^{a}}{\bar{f}(\auxi)\left(1+\sqrt{1-\auxi}\right)}=\tfrac{f(\auxi)}{\os{}{\bar{f}}(\auxi)}\frac{(\rho\gamma^{a}\rho)}{(\rho\bar{\rho})+\sqrt{(\rho\bar{\rho})^{2}-\frac{1}{2}(\rho\gamma^{a}\rho)(\bar{\rho}\gamma_{a}\bar{\rho})}}\label{newz}
\end{equation}
as new variable, so that the inverse transformation up to an overall
prefactor is of the form $\rho^{\alpha}\propto\lambda^{\alpha}+\tfrac{1}{2}\Auxii^{a}\left(\bar{\lambda}\gamma_{a}\right)^{\alpha}$.
Note that, of course, one can also absorb the factor $\frac{1}{2}$
into the definition of $\Auxii^{a}$. However, the way we have defined
it here will guarantee that $\Auxi$, which we are going to introduce
now, lies like $\auxi$ in the interval $[0,1]$. In order to determine
also the overall prefactor after reparametrization, let us calculate
a few more relations between old and new variables. The absolute value
squares are related via
\begin{eqnarray}
\Auxi\equiv\tfrac{1}{2}\Auxii^{a}\bar{\Auxii}_{a} & = & \frac{\auxi}{\left(1+\sqrt{1-\auxi}\right)^{2}}=\frac{1-\sqrt{1-\auxi}}{1+\sqrt{1-\auxi}}\label{newx}
\end{eqnarray}
which implies the relations $(1-\Auxi)=\frac{2\sqrt{1-\auxi}}{1+\sqrt{1-\auxi}}$,
$(1+\Auxi)=\frac{2}{1+\sqrt{1-\auxi}}$ and thus in particular 
\begin{equation}
\sqrt{1-\auxi}=\frac{1-\Auxi}{1+\Auxi}\quad,\quad1-\sqrt{1-\auxi}=\frac{2\Auxi}{1+\Auxi}\quad,\quad1+\sqrt{1-\auxi}=\frac{2}{1+\Auxi}
\end{equation}
The last can be used to invert the relation (\ref{newx}) and obtain
\begin{equation}
\auxi=\frac{4\Auxi}{(1+\Auxi)^{2}}
\end{equation}
Because of $0\leq\auxi\leq1$ also the new variable lives in this
interval
\begin{equation}
0\leq\Auxi\leq1
\end{equation}
and also $\Auxii^{a}$ and its absolute value vanish on the constraint
surface
\begin{equation}
(\rho\gamma^{a}\rho)=0\:\forall a\quad\stackrel{{\rm for\,}\rho^{\alpha}\neq(0,\ldots,0)}{\iff}\quad\Auxii^{a}=0\:\forall a\quad(\Auxi=0)
\end{equation}
Finally we can use the above relations to also write down the inverse
of (\ref{newz})
\begin{equation}
\auxii^{a}=\frac{\bar{\tilde{f}}(\Auxi)}{\tilde{f}(\Auxi)}\frac{2\Auxii^{a}}{(1+\Auxi)}
\end{equation}
with%
\footnote{\label{fn:fprime}For the derivatives this implies
\[
f'(\auxi)=\partiell{\Auxi}{\auxi}\tilde{f}'(\Auxi)=\frac{1}{\sqrt{1-\auxi}\left(1+\sqrt{1-\auxi}\right)^{2}}\tilde{f}'(\Auxi)
\]
I.e. 
\[
\tilde{f}'(\Auxi)=\sqrt{1-\auxi}\left(1+\sqrt{1-\auxi}\right)^{2}f'(\auxi)
\]
or 
\[
f'(\auxi)=\frac{(1+\Auxi)^{3}}{4(1-\Auxi)}\tilde{f}'(\Auxi)\qquad\fussend
\]
} 
\begin{equation}
\tilde{f}(\Auxi)\equiv f(\auxi)
\end{equation}
The Weyl spinor $\rho^{\alpha}$ expressed in terms of the pure spinor
$\lambda^{\alpha}$ and the variable $\auxii^{a}$ as in (\ref{Proj-inv})
can now be written as
\begin{equation}
\rho^{\alpha}=\frac{1}{\tilde{f}(\Auxi)(1-\Auxi)}\left(\lambda^{\alpha}+\tfrac{1}{2}\Auxii^{a}\left(\bar{\lambda}\gamma_{a}\right)^{\alpha}\right)\label{Proj-inv-alt}
\end{equation}
The absolute value squared (\ref{Proj-modulus-repeated}) turns into
\begin{equation}
(\rho\bar{\rho})=\frac{1+\Auxi}{\abs{\tilde{f}(\Auxi)}^{2}(1-\Auxi)^{2}}(\lambda\bar{\lambda})
\end{equation}
Apparently the inverse transformation (\ref{Proj-inv-alt}) becomes
most simple for the special choice $f=h$ where 
\begin{equation}
h(\auxi)\equiv\tilde{h}(\Auxi)\equiv\tfrac{1}{(1-\Auxi)}=\frac{1+\sqrt{1-\auxi}}{2\sqrt{1-\auxi}}\label{hdef-first}
\end{equation}
It is interesting that we will rediscover this function a bit later
in the context of Hermiticity and thus gave it the name $h$ here.
So we simply have 
\begin{equation}
\rho^{\alpha}=\lambda^{\alpha}+\tfrac{1}{2}\Auxii^{a}\left(\bar{\lambda}\gamma_{a}\right)^{\alpha}\quad\mbox{for }f=h.\label{Proj-inv-alt-h}
\end{equation}

\paragraph{Reference spinor }

The projector (\ref{ProjGeneral}) with (\ref{zxDefandfzero}) is
globally well defined, but it is non-holomorphic in $\rho$. This
can be changed, if one simply replaces $\bar{\rho}_{\alpha}$ by some
reference spinor $\bar{\chi}_{\alpha}$ which is not related to $\rho^{\alpha}$
by complex conjugation or in any other way. Almost everything still
works in the same way as before, but it is useful to change the notation
slightly, in order not to get confused:
\begin{eqnarray}
\auxii^{a} & \equiv & \frac{\rho\gamma^{a}\rho}{(\rho\bar{\chi})},\quad\bar{\auxiii}_{a}\equiv\frac{\bar{\chi}\gamma_{a}\bar{\chi}}{(\rho\bar{\chi})},\quad\auxi\equiv\tfrac{1}{2}\auxii^{a}\bar{\auxiii}_{a}\label{xz-noncov}\\
P_{(f,\bar{\chi})}^{\alpha}(\rho) & \equiv & f\left(\auxi\right)\left(\rho^{\alpha}-\tfrac{1}{2}\frac{\auxii^{a}(\bar{\chi}\gamma_{a})^{\alpha}}{1+\sqrt{1-\auxi}}\right)\label{Pnoncov}
\end{eqnarray}
One immediate consequence is that $\auxi$ is not in general real
any longer and we have no guarantee that $\abs{\auxi}\leq1$. If its
absolute value exceeds 1, however, the projection map becomes non-continuous
at the branch-cut of the square root. So one should better restrict
manually to $\abs{\auxi}\leq1$ which means that one cannot use one
global projection map, but needs different projection maps for different
neighbourhoods. 

With the reference spinor $\bar{\chi}_{\alpha}$ being independent
of $\rho^{\alpha}$ , it can itself be chosen to be a pure spinor,
which implies 
\begin{equation}
\auxi=\bar{\auxiii}_{a}=0\quad\mbox{(if }\mbox{\ensuremath{\bar{\chi}}}_{\alpha}\mbox{ is pure)}
\end{equation}
The projection then becomes independent of the function $f$ and reduces
to
\begin{equation}
P_{(\bar{\chi})}^{\alpha}(\rho)\equiv P_{(f,\bar{\chi})}^{\alpha}(\rho)=\rho^{\alpha}-\tfrac{1}{4}\frac{(\rho\gamma^{a}\rho)(\bar{\chi}\gamma_{a})^{\alpha}}{(\rho\bar{\chi})}\quad\mbox{(if }\mbox{\ensuremath{\bar{\chi}}}_{\alpha}\mbox{ is pure)}\label{non-cov-integrated}
\end{equation}
In the subsequent section \ref{sec:linearized} we will discuss the
variation of the general non-linear projection map \eqref{ProjGeneral}
which leads to projection matrices. For the above simplified case
of a projection map with a reference \emph{pure} spinor (\ref{non-cov-integrated})
it is already recognizable that this will yield the non-covariant
projection-matrix defined by Oda and Tonin in \cite[eq.(17)]{Oda:2001zm}
and used by them in \cite[eq.(6)]{Oda:2004bg} in order to extract
the gauge-invariant part of the antighost. We will come back to this
in the remark on page \pageref{rmk:projection-matrix}. \rembreak\rem{Observation:
If one chooses in (\ref{Pnoncov}) $\rho$ and $\bar{\chi}$ both
to be real, the result would be (if $\auxi\leq1$) a real pure spinor
which should not exist, at least if non-vanishing. This seems to suggest
that choosing $\rho$ and $\bar{\chi}=\chi$ both to be real, it forces
$\auxi\geq1$ with $\auxi=1$ only if $\chi=\rho$. Let us try to
verify this guess
\begin{eqnarray}
\auxi & = & \tfrac{1}{2}\auxii^{a}\bar{\mf w}_{a}=\tfrac{1}{2}\frac{(\rho^{\alpha}\gamma_{\alpha\beta}^{a}\rho^{\beta})(\bar{\chi}_{\gamma}\gamma_{a}^{\gamma\delta}\bar{\chi}_{\delta})}{(\rho\bar{\chi})^{2}}
\end{eqnarray}
Using $\bar{\chi}_{\gamma}=\chi^{\gamma}$ and $\gamma_{10}^{\gamma\delta}=-\gamma_{10\,\gamma\delta}$,
we can rewrite this as 
\begin{equation}
\auxi\stackrel{\rho,\chi{\rm \, real}}{=}\tfrac{1}{2}\frac{(\rho^{\alpha}\gamma_{\alpha\beta}^{a}\rho^{\beta})(\chi^{\gamma}\gamma_{a\,\gamma\delta}\chi^{\delta})}{(\rho\chi)^{2}}-\frac{(\rho^{\alpha}\gamma_{\alpha\beta}^{10}\rho^{\beta})(\chi^{\gamma}\gamma_{10\,\gamma\delta}\chi^{\delta})}{(\rho\chi)^{2}}
\end{equation}
Now we can use the Fierz identity for the first term and $\gamma_{\alpha\beta}^{10}=-i\delta_{\alpha\beta}$
(from footnote \ref{fn:10dgamma}) for the second term to arrive at
\begin{equation}
\auxi\stackrel{\rho,\chi{\rm \, real}}{=}-\frac{(\rho^{\alpha}\gamma_{\alpha\beta}^{a}\chi^{\beta})(\rho^{\gamma}\gamma_{a\,\gamma\delta}\chi^{\delta})}{(\rho\chi)^{2}}+\frac{\rho^{2}\chi^{2}}{(\rho\chi)^{2}}
\end{equation}
Let's split the sum over $a$ and use $\gamma_{\alpha\beta}^{10}=-i\delta_{\alpha\beta}$
once more
\begin{equation}
\auxi\stackrel{\rho,\chi{\rm \, real}}{=}1-\frac{\sum_{i=1}^{9}(\rho^{\alpha}\gamma_{\alpha\beta}^{i}\chi^{\beta})(\rho^{\gamma}\gamma_{i\,\gamma\delta}\chi^{\delta})}{(\rho\chi)^{2}}+\frac{\rho^{2}\chi^{2}}{(\rho\chi)^{2}}
\end{equation}
It remains to show that the last term is bigger than the second....
Note that the calculation works the same for SO(1,9) up to this point.

If proven: for real $\rho,\bar{\chi}$ (we still write $\bar{\chi}$
to stress the opposite chirality in the Minkowskian case), we can
replace the projection map (\ref{Pnoncov}) by
\begin{equation}
P_{(f,\bar{\chi})}^{\alpha}(\rho)\equiv f\left(\auxi\right)\left(\rho^{\alpha}-\tfrac{1}{2}\frac{\auxii^{a}(\bar{\chi}\gamma_{a})^{\alpha}}{1+i\sqrt{\auxi-1}}\right)
\end{equation}
This is of particular interest if we don't want to think of $\lambda$
and $\bar{\lambda}$ in the non-minimal formalism as complex conjugates
(as this does not allow an SO(1,9) interpretation of the non-minimal
formalism) but as independent complex Weyl spinors of opposite chirality
which are obtained from two real Majorana-Weyl spinors $\rho$ and
$\bar{\rho}$ via the projection above and a corresponding one with
interchanged role of $\rho$ and $\bar{\rho}$ (and as a matter of
taste $i\to-i$?) }\newpage

\section{Linearized}

\label{sec:linearized}

As we will quite frequently use the image of the projection $P_{(f)}^{\alpha}$
as an argument of another function, let us denote for simplicity like
previously\enlargethispage*{1cm} 
\begin{equation}
\lambda^{\alpha}\equiv P_{(f)}^{\alpha}(\rho,\bar{\rho})
\end{equation}
Variations of $\rho$ live in the tangent space and are mapped via
the push-forward map, i.e. via the Jacobian-matrix defined by the
derivatives of $P_{(f)}^{\alpha}$: 

\begin{prop}\label{prop:linProj}The Jacobian matrices $\Pi_{(f)\bot}$
and $\pi_{(f)\bot}$, defined for differentiable%
\footnote{\label{fn:differentiable}Remember that $f$ is defined on the closed
interval $[0,1]$. With differentiability at $0$ we thus mean the
existence of only a limit from the right ($\auxi>0$) 
\[
f'(0)\equiv\lim_{\auxi\to0^{+}}\frac{f(\auxi)-f(0)}{\auxi}
\]
Similarly differentiability at 1 is understood as the existence of
only a limit from the left ($\auxi<1$) 
\[
f'(1)\equiv\lim_{\auxi\to1^{-}}\frac{f(\auxi)-f(1)}{\auxi-1}\quad\fussend
\]
} $f$ via $\delta P_{(f)}^{\alpha}(\rho,\bar{\rho})=\Pi_{(f)\bot\beta}^{\alpha}(\rho,\bar{\rho})\delta\rho^{\beta}+\rest_{(f)\bot}^{\alpha\beta}(\rho,\bar{\rho})\delta\bar{\rho}_{\beta}$
or equivalently 
\begin{equation}
\left(\!\!\zwek{\delta\lambda^{\beta}}{\delta\bar{\lambda}_{\beta}}\!\!\right)\!\equiv\!\left(\!\!\zwek{\delta P_{(f)}^{\alpha}(\rho,\bar{\rho})}{\delta\bar{P}_{(f)\alpha}(\rho,\bar{\rho})}\!\!\right)\!\equiv\!\left(\!\!\begin{array}{cc}
\Pi_{(f)\bot\beta}^{\alpha}(\rho,\bar{\rho}) & \rest_{(f)\bot}^{\alpha\beta}(\rho,\bar{\rho})\\
\bar{\rest}_{(f)\bot\alpha\beta}(\rho,\bar{\rho}) & \bar{\Pi}_{(f)\bot\alpha}\hoch{\beta}(\rho,\bar{\rho})
\end{array}\!\!\right)\!\!\left(\!\!\zwek{\delta\rho^{\beta}}{\delta\bar{\rho}_{\beta}}\!\!\right)\label{bigMatrix}
\end{equation}

\begin{enumerate}
\item ... are explicitly given by 
\begin{eqnarray}
\lqn{\Pi_{(f)\bot\beta}^{\alpha}(\rho,\bar{\rho})=\partial_{\rho^{\beta}}P_{(f)}^{\alpha}(\rho,\bar{\rho})=}\\
 & = & f\left(\auxi\right)\delta_{\beta}^{\alpha}+f'\left(\auxi\right)\frac{\rho^{\alpha}\bar{\auxii}_{b}(\rho\gamma^{b})_{\beta}}{\rho\bar{\rho}}-2\auxi f'\left(\auxi\right)\frac{\rho^{\alpha}\bar{\rho}_{\beta}}{\rho\bar{\rho}}+\nonumber \\
 &  & -\tfrac{1}{1+\sqrt{1-\auxi}}\left(f\left(\auxi\right)\delta_{b}^{a}+\biggl(\tfrac{f\left(\auxi\right)}{4\sqrt{1-\auxi}\left(1+\sqrt{1-\auxi}\right)}+\tfrac{1}{2}f'\left(\auxi\right)\biggr)\auxii^{a}\bar{\auxii}_{b}\right)\frac{(\gamma_{a}\bar{\rho})^{\alpha}(\rho\gamma^{b})_{\beta}}{\rho\bar{\rho}}+\nonumber \\
 &  & +\left(\tfrac{f\left(\auxi\right)}{2\sqrt{1-\auxi}\left(1+\sqrt{1-\auxi}\right)}+{\scriptstyle \left(1-\sqrt{1-\auxi}\right)}f'\left(\auxi\right)\right)\frac{\auxii^{a}(\gamma_{a}\bar{\rho})^{\alpha}\bar{\rho}_{\beta}}{\rho\bar{\rho}}\label{LinProjGen}\\
\lqn{\rest_{(f)\bot}^{\alpha\beta}(\rho,\bar{\rho})=\partial_{\bar{\rho}_{\beta}}P_{(f)}^{\alpha}(\rho,\bar{\rho})=}\\
 & = & -\tfrac{1}{2}\tfrac{f\left(\auxi\right)}{1+\sqrt{1-\auxi}}\auxii^{a}\gamma_{a}^{\alpha\beta}-2\auxi f'\left(\auxi\right)\frac{\rho^{\alpha}\rho^{\beta}}{\rho\bar{\rho}}+f'\left(\auxi\right)\frac{\rho^{\alpha}\auxii^{b}(\bar{\rho}\gamma_{b})^{\beta}}{\rho\bar{\rho}}+\nonumber \\
 &  & +\left(\tfrac{f\left(\auxi\right)}{2\sqrt{1-\auxi}\left(1+\sqrt{1-\auxi}\right)}+{\scriptstyle \left(1-\sqrt{1-\auxi}\right)}f'\left(\auxi\right)\right)\frac{\auxii^{a}(\gamma_{a}\bar{\rho})^{\alpha}\rho^{\beta}}{\rho\bar{\rho}}+\nonumber \\
 &  & -\tfrac{1}{2\left(1+\sqrt{1-\auxi}\right)}\left(\tfrac{f\left(\auxi\right)}{2\sqrt{1-\auxi}\left(1+\sqrt{1-\auxi}\right)}+f'\left(\auxi\right)\right)\frac{\auxii^{a}(\gamma_{a}\bar{\rho})^{\alpha}\auxii^{b}(\bar{\rho}\gamma_{b})^{\beta}}{\rho\bar{\rho}}\label{linprojgen}
\end{eqnarray}
with still $\auxii^{a}\equiv\frac{(\rho\gamma^{a}\rho)}{(\rho\bar{\rho})},\quad\auxi\equiv\frac{1}{2}\auxii^{a}\bar{\auxii}_{a}$
from (\ref{zxDefandfzero}).
\item ... or equivalently in terms of $\lambda^{\alpha}\equiv P_{(f)}^{\alpha}(\rho,\bar{\rho})$
and $\auxii^{a},\auxi$ by%
\footnote{\label{fn:Pi-intermsofnewz}In terms of the alternative parametrization
$\Auxii^{a}$ from equation (\ref{newz}), equations (\ref{LinProjGenLam})
and (\ref{linprojgenlam}) become 
\begin{eqnarray*}
\Pi_{(f)\bot\beta}^{\alpha}(\rho,\bar{\rho}) & = & \fa(\Auxi)\left(\delta_{\beta}^{\alpha}-\tfrac{(\gamma_{a}\bar{\lambda})^{\alpha}(\lambda\gamma^{a})_{\beta}}{2(\lambda\bar{\lambda})}\right)-\tfrac{1}{2}{\scriptstyle \left(\tilde{f}(\Auxi)-\left(1-\Auxi\right)\tilde{f}'(\Auxi)\right)}\tfrac{\lambda^{\alpha}\bar{\Auxii}_{c}\left(\lambda\gamma^{c}\right)_{\beta}}{(\lambda\bar{\lambda})}+\\
 &  & -\tfrac{\Auxi(1-\Auxi)\tilde{f}'(\Auxi)\lambda^{\alpha}\bar{\lambda}_{\beta}}{(\lambda\bar{\lambda})}-\tfrac{\fa(\Auxi)\bar{\Auxii}_{c}(\gamma^{c}\gamma_{b}\lambda)^{\alpha}\Auxii^{d}(\bar{\lambda}\gamma^{b}\gamma_{d})_{\beta}}{8(\lambda\bar{\lambda})}\\
\rest_{(f)\bot}^{\alpha\beta}(\rho,\bar{\rho}) & = & -\tfrac{\bar{\tilde{f}}(\Auxi)}{2}\left(\Auxii^{a}\gamma_{a}^{\alpha\beta}-\tfrac{\Auxii^{a}(\gamma_{a}\bar{\lambda})^{\alpha}\lambda^{\beta}}{(\lambda\bar{\lambda})}\right)+\Auxi\left(\bar{\fa}(\Auxi)-\tfrac{(1-\Auxi)\bar{\fa}(\Auxi)\tilde{f}'(\Auxi)}{\fa(\Auxi)}\right)\tfrac{\lambda^{\alpha}\lambda^{\beta}}{(\lambda\bar{\lambda})}+\\
 &  & +{\scriptstyle \frac{(1-\Auxi)}{2}}\tfrac{\bar{\tilde{f}}(\Auxi)\tilde{f}'(\Auxi)\lambda^{\alpha}\Auxii^{c}\left(\bar{\lambda}\gamma_{c}\right)^{\beta}}{\tilde{f}(\Auxi)(\lambda\bar{\lambda})}
\end{eqnarray*}
For example for $f=1$ one obtains
\begin{eqnarray*}
\Pi_{(1)\bot\beta}^{\alpha}(\rho,\bar{\rho}) & = & \left(\delta_{\beta}^{\alpha}-\tfrac{(\gamma_{a}\bar{\lambda})^{\alpha}(\lambda\gamma^{a})_{\beta}}{2(\lambda\bar{\lambda})}\right)-\tfrac{\lambda^{\alpha}\bar{\Auxii}_{c}\left(\lambda\gamma^{c}\right)_{\beta}}{2(\lambda\bar{\lambda})}-\tfrac{\bar{\Auxii}_{c}(\gamma^{c}\gamma_{b}\lambda)^{\alpha}\Auxii^{d}(\bar{\lambda}\gamma^{b}\gamma_{d})_{\beta}}{8(\lambda\bar{\lambda})}\\
\rest_{(1)\bot}^{\alpha\beta}(\rho,\bar{\rho}) & = & -\tfrac{1}{2}\left(\Auxii^{a}\gamma_{a}^{\alpha\beta}-\tfrac{\Auxii^{a}(\gamma_{a}\bar{\lambda})^{\alpha}\lambda^{\beta}}{(\lambda\bar{\lambda})}\right)+\tfrac{\Auxi\lambda^{\alpha}\lambda^{\beta}}{(\lambda\bar{\lambda})}\qquad\fussend
\end{eqnarray*}
} 
\begin{eqnarray}
\lqn{\Pi_{(f)\bot\beta}^{\alpha}(\rho,\bar{\rho})=}\nonumber \\
 & = & f(\auxi)\left(\delta_{\beta}^{\alpha}-\tfrac{1}{2}\frac{(\gamma_{a}\bar{\lambda})^{\alpha}(\lambda\gamma^{a})_{\beta}}{(\lambda\bar{\lambda})}\right)+\nonumber \\
 &  & -\left(\tfrac{\bar{f}(\auxi)}{2\left(1+\sqrt{1-\auxi}\right)}-\tfrac{\bar{f}(\auxi)f'\left(\auxi\right)\left(1-\auxi\right)}{f(\auxi)}\right)\frac{\lambda^{\alpha}\bar{\auxii}_{c}\left(\lambda\gamma^{c}\right)_{\beta}}{(\lambda\bar{\lambda})}+\label{LinProjGenLam}\\
 &  & -{\scriptstyle 2\left(1-\auxi\right)\left(1-\sqrt{1-\auxi}\right)}f'\left(\auxi\right)\frac{\lambda^{\alpha}\bar{\lambda}_{\beta}}{(\lambda\bar{\lambda})}-\tfrac{f\left(\auxi\right)}{8\left(1+\sqrt{1-\auxi}\right)^{2}}\frac{\bar{\auxii}_{c}(\gamma^{c}\gamma_{b}\lambda)^{\alpha}\auxii^{d}(\bar{\lambda}\gamma^{b}\gamma_{d})_{\beta}}{(\lambda\bar{\lambda})}\nonumber \\
\lqn{\rest_{(f)\bot}^{\alpha\beta}(\rho,\bar{\rho})=}\nonumber \\
 & = & -\tfrac{f(\auxi)}{2\left(1+\sqrt{1-\auxi}\right)}\left(\auxii^{a}\gamma_{a}^{\alpha\beta}-\frac{\auxii^{a}(\gamma_{a}\bar{\lambda})^{\alpha}\lambda^{\beta}}{(\lambda\bar{\lambda})}\right)+\nonumber \\
 &  & +{\scriptstyle \left(1-\sqrt{1-\auxi}\right)}\left(\tfrac{\bar{f}(\auxi)}{1+\sqrt{1-\auxi}}-\tfrac{2\left(1-\auxi\right)\bar{f}(\auxi)f'\left(\auxi\right)}{f(\auxi)}\right)\frac{\lambda^{\alpha}\lambda^{\beta}}{(\lambda\bar{\lambda})}+\nonumber \\
 &  & +{\scriptstyle \left(1-\auxi\right)}f'\left(\auxi\right)\frac{\lambda^{\alpha}\auxii^{c}\left(\bar{\lambda}\gamma_{c}\right)^{\beta}}{(\lambda\bar{\lambda})}\label{linprojgenlam}
\end{eqnarray}

\item ...\enlargethispage*{1cm} map general variations $\delta\rho$ to
variations $\delta\lambda\equiv\Pi_{(f)\bot}\delta\rho+\rest_{(f)\bot}\delta\bar{\rho}$
which are $\gamma$-\textbf{orthogonal} to the pure spinor $\lambda^{\alpha}\equiv P_{(f)}^{\alpha}(\rho,\bar{\rho})$,
i.e. $(\lambda\gamma^{c}\delta\lambda)=0$, or equivalently 
\begin{eqnarray}
\lambda^{\alpha}\gamma_{\alpha\gamma}^{c}\Pi_{(f)\bot\beta}^{\gamma}(\rho,\bar{\rho}) & = & 0\label{LinProjGamma-orth}\\
\lambda^{\alpha}\gamma_{\alpha\gamma}^{c}\rest_{(f)\bot}^{\gamma\beta}(\rho,\bar{\rho}) & = & 0\qquad\forall\rho\label{linprojgamma-orth}
\end{eqnarray}

\item ... on the constraint surface ($\rho^{\alpha}=\lambda^{\alpha}$,
$\auxi=\auxii^{a}=0$) reduce for all $f$ which are differentiable
at $\auxi=0$ to 
\begin{eqnarray}
\Pi_{\bot\beta}^{\alpha}\equiv\Pi_{(f)\bot\beta}^{\alpha}(\lambda,\bar{\lambda}) & = & \delta_{\beta}^{\alpha}-\tfrac{1}{2}\frac{(\gamma_{a}\bar{\lambda})^{\alpha}(\lambda\gamma^{a})_{\beta}}{(\lambda\bar{\lambda})}\label{Proj-matrix-onshell}\\
\rest_{(f)\bot}^{\alpha\beta}(\lambda,\bar{\lambda}) & = & 0\label{proj-matrix-on-shell}
\end{eqnarray}
with 
\begin{equation}
\left(\Pi_{\bot}\right)^{2}=\Pi_{\bot}\label{PiSquare}
\end{equation}

\item ... obey projection properties in the sense%
\footnote{The equations (\ref{LinProjProp}) and (\ref{linprojprop}) are not
usual projection properties of a linear projection which would be
of the form $\Pi^{2}=\Pi$. Instead they are understood as being part
of a projection $\mc P_{(f)}$ acting on the tangent bundle of the
spinor space: 
\begin{eqnarray*}
\mc P_{(f)}:\quad(\rho,\bar{\rho}) & \stackrel{P_{(f)}}{\mapsto} & (\lambda,\bar{\lambda})\\
\left(\zwek{\delta\rho}{\delta\bar{\rho}}\right) & \mapsto & \left(\begin{array}{cc}
\Pi_{(f)\bot}(\rho,\bar{\rho}) & \rest_{(f)\bot}(\rho,\bar{\rho})\\
\rest_{(f)\bot}(\rho,\bar{\rho}) & \bar{\Pi}_{(f)\bot}(\rho,\bar{\rho})
\end{array}\right)\left(\zwek{\delta\rho}{\delta\bar{\rho}}\right)
\end{eqnarray*}
The projection property $\mc P_{(f)}\circ\mc P_{(f)}=\mc P_{(f)}$
for the tangent bundle map then is equivalent to the three equations
(\ref{LinProjProp}), (\ref{linprojprop}) and $P_{(f)}\circ P_{(f)}=P_{(f)}$
(\ref{proj-prop2}).$\quad\fussend$%
} 
\begin{eqnarray}
\Pi_{(f)\bot}(\lambda,\bar{\lambda})\Pi_{(f)\bot}(\rho,\bar{\rho}) & = & \Pi_{\bot(f)}(\rho,\bar{\rho})\qquad\label{LinProjProp}\\
\Pi_{(f)\bot}(\lambda,\bar{\lambda})\rest_{(f)\bot}(\rho,\bar{\rho}) & = & \rest_{\bot(f)}(\rho,\bar{\rho})\label{linprojprop}
\end{eqnarray}

\item ... have trace
\begin{equation}
\tr\Pi_{(f)\bot}(\rho,\bar{\rho})=\Bigl(11-\tfrac{4(1-\sqrt{1-\auxi})}{1+\sqrt{1-\auxi}}\Bigr)f\left(\auxi\right)-2(1-\auxi)\left(1-\sqrt{1-\auxi}\right)f'\left(\auxi\right)\label{LinProjTrace}
\end{equation}
which in general reduces only on the constraint surface $\rho=\lambda,\:\auxi=0$
to $\tr\Pi_{\bot}=11$.
\end{enumerate}
\end{prop}

The proof of this proposition is given in appendix \ref{app:proof2}
on page \pageref{app:proof2}.

\rem{There is no $f$ such that $\Pi_{(f)\bot}(\rho,\bar{\rho})=\Pi_{(f)\bot}(P(\rho,\bar{\rho}),\bar{P}(\rho,\bar{\rho}))$,
because the right side is Hermitian, so also the left side should
be. But we'll later see that there is a unique $f$ for a Hermitian
$\Pi$where this relation does not hold.}

\paragraph{Remarks}
\begin{itemize}
\item \label{rmk:projection-matrix}The projection matrix on the constraint
surface (\ref{Proj-matrix-onshell}) is the transpose of the projection
matrix $(1-K)$ given in equations (2.11) and (2.9) of \cite{Berkovits:2010zz}.
For the non-covariant projection \eqref{Pnoncov} we would replace
\eqref{Proj-matrix-onshell} by 
\begin{equation}
\Pi_{\bot\beta}^{\alpha}\equiv\Pi_{(f,\bar{\chi})\bot\beta}^{\quad\alpha}(\lambda)=\delta_{\beta}^{\alpha}-\tfrac{1}{2}\frac{(\gamma_{a}\bar{\chi})^{\alpha}(\lambda\gamma^{a})_{\beta}}{(\lambda\bar{\chi})}
\end{equation}
which is the transpose of $(1-K)$ in equation (17) of \cite{Oda:2001zm}
(see also (7) in \cite{Oda:2004bg}, (14) in \cite{Oda:2005sd} or
more recently (2.12) in \cite{Oda:2007ak}). The relation of this
non-covariant $\Pi_{(f,\bar{\chi})\bot\beta}^{\quad\alpha}(\lambda)$
to (\ref{Pnoncov}) is via the derivative like in the covariant case.
This is particularly obvious in the case where $\bar{\chi}_{\alpha}$
is pure and one can start from (\ref{non-cov-integrated}):
\begin{eqnarray}
\partial_{\rho^{\beta}}P_{(f,\bar{\chi})}^{\alpha}(\rho) & = & \delta_{\beta}^{\alpha}-\tfrac{1}{2}\frac{(\gamma_{a}\bar{\chi})^{\alpha}(\rho\gamma^{a})_{\beta}}{(\rho\bar{\chi})}+\tfrac{1}{4}\frac{(\rho\gamma^{a}\rho)(\bar{\chi}\gamma_{a})^{\alpha}\bar{\chi}_{\beta}}{(\rho\bar{\chi})^{2}}=\qquad\\
 & \stackrel{{\rm if\,}\rho=\lambda}{=} & \delta_{\beta}^{\alpha}-\tfrac{1}{2}\frac{(\gamma_{a}\bar{\chi})^{\alpha}(\lambda\gamma^{a})_{\beta}}{(\lambda\bar{\chi})}
\end{eqnarray}

\item Concerning the 6th statement of the proposal, there exists an $f$
for which $\tr\Pi_{(f)\bot}(\rho,\bar{\rho})=11\quad\forall\rho^{\alpha}$.
This solution is derived in footnote \ref{fn:differentialEqForf}
on page \pageref{fn:differentialEqForf}. It could therefore be that
this solution obeys a strict linear projector property $\Pi_{(f)\bot}(\rho,\bar{\rho})^{2}=\Pi_{(f)\bot}(\rho,\bar{\rho})$
as this is the case for $\Pi_{(1)\bot}(\rho,\bar{\rho})$ within the
toy-model in (\ref{Pi1quad}) on page \pageref{Pi1quad}.
\item The projection matrix maps general variations $\delta\rho^{\alpha}$
to some $\delta\lambda^{\alpha}$ that are $\gamma-$orthogonal to
$\lambda^{\alpha}\equiv P_{(f)}^{\alpha}(\rho,\bar{\rho})$, i.e.
such that $\lambda\gamma^{a}\delta\lambda=0$. As one can easily check,
any pure spinor variable together with a variation with the just mentioned
property $\lambda\gamma^{a}\delta\lambda=0$ obeys 
\begin{equation}
\delta\lambda^{\alpha}=(\Pi_{\bot}\delta\lambda)^{\alpha}\label{deltalambda}
\end{equation}
with $\Pi_{\bot}=\one-\tfrac{1}{2}\frac{(\gamma_{a}\bar{\lambda})\otimes(\lambda\gamma^{a})}{(\lambda\bar{\lambda})}$.
In particular for $\lambda^{\alpha}\equiv P_{(f)}^{\alpha}(\rho,\bar{\rho})$
this follows from \eqref{bigMatrix} together with \eqref{Proj-matrix-onshell}
and \eqref{proj-matrix-on-shell} at $\rho=\lambda$. A variation
of the constrained variable $\lambda$ thus can be written as 
\begin{eqnarray}
\delta\lambda^{\alpha}\partial_{\lambda^{\alpha}} & = & (\Pi_{\bot}\delta\lambda)^{\alpha}\partial_{\lambda^{\alpha}}=\delta\lambda^{\alpha}(\Pi_{\bot}^{T}\partial_{\lambda})_{\alpha}
\end{eqnarray}
The resulting \textbf{covariant derivative} $(\Pi_{\bot}^{T}\partial_{\lambda})_{\alpha}$
contains the transposed projection matrix and it leaves the constraint
invariant
\begin{equation}
(\Pi_{\bot}^{T}\partial_{\lambda})_{\alpha}(\lambda\gamma^{c}\lambda)=0
\end{equation}
In the sigma model the antighost $\omega_{z\alpha}$ plays the role
of the partial derivative $\partial_{\lambda^{\alpha}}$. As was noted
in the above cited references (e.g. \cite{Oda:2004bg} or \cite{Berkovits:2010zz}),
the expression $(\Pi_{\bot}^{T}\omega_{z})_{\alpha}$ is gauge invariant.
Gauge invariance means that via Poisson bracket or commutator the
constraint $\lambda\gamma^{c}\lambda$ (generating the gauge transformation)
annihilates $(\Pi_{\bot}^{T}\omega_{z})_{\alpha}$, or equivalently
the other way round, $(\Pi_{\bot}^{T}\omega_{z})_{\alpha}$ annihilates
the constraint. The latter point of view makes it a covariant derivative.
\end{itemize}

\section{Some natural projection matrices on the constraint surface}

\label{sec:projectorsOnSurface}Let us give in the following an overview
over several projection matrices acting on the tangent or cotangent
space of the constraint surface and list some of their properties.
They will appear frequently in the remaining discussion and therefore
a summary at this point will be very convenient. We are restricting
to the constraint surface, as we will later present a way to express
the linear projectors away from the surface in terms of those at the
surface. 

Let us start with our familiar
\begin{equation}
\Pi_{\bot}=\one-\frac{(\gamma^{a}\bar{\lambda})\otimes(\lambda\gamma_{a})}{2(\bar{\lambda}\lambda)}\quad,\quad\Pi_{\bot}^{\dagger}=\Pi_{\bot}\label{Piperp}
\end{equation}
It projects to an eleven dimensional subspace ($\tr\Pi_{\bot}=11$)
and further obeys 
\begin{eqnarray}
\left(\Pi_{\bot}\right)^{2} & \stackrel{(\ref{PiSurfIdem})}{=} & \Pi_{\bot}\\
\Pi_{\bot}\lambda & = & \lambda\quad,\quad\bar{\lambda}\Pi_{\bot}=\bar{\lambda}\nonumber \\
(\lambda\gamma^{c}\Pi_{\bot}) & = & 0=(\Pi_{\bot}\gamma^{c}\bar{\lambda})
\end{eqnarray}
The last line shows again that it projects to tangent `vectors' (spinors)
that are $\gamma$-orthogonal to $\lambda$. 

The unit matrix minus a projection matrix is always another projection
matrix that maps to the complementary subspace. It is thus not surprising
that
\begin{equation}
\Pi_{\Vert}\equiv\one-\Pi_{\bot}=\frac{(\gamma^{a}\bar{\lambda})\otimes(\lambda\gamma_{a})}{2(\bar{\lambda}\lambda)}\quad,\quad\Pi_{\Vert}^{\dagger}=\Pi_{\Vert}\label{Pipar}
\end{equation}
maps to a 5-dimensional space ($\tr\Pi_{\Vert}=5$) and obeys 
\begin{eqnarray}
\left(\Pi_{\Vert}\right)^{2} & = & \Pi_{\Vert}\label{PiParsquare}\\
\Pi_{\Vert}\lambda & = & 0=\bar{\lambda}\Pi_{\Vert}\\
(\lambda\gamma^{c}\Pi_{\Vert}) & = & (\lambda\gamma^{c})\quad,\quad(\Pi_{\Vert}\gamma^{c}\bar{\lambda})=(\gamma^{c}\bar{\lambda})
\end{eqnarray}
From the last line we would say that $\Pi_{\Vert}$ maps to spinors
which are $\gamma$-`parallel' to $\lambda^{\alpha}$. The matrices
$\Pi_{\bot}$ and $\Pi_{\Vert}$ are of course orthogonal to each
other 
\begin{equation}
\Pi_{\Vert}\Pi_{\bot}=\Pi_{\bot}\Pi_{\Vert}=0\label{PiparPiperp}
\end{equation}
Another projection matrix that will frequently appear is 
\begin{eqnarray}
\Pi_{\lambda} & \equiv & \frac{\lambda\otimes\bar{\lambda}}{(\bar{\lambda}\lambda)}\quad,\quad\Pi_{\lambda}^{\dagger}=\Pi_{\lambda}\label{Pilam}
\end{eqnarray}
It maps to the 1-dimensional subspace ($\tr\Pi_{\lambda}=1$) spanned
by $\lambda$ itself. Apart from that it shares several properties
with $\Pi_{\bot}$: 
\begin{eqnarray}
(\Pi_{\lambda})^{2} & = & \Pi_{\lambda}\\
\Pi_{\lambda}\lambda & = & \lambda\quad,\quad\bar{\lambda}\Pi_{\lambda}=\bar{\lambda}\\
(\lambda\gamma^{c}\Pi_{\lambda}) & = & 0=(\Pi_{\lambda}\gamma^{c}\bar{\lambda})
\end{eqnarray}
And then there is of course $\one-\Pi_{\lambda}$ with trace $15$
and sharing many properties with $\Pi_{\Vert}$:
\begin{eqnarray}
(\one-\Pi_{\lambda})\lambda & = & 0=\bar{\lambda}(\one-\Pi_{\lambda})\\
(\lambda\gamma^{c}(\one-\Pi_{\lambda})) & = & (\lambda\gamma^{c})\quad,\quad((\one-\Pi_{\lambda})\gamma^{c}\bar{\lambda})=(\gamma^{c}\bar{\lambda})
\end{eqnarray}
Finally we could take the transpose or complex conjugate of each
of the above projection matrices. As they are all Hermitian, the result
is the same: 
\begin{eqnarray}
\bar{\Pi}_{\bot}=\Pi_{\bot}^{T} & = & \one-\frac{(\gamma^{a}\lambda)\otimes(\bar{\lambda}\gamma_{a})}{2(\bar{\lambda}\lambda)}\\
\bar{\Pi}_{\Vert}=\Pi_{\Vert}^{T} & = & \frac{(\gamma^{a}\lambda)\otimes(\bar{\lambda}\gamma_{a})}{2(\bar{\lambda}\lambda)}\\
\bar{\Pi}_{\lambda}=\Pi_{\lambda}^{T} & = & \frac{\bar{\lambda}\otimes\lambda}{(\bar{\lambda}\lambda)}
\end{eqnarray}
They have the same trace as the original one and all the other properties
are obtained by interchanging $\lambda$ and $\bar{\lambda}$. 

It will be useful to see that one can express $\Pi_{\lambda}$ in
terms of $\bar{\Pi}_{\Vert}$ and vice versa%
\footnote{\label{fn:PiparPilamProofs}(\ref{PiparInTermsOfPilam}) is obvious
from the definitions (\ref{Pipar}) and (\ref{Pilam}) of $\Pi_{\Vert}$
and $\Pi_{\lambda}$. Contracting now (\ref{PiparInTermsOfPilam})
from the right with $\gamma_{b}\auxii^{b}$, one obtains 
\[
(\Pi_{\Vert}\gamma_{b})\auxii^{b}=\tfrac{1}{2}(\gamma^{a}\bar{\Pi}_{\lambda}\underbrace{\gamma_{a}\gamma_{b}}_{-\gamma_{b}\gamma_{a}\lqn{{\scriptstyle +2\eta_{ab}}}})\auxii^{b}=-\tfrac{1}{2}(\gamma^{a}\underbrace{\bar{\Pi}_{\lambda}\gamma_{b}\gamma_{a})\auxii^{b}}_{=0\quad\mbox{\tiny\eqref{ProjZsquareZero}}}+\auxii^{a}(\gamma_{a}\bar{\Pi}_{\lambda})
\]
Together with its complex conjugate version, this gives (\ref{PiparZisZPibarlam}).
Equation (\ref{ZbarPiparIsPibarlamZbar}) is then simply obtained
by Hermitian conjugation, using that the projection matrices and the
gamma-matrices are all Hermitian. Finally, contracting the second
equation of (\ref{PiparZisZPibarlam}) from the left with $\tfrac{1}{4\auxi}\auxii^{b}\gamma_{b}$,
we obtain 
\[
\tfrac{1}{4\auxi}\auxii^{b}(\gamma_{b}\bar{\Pi}_{\Vert}\gamma^{a})\bar{\auxii}_{a}=\tfrac{1}{4\auxi}\auxii^{b}\bar{\auxii}_{a}(\underbrace{\gamma_{b}\gamma^{a}}_{-\gamma^{a}\gamma_{b}\lqn{{\scriptstyle +2\delta_{b}^{a}}}}\Pi_{\lambda})=-\tfrac{1}{4\auxi}\bar{\auxii}_{a}\underbrace{\auxii^{b}(\gamma^{a}\gamma_{b}\Pi_{\lambda}}_{=0\quad\mbox{\tiny\eqref{zsquarezero}}})+\underbrace{\tfrac{1}{2\auxi}\auxii^{a}\bar{\auxii}_{a}}_{=1}\Pi_{\lambda}
\]
which proves (\ref{PilamInTermsOfPipar}).$\quad\fussend$%
}:
\begin{eqnarray}
\Pi_{\Vert} & = & \tfrac{1}{2}(\gamma^{a}\bar{\Pi}_{\lambda}\gamma_{a})\label{PiparInTermsOfPilam}\\
\Pi_{\lambda} & = & \tfrac{1}{4\auxi}\auxii^{b}(\gamma_{b}\bar{\Pi}_{\Vert}\gamma^{c})\bar{\auxii}_{c}\label{PilamInTermsOfPipar}
\end{eqnarray}
As intermediate steps between the above equations, we also have 
\begin{eqnarray}
(\Pi_{\Vert}\gamma_{a})\auxii^{a} & = & \auxii^{a}(\gamma_{a}\bar{\Pi}_{\lambda})\quad,\quad(\bar{\Pi}_{\Vert}\gamma^{a})\bar{\auxii}_{a}=\bar{\auxii}_{a}(\gamma^{a}\Pi_{\lambda})\label{PiparZisZPibarlam}\\
\bar{\auxii}_{a}(\gamma^{a}\Pi_{\Vert}) & = & (\bar{\Pi}_{\lambda}\gamma^{a})\bar{\auxii}_{a}\quad,\quad\auxii^{a}(\gamma_{a}\bar{\Pi}_{\Vert})=(\Pi_{\lambda}\gamma_{a})\auxii^{a}\label{ZbarPiparIsPibarlamZbar}
\end{eqnarray}
Note that the same relations hold if one replaces $\auxii^{a}$ and
$\bar{\auxii}_{a}$ by $(\gamma^{a}\lambda)_{\alpha}$ and $(\gamma_{a}\bar{\lambda})^{\alpha}$
where the proof works just like the one in footnote \ref{fn:PiparPilamProofs}:
\begin{eqnarray}
(\Pi_{\Vert}\gamma_{a})(\gamma^{a}\lambda)_{\alpha} & = & (\gamma^{a}\lambda)_{\alpha}(\gamma_{a}\bar{\Pi}_{\lambda})\quad,\quad(\bar{\Pi}_{\Vert}\gamma^{a})(\gamma_{a}\bar{\lambda})^{\alpha}=(\gamma_{a}\bar{\lambda})^{\alpha}(\gamma^{a}\Pi_{\lambda})\qquad\label{PiparZisZPibarlam2}\\
(\gamma_{a}\bar{\lambda})^{\alpha}(\gamma^{a}\Pi_{\Vert}) & = & (\bar{\Pi}_{\lambda}\gamma^{a})(\gamma_{a}\bar{\lambda})^{\alpha}\quad,\quad(\gamma^{a}\lambda)_{\alpha}(\gamma_{a}\bar{\Pi}_{\Vert})=(\Pi_{\lambda}\gamma_{a})(\gamma^{a}\lambda)_{\alpha}\qquad\label{ZbarPiparIsPibarlamZbar2}
\end{eqnarray}
In addition one can easily check from the definitions that we have
\begin{eqnarray}
\Pi_{\bot}\Pi_{\lambda} & = & \Pi_{\lambda}\Pi_{\bot}=\Pi_{\lambda}\label{PiperpPilam}\\
\Pi_{\bot}\gamma^{c}\bar{\Pi}_{\lambda} & = & \bar{\Pi}_{\lambda}\gamma^{c}\Pi_{\bot}=\Pi_{\lambda}\gamma^{c}\bar{\Pi}_{\bot}=\bar{\Pi}_{\bot}\gamma^{c}\Pi_{\lambda}=0\label{PiperpGammaPibarlam}
\end{eqnarray}
We will later also need the variation of $\Pi_{\bot}=\one-\tfrac{1}{2}\frac{(\gamma_{a}\bar{\lambda})\otimes(\lambda\gamma^{a})}{(\lambda\bar{\lambda})}$
and some of the above formulas simplify this calculation: 
\begin{equation}
\delta\Pi_{\bot\beta}^{\alpha}=-\tfrac{(\gamma_{a}\bar{\lambda})^{\alpha}}{2(\lambda\bar{\lambda})}\underbrace{\left(\gamma_{\beta\gamma}^{a}-\frac{(\lambda\gamma^{a})_{\beta}\bar{\lambda}_{\gamma}}{(\lambda\bar{\lambda})}\right)}_{\gamma_{\beta\delta}^{a}\left(\one-\Pi_{(\lambda)}\right)^{\delta}\tief{\gamma}}\delta\lambda^{\gamma}-\delta\bar{\lambda}_{\gamma}\underbrace{\left(\gamma_{a}^{\gamma\alpha}-\lambda^{\gamma}\frac{(\gamma_{a}\bar{\lambda})^{\alpha}}{(\lambda\bar{\lambda})}\right)}_{\left(\one-\Pi_{(\lambda)}\right)^{\gamma}\tief{\delta}\gamma_{a}^{\delta\alpha}}\tfrac{(\lambda\gamma^{a})_{\beta}}{2(\lambda\bar{\lambda})}\quad
\end{equation}
Using (\ref{PiparZisZPibarlam2}) and (\ref{ZbarPiparIsPibarlamZbar2})
we obtain 
\begin{eqnarray}
\delta\Pi_{\bot\beta}^{\alpha} & = & -\tfrac{(\gamma_{a}\bar{\lambda})^{\alpha}}{2(\lambda\bar{\lambda})}(\bar{\Pi}_{\bot}\gamma^{a})_{\beta\gamma}\delta\lambda^{\gamma}-\delta\bar{\lambda}_{\gamma}(\gamma_{a}\bar{\Pi}_{\bot})^{\gamma\alpha}\tfrac{(\lambda\gamma^{a})_{\beta}}{2(\lambda\bar{\lambda})}\label{deltaPiperp}
\end{eqnarray}

\section{Hermitian projection matrix and the projection potential}

The matrix $\Pi_{\bot}=\one-\frac{(\gamma^{a}\bar{\lambda})\otimes(\lambda\gamma_{a})}{2(\bar{\lambda}\lambda)}$
(\ref{Piperp}),(\ref{Proj-matrix-onshell}) on the constraint surface
is Hermitian. We will later in the field theory application see that
this is a very useful property. It is therefore natural to ask, whether
this can be realized also off the constraint surface. Off the constraint
surface we have a non-vanishing contribution $\rest_{(f)\bot}$ in
(\ref{bigMatrix}) and therefore should consider the complete matrix
$\left(\begin{array}{cc}
\Pi_{(f)\bot} & \rest_{(f)\bot}\\
\bar{\rest}_{(f)\bot} & \bar{\Pi}_{(f)\bot}
\end{array}\right)$ instead of only the block $\Pi_{(f)\bot}$.

\label{sec:hermitean}\begin{prop}[Hermiticity]\label{prop:hermiticity}Remember
the definition in equation (\ref{hdef-first}) of the function $h$
\begin{equation}
h(\auxi)\equiv\frac{1+\sqrt{1-\auxi}}{2\sqrt{1-\auxi}}\mbox{ for }\auxi\in[0,1[\label{h-def}
\end{equation}
and assume that the function $f$ which defines the projection $P_{(f)}^{\alpha}$
is differentiable%
\footnote{At $0$ differentiability is again understood in the sense of footnote
\ref{fn:differentiable} on page \pageref{fn:differentiable}.$\quad\fussend$%
} in an interval $I\subset[0,b[,\quad b\leq1$ (so in a neighbourhood
of the constraint surface $\auxi=0\in I$) with continuous $f'$ at
least at $0$. Then the following statements hold:
\begin{enumerate}
\item The matrix $\left(\begin{array}{cc}
\Pi_{(f)\bot\beta}^{\alpha}(\rho,\bar{\rho}) & \rest_{(f)\bot}^{\alpha\beta}(\rho,\bar{\rho})\\
\bar{\rest}_{(f)\bot\alpha\beta}(\rho,\bar{\rho}) & \bar{\Pi}_{(f)\bot\alpha}\hoch{\beta}(\rho,\bar{\rho})
\end{array}\right)$ is Hermitian for all $\rho$ where $\auxi\equiv\frac{(\rho\gamma^{a}\rho)(\bar{\rho}\gamma_{a}\bar{\rho})}{2(\rho\bar{\rho})}\in I$
if and only if $f=h$ in $I$. \\
For the blocks of the matrix, this means that 
\begin{eqnarray}
\Pi_{(h)\bot}^{\dagger} & = & \Pi_{(h)\bot}\\
\rest_{(h)\bot}^{T} & = & \rest_{(h)\bot}
\end{eqnarray}

\item There exists a potential 
\begin{equation}
\Phi(\rho,\bar{\rho})\equiv\tfrac{(\rho\bar{\rho})}{2}(1+\sqrt{1-\auxi})\label{Potential}
\end{equation}
such that 
\begin{equation}
P_{(h)}^{\alpha}=\partial_{\bar{\rho}_{\alpha}}\Phi\quad,\quad\bar{P}_{(h)\alpha}=\partial_{\rho^{\alpha}}\Phi
\end{equation}

\item The potential $\Phi$ can be written as
\begin{equation}
\Phi(\rho,\bar{\rho})=P_{(h)}^{\alpha}(\rho,\bar{\rho})\bar{P}_{(h)\alpha}(\rho,\bar{\rho})\label{PhiIsModSqu}
\end{equation}

\end{enumerate}
\end{prop}

The proof of this proposition is given in appendix \ref{app:proof3}
on page \pageref{app:proof3}.

\paragraph{Remarks}
\begin{itemize}
\item The Kähler potential on the pure spinor space is given by $(\lambda\bar{\lambda})$.
The pontential (\ref{PhiIsModSqu}) is therefore the pullback of the
Kähler potential into the ambient space along the projection $P_{(h)}^{\alpha}$.
That does not explain, however, why this is at the same time a potential
for the projection itself. 
\item The function $h$ defined in (\ref{h-def}) obviously is divergent
at $\auxi=1$. It thus does not really belong to the class of functions
that we discussed in the previous propositions, as there we assumed
$f$ to be defined on the closed interval $[0,1]$. Statements of
these propositions that where about $\auxi=1$ thus need to be reconsidered.
In particular the projection $P_{(h)}^{\alpha}$ is not defined for
those Weyl spinors $\rho^{\alpha}$ for which $\auxi=1$. Thus $\auxi=1$
also drops out of the zero-locus \eqref{kerP}. In addition $h$ does
not have any zeroes, so the zero-locus is simply given by
\begin{equation}
P_{(h)}^{-1}(0)=\{0\}\label{kerPh}
\end{equation}
where $P_{(h)}^{\alpha}(0,0)$ is defined just via the limit $\lim_{\abs{\rho}\to0}P_{(h)}^{\alpha}(\rho,\bar{\rho})=0$
like in proposition \ref{prop:covproj}. Having no well-defined projection
at $\auxi=1$ seems odd, but a priori it is most important that $P_{(h)}^{\alpha}$
is well-behaved close to the constraint surface where $\auxi=0$ and
which is thus `far away' from the troublesome points. In addition
at least the absolute value of $P_{(h)}^{\alpha}(\rho,\bar{\rho})$
which is given according to \eqref{PhiIsModSqu} by the potential
\eqref{Potential} is well behaved at $\auxi=1$ and converges to
\begin{equation}
\lim_{\auxi\to1}\left(P_{(h)}^{\alpha}(\rho,\bar{\rho})\bar{P}_{(h)\alpha}(\rho,\bar{\rho})\right)=\tfrac{1}{2}(\rho\bar{\rho})\label{ModSqAtOne}
\end{equation}
\rem{There is thus the hope that the path-integral converges, even
\\
when integrating over the whole space of Weyl spinors. I.e. that the
$\auxi=1$ subspace does not contribute. }
\end{itemize}
Let us rewrite some of the formulas for our projection in the specific
case $f(\auxi)=h(\auxi)\equiv\frac{1+\sqrt{1-\auxi}}{2\sqrt{1-\auxi}}$
whose derivative is

\begin{equation}
h'(\auxi)=\tfrac{1}{4}\tfrac{1}{\sqrt{1-\auxi}^{3}}\label{hprime}
\end{equation}
The non-linear projection map (\ref{ProjGeneral}) becomes%
\footnote{\label{fn:anotherHermProp}Another property of the Hermitian projector
(though probably quite meaningless) is that the difference from $\rho^{\alpha}$
to its projection $P_{(h)}^{\alpha}(\rho,\bar{\rho})$ can be nicely
expressed in terms of $\bar{P}_{(h)\alpha}(\rho,\bar{\rho})$:
\[
\rho^{\alpha}-P_{(h)}^{\alpha}(\rho,\bar{\rho})=\tfrac{1}{2\left(1+\sqrt{1-\auxi}\right)}\auxii^{a}\gamma_{a}^{\alpha\beta}\bar{P}_{(h)\beta}(\rho,\bar{\rho})\qquad\fussend
\]
\frem{Also in the toy model we have $\rho^{I}-P_{(h)}^{I}\propto\bar{P}_{(h)}^{I}$}%
} 
\begin{equation}
P_{(h)}^{\alpha}(\rho,\bar{\rho})\equiv\frac{1+\sqrt{1-\auxi}}{2\sqrt{1-\auxi}}\rho^{\alpha}-\frac{\auxii^{a}(\bar{\rho}\gamma_{a})^{\alpha}}{4\sqrt{1-\auxi}}\label{hermProjector}
\end{equation}
with still (\ref{zxDefandfzero}) 
\begin{equation}
\auxii^{a}\equiv\frac{(\rho\gamma^{a}\rho)}{(\rho\bar{\rho})},\quad\bar{\auxii}_{a}\equiv\frac{(\bar{\rho}\gamma_{a}\bar{\rho})}{(\rho\bar{\rho})},\quad\auxi\equiv\tfrac{1}{2}\auxii^{a}\bar{\auxii}_{a}\label{xzDefRep}
\end{equation}
The modulus squared \eqref{Proj-modulus} is now according to \eqref{PhiIsModSqu}
in the proposition given by $\Phi$ in \eqref{Potential}. 

The linearized tangent space projection matrices (\ref{LinProjGen})
and (\ref{linprojgen}) become 
\begin{eqnarray}
\lqn{\Pi_{(h)\bot\beta}^{\alpha}(\rho,\bar{\rho})=}\nonumber \\
 & = & \tfrac{1+\sqrt{1-\auxi}}{2\sqrt{1-\auxi}}\delta_{\beta}^{\alpha}-\tfrac{1}{2\sqrt{1-\auxi}}\frac{(\gamma_{a}\bar{\rho})^{\alpha}(\gamma^{a}\rho)_{\beta}}{(\rho\bar{\rho})}-\tfrac{1}{8}\tfrac{1}{\sqrt{1-\auxi}^{3}}\frac{\auxii^{a}(\gamma_{a}\bar{\rho})^{\alpha}\bar{\auxii}_{b}(\gamma^{b}\rho)_{\beta}}{(\rho\bar{\rho})}+\nonumber \\
 &  & -\tfrac{1}{2}\tfrac{\auxi}{\sqrt{1-\auxi}^{3}}\frac{\rho^{\alpha}\bar{\rho}_{\beta}}{(\rho\bar{\rho})}+\tfrac{1}{4}\tfrac{1}{\sqrt{1-\auxi}^{3}}\frac{\rho^{\alpha}\bar{\auxii}_{b}(\gamma^{b}\rho)_{\beta}}{(\rho\bar{\rho})}+\tfrac{1}{4}\tfrac{1}{\sqrt{1-\auxi}^{3}}\frac{\auxii^{a}(\gamma_{a}\bar{\rho})^{\alpha}\bar{\rho}_{\beta}}{(\rho\bar{\rho})}\label{HermLinProj}\\
\lqn{\rest_{(h)\bot}^{\alpha\beta}(\rho,\bar{\rho})=}\nonumber \\
 & = & -\tfrac{1}{2}\tfrac{\auxi}{\sqrt{1-\auxi}^{3}}\frac{\rho^{\alpha}\rho^{\beta}}{(\rho\bar{\rho})}+\tfrac{1}{4}\tfrac{1}{\sqrt{1-\auxi}^{3}}\frac{\rho^{\alpha}\auxii^{b}(\gamma_{b}\bar{\rho})^{\beta}}{(\rho\bar{\rho})}+\tfrac{1}{4\sqrt{1-\auxi}^{3}}\frac{\auxii^{a}(\gamma_{a}\bar{\rho})^{\alpha}\rho^{\beta}}{(\rho\bar{\rho})}+\nonumber \\
 &  & -\tfrac{1}{8\sqrt{1-\auxi}^{3}}\frac{\auxii^{a}\auxii^{b}(\gamma_{a}\bar{\rho})^{\alpha}(\gamma_{b}\bar{\rho})^{\beta}}{(\rho\bar{\rho})}-\tfrac{1}{4\sqrt{1-\auxi}}\auxii^{a}\gamma_{a}^{\alpha\beta}\label{hermlinproj}
\end{eqnarray}
\rem{In the toy model (Hermitian case), $\pi_{(h)}$ was proportional
to $\bar{\pi}_{(h)}$. A similar relation holds also in the pure spinor
case: $\tfrac{1}{4\auxi}\bar{\auxii}_{b}\gamma^{b}\pi\gamma^{c}\bar{\auxii}_{c}=\bar{\pi}$
Corresponding calculation is hidden here. }\rembreak  The equation
(\ref{Proj-inv}) for expressing $\rho^{\alpha}$ in terms of $\lambda^{\alpha}\equiv P^{\alpha}(\rho,\bar{\rho})$
and $\bar{\lambda}_{\alpha}\equiv\bar{P}_{\alpha}(\rho,\bar{\rho})$
turns for $f(\auxi)=\frac{1+\sqrt{1-\auxi}}{2\sqrt{1-\auxi}}$ into
\begin{equation}
\rho^{\alpha}=\lambda^{\alpha}+\frac{1}{2\left(1+\sqrt{1-\auxi}\right)}\auxii^{a}\left(\bar{\lambda}\gamma_{a}\right)^{\alpha}\label{hermInverse}
\end{equation}
In the Hermitian case the rewriting of the projection matrices in
terms of $\lambda^{\alpha}$ and its complex conjugate as was done
in general in (\ref{LinProjGenLam}) and (\ref{linprojgenlam}) becomes
particularly useful, as the $\lambda\otimes(\lambda\gamma^{c})$-term
drops (in addition to the already missing $(\gamma^{c}\bar{\lambda})\otimes\bar{\lambda}$-terms)
and one is left with%
\footnote{\label{fn:xzrho-derivsLamHerm}The $\rho$-derivatives of $\auxii^{a}$
and $\auxi$ in (\ref{zbyrho-derivative})-(\ref{xbyrhoderivative}),
rewritten in terms of $\lambda\equiv P_{(f)}(\rho,\bar{\rho})$, as
it is done in the appendix on page \pageref{fn:zx-derivs-lambda}
in footnote \ref{fn:zx-derivs-lambda}, now turn for $f(\auxi)=h(\auxi)\equiv\frac{1+\sqrt{1-\auxi}}{2\sqrt{1-\auxi}}$
into
\begin{eqnarray*}
\partial_{\rho^{\beta}}\auxii^{a} & = & \frac{1}{4(\lambda\bar{\lambda})}\Bigl\{4\left(1+\sqrt{1-\auxi}\right)\left(\lambda\gamma^{a}\right)_{\beta}-\auxii^{a}\bar{\auxii}_{b}\left(\lambda\gamma^{b}\right)_{\beta}+\\
 &  & -2\auxii^{b}\left(\gamma_{b}\gamma^{a}\bar{\lambda}\right)_{\beta}+2\left(1-\sqrt{1-\auxi}\right)\auxii^{a}\bar{\lambda}_{\beta}\Bigr\}\qquad\\
\partial_{\bar{\rho}_{\beta}}\auxii^{a} & = & -\frac{\auxii^{a}}{4(\lambda\bar{\lambda})}\Bigl\{2\left(1+\sqrt{1-\auxi}\right)\lambda^{\beta}+\auxii^{b}\left(\bar{\lambda}\gamma_{b}\right)^{\beta}\Bigr\}\\
\partial_{\rho^{\beta}}\auxi & = & \frac{\sqrt{1-\auxi}}{2(\lambda\bar{\lambda})}\Bigl\{\left(1+\sqrt{1-\auxi}\right)\bar{\auxii}_{c}\left(\gamma^{c}\lambda\right)_{\beta}-2\auxi\bar{\lambda}_{\beta}\Bigr\}\qquad\fussend
\end{eqnarray*}
}$\hoch ,$%
\footnote{\label{fn:Pih-intermsofnewz}In terms of the alternative parametrization
$\Auxii^{a}$ from equation (\ref{newz}), equations (\ref{HermLinProjLam})
and (\ref{hermlinprojlam}) turn into
\begin{eqnarray*}
\Pi_{(h)\bot\beta}^{\alpha}(\rho,\bar{\rho}) & = & \tfrac{1}{1-\Auxi}\left(\delta_{\beta}^{\alpha}-\tfrac{(\gamma_{a}\bar{\lambda})^{\alpha}(\lambda\gamma^{a})_{\beta}}{2(\lambda\bar{\lambda})}\right)-\tfrac{\Auxi}{(1-\Auxi)}\tfrac{\lambda^{\alpha}\bar{\lambda}_{\beta}}{(\lambda\bar{\lambda})}-\tfrac{1}{8(1-\Auxi)}\tfrac{\bar{\Auxii}_{c}(\gamma^{c}\gamma_{b}\lambda)^{\alpha}\Auxii^{d}(\bar{\lambda}\gamma^{b}\gamma_{d})_{\beta}}{(\lambda\bar{\lambda})}\\
\rest_{(h)\bot}^{\alpha\beta}(\rho,\bar{\rho}) & = & -\tfrac{\Auxii^{a}}{2(1-\Auxi)}\left(\gamma_{a}^{\alpha\beta}-\tfrac{(\gamma_{a}\bar{\lambda})^{\alpha}\lambda^{\beta}}{(\lambda\bar{\lambda})}-\tfrac{\lambda^{\alpha}\left(\bar{\lambda}\gamma_{a}\right)^{\beta}}{(\lambda\bar{\lambda})}\right)\qquad\fussend
\end{eqnarray*}
The same result is obtained when using $\tilde{h}(\Auxi)\equiv h(\auxi)=\frac{1}{1-\Auxi}$
(\ref{hdef-first}) and $\tilde{h}'(\Auxi)=\frac{1}{(1-\Auxi)^{2}}$
in the equations of footnote \ref{fn:Pi-intermsofnewz} on page \pageref{fn:Pi-intermsofnewz}.$\quad\fussend$%
}:
\begin{eqnarray}
\Pi_{(h)\bot\beta}^{\alpha}(\rho,\bar{\rho}) & = & \frac{1+\sqrt{1-\auxi}}{2\sqrt{1-\auxi}}\left(\delta_{\beta}^{\alpha}-\frac{1}{2}\frac{(\gamma^{a}\bar{\lambda})^{\alpha}(\gamma_{a}\lambda)_{\beta}}{(\lambda\bar{\lambda})}\right)-\frac{1-\sqrt{1-\auxi}}{2\sqrt{1-\auxi}}\frac{\lambda^{\alpha}\bar{\lambda}_{\beta}}{(\lambda\bar{\lambda})}+\nonumber \\
 &  & -\frac{1}{16\sqrt{1-\auxi}\left(1+\sqrt{1-\auxi}\right)}\frac{\bar{\auxii}_{c}(\gamma^{c}\gamma_{b}\lambda)^{\alpha}\auxii^{d}(\bar{\lambda}\gamma^{b}\gamma_{d})_{\beta}}{(\lambda\bar{\lambda})}\label{HermLinProjLam}\\
\rest_{(h)\bot}^{\alpha\beta}(\rho,\bar{\rho}) & = & -\frac{\auxii^{a}}{4\sqrt{1-\auxi}}\Bigl(\gamma_{a}^{\alpha\beta}-\tfrac{(\gamma_{a}\bar{\lambda})^{\alpha}\lambda^{\beta}}{(\lambda\bar{\lambda})}-\tfrac{\lambda^{\alpha}(\gamma_{a}\bar{\lambda})^{\beta}}{(\lambda\bar{\lambda})}\Bigr)\label{hermlinprojlam}
\end{eqnarray}
It is further convenient in the Hermitian case, to express the projection
matrices off the constraint surface in terms of the projection matrices
$\Pi_{\bot\beta}^{\alpha}$ (\ref{Piperp}) and $\Pi_{(\lambda)\beta}^{\alpha}$
(\ref{Pilam}) defined on the surface:
\begin{eqnarray}
\Pi_{(h)\bot\beta}^{\alpha}(\rho,\bar{\rho}) & = & \frac{1+\sqrt{1-\auxi}}{2\sqrt{1-\auxi}}\Pi_{\bot\beta}^{\alpha}-\frac{1-\sqrt{1-\auxi}}{2\sqrt{1-\auxi}}\Pi_{(\lambda)\beta}^{\alpha}+\nonumber \\
 &  & -\frac{1}{8\sqrt{1-\auxi}\left(1+\sqrt{1-\auxi}\right)}\frac{\bar{\auxii}_{c}(\gamma^{c}\bar{\Pi}_{\Vert}\gamma_{d})_{\beta}\auxii^{d}}{(\lambda\bar{\lambda})}\qquad\label{HermLinProjPi}\\
\rest_{(h)\bot}^{\alpha\beta}(\rho,\bar{\rho}) & = & -\frac{1}{4\sqrt{1-\auxi}}\Bigl((\Pi_{\bot}\gamma_{a})^{\alpha\beta}\auxii^{a}-(\Pi_{(\lambda)}\gamma_{a})^{\alpha\beta}\auxii^{a}\Bigr)=\label{hermlinprojpileft}\\
 & \ous{{\scriptstyle (\ref{PiparZisZPibarlam})}}={{\scriptstyle (\ref{ZbarPiparIsPibarlamZbar})}} & -\frac{1}{4\sqrt{1-\auxi}}\Bigl(\auxii^{a}(\gamma_{a}\bar{\Pi}_{\bot})^{\alpha\beta}-\auxii^{a}(\gamma_{a}\bar{\Pi}_{\lambda})^{\alpha\beta}\Bigr)\label{hermlinprojpiright}
\end{eqnarray}
For the field theory application in section \ref{sec:ghost-action},
one of the most important properties of the Hermitian projection-matrices
will be the fact that the order of the matrix multiplication in the
projection property (\ref{LinProjProp}) will not matter any longer
as these matrices will commute. Starting from the original projection
property{\small 
\begin{equation}
{\scriptstyle \left(\!\!\begin{array}{cc}
\Pi_{\bot\gamma}^{\alpha} & 0\\
0 & \bar{\Pi}_{\bot\alpha}\hoch{\gamma}
\end{array}\!\!\right)\!\!\left(\!\!\begin{array}{cc}
\Pi_{(h)\bot\beta}^{\gamma}(\rho,\bar{\rho}) & \rest_{(h)\bot}^{\gamma\beta}(\rho,\bar{\rho})\\
\bar{\rest}_{(h)\bot\gamma\beta}(\rho,\bar{\rho}) & \bar{\Pi}_{(h)\bot\gamma}\hoch{\beta}(\rho,\bar{\rho})
\end{array}\!\!\right)=\left(\!\!\begin{array}{cc}
\Pi_{(h)\bot\beta}^{\alpha}(\rho,\bar{\rho}) & \rest_{(h)\bot}^{\alpha\beta}(\rho,\bar{\rho})\\
\bar{\rest}_{(h)\bot\alpha\beta}(\rho,\bar{\rho}) & \bar{\Pi}_{(h)\bot\alpha}\hoch{\beta}(\rho,\bar{\rho})
\end{array}\!\!\right)}
\end{equation}
}and taking the Hermitian conjugate on both sides, we arrive (because
of Hermiticity) at{\small  
\begin{equation}
{\scriptstyle \left(\!\!\begin{array}{cc}
\Pi_{(h)\bot\gamma}^{\alpha}(\rho,\bar{\rho}) & \rest_{(h)\bot}^{\alpha\gamma}(\rho,\bar{\rho})\\
\bar{\rest}_{(h)\bot\alpha\gamma}(\rho,\bar{\rho}) & \bar{\Pi}_{(h)\bot\alpha}\hoch{\gamma}(\rho,\bar{\rho})
\end{array}\!\!\right)\!\!\left(\!\!\begin{array}{cc}
\Pi_{\bot\beta}^{\gamma} & 0\\
0 & \bar{\Pi}_{\bot\gamma}\hoch{\beta}
\end{array}\!\!\right)=\left(\!\!\begin{array}{cc}
\Pi_{(h)\bot\beta}^{\alpha}(\rho,\bar{\rho}) & \rest_{(h)\bot}^{\alpha\beta}(\rho,\bar{\rho})\\
\bar{\rest}_{(h)\bot\alpha\beta}(\rho,\bar{\rho}) & \bar{\Pi}_{(h)\bot\alpha}\hoch{\beta}(\rho,\bar{\rho})
\end{array}\!\!\right)}
\end{equation}
}which is equivalent to the following two equations and their complex
conjugates respectively\vspace{-.3cm} 

\begin{eqnarray}
\Pi_{(h)\bot\gamma}^{\alpha}(\rho,\bar{\rho})\Pi_{\bot\beta}^{\gamma} & = & \Pi_{\bot\gamma}^{\alpha}\Pi_{(h)\bot\beta}^{\gamma}(\rho,\bar{\rho})\quad\left(=\Pi_{(h)\bot\gamma}^{\alpha}(\rho,\bar{\rho})\right)\label{Proj-commutativity}\\
\rest_{(h)\bot}^{\alpha\gamma}\bar{\Pi}_{\bot\gamma}\hoch{\beta} & = & \Pi_{\bot\gamma}^{\alpha}\rest_{(h)\bot}^{\gamma\beta}\quad\left(=\rest_{(h)\bot}^{\alpha\beta}\right)\label{proj-commutativity}
\end{eqnarray}
Using (\ref{PiperpPilam}),(\ref{ZbarPiparIsPibarlamZbar}) and (\ref{PiparPiperp})
it is also easy to check these equations explicitly.

\section{Natural projection in the U(5) formalism}

\label{sec:U5}The 10 Dirac gamma matrices $\Gamma^{a}$ can be used
to define 5 pairs of creation and annihilation matrices
\begin{eqnarray}
b^{\mf a} & \equiv & \tfrac{1}{2}\left(\Gamma^{2\mf a-1}-i\Gamma^{2\mf a}\right)\label{annihilator}\\
b_{\mf b}^{\dagger} & \equiv & \tfrac{1}{2}\left(\Gamma_{2\mf b-1}+i\Gamma_{2\mf b}\right),\quad{\scriptstyle \mf a,\mf b\in\{1,\ldots,5\}}\label{creator}\\
\{b^{\mf a},b_{\mf b}^{\dagger}\} & = & \delta_{\mf b}^{\mf a}
\end{eqnarray}
For $SO(1,9)$ one simply replaces $\Gamma^{10}$ by $i\Gamma^{0}$
or $\Gamma_{10}$ by $-i\Gamma_{0}$. The creation matrices can be
used to build a Fock space representation of spinors by acting on
a vacuum spinor $\ket{\Omega}$ which is annihilated by all annihilation
matrices (see e.g. Appendix B.1 of \cite[p.430]{Polchinski:1998rr}
and in particular in the pure spinor context appendix D of \cite{Nh:2003cm}
for a more detailed discussion of this parametrization): 
\begin{eqnarray}
\ket{\tief{\mf a}} & \equiv & b_{\mf a}^{\dagger}\ket{\Omega},\qquad\tbinom{5}{1}=5\mbox{ states}\\
\ket{\tief{\mf a_{1}\mf a_{2}}} & \equiv & b_{\mf a_{1}}^{\dagger}b_{\mf a_{2}}^{\dagger}\ket{\Omega},\qquad\tbinom{5}{2}=10\mbox{ states}\\
 & \ddots\nonumber \\
\ket{\tief{\mf a_{1}\ldots\mf a_{5}}} & \equiv & b_{\mf a_{1}}^{\dagger}\cdots b_{\mf a_{5}}^{\dagger}\ket{\Omega},\qquad\tbinom{5}{5}=1\mbox{ state}\label{alternative-vac}
\end{eqnarray}
Together with the vacuum $\ket{\Omega}$, these are precisely $\sum_{k=0}^{5}\tbinom{5}{k}=2^{5}=32$
states. An arbitrary Dirac spinor $\ket{\Psi}$ can therefore be expanded
in this basis. It is a well known fact that chirality in this picture
corresponds to an even number of creators: 
\begin{eqnarray}
\ket{\Psi} & \equiv & \!\!\underbrace{\Psi^{+}\ket{\Omega}\!+\!\tfrac{1}{2}\Psi^{\mf a_{1}\mf a_{2}}\ket{\tief{\mf a_{1}\mf a_{2}}}+\tfrac{1}{4!}\overbrace{\Psi_{\mf a}\epsilon^{\mf a\mf b_{1}\mf b_{2}\mf b_{3}\mf b_{4}}}^{\Psi^{\mf b_{1}\mf b_{2}\mf b_{3}\mf b_{4}}}\ket{\tief{\mf b_{1}\mf b_{2}\mf b_{3}\mf b_{4}}}}_{{\rm chiral}}+\label{U5-exp}\\
 &  & \!\!+\!\underbrace{\tfrac{1}{5!}\overbrace{\Psi_{+}\epsilon^{\mf b_{1}\mf b_{2}\mf b_{3}\mf b_{4}\mf b_{5}}}^{\Psi^{\mf b_{1}\mf b_{2}\mf b_{3}\mf b_{4}\mf b_{5}}}\ket{\tief{\mf b_{1}\mf b_{2}\mf b_{3}\mf b_{4}\mf b_{5}}}\!+\!\tfrac{1}{3!}\overbrace{\tfrac{1}{2}\Psi_{\mf a_{1}\mf a_{2}}\epsilon^{\mf a_{1}\mf a_{2}\mf b_{1}\mf b_{2}\mf b_{3}}}^{\Psi^{\mf b_{1}\mf b_{2}\mf b_{3}\mf b_{4}}}\ket{\tief{\mf b_{1}\mf b_{2}\mf b_{3}}}\!+\!\Psi^{\mf a}\ket{\tief{\mf a}}}_{{\rm antichiral}}\qquad\nonumber 
\end{eqnarray}
Chiral SO(10) Weyl spinors $\rho^{\alpha}$ can therefore be U(5)-covariantly
parametrized by a U(5)-singlet $\rho^{+}$, a U(5)-bivector $\rho^{\mf a_{1}\mf a_{2}}$
(antisymmetric with ${\scriptstyle \mf a_{1},\mf a_{2}\in\{1,\ldots,5\}}$,
i.e. 10 components) and a U(5) covector $\rho_{\mf a}$ (with 5 components).
The pure spinor constraint $(\lambda\gamma^{a}\lambda)=0$ turns into
$\bra{\lambda}Cb^{\mf a}\ket{\lambda}=0$ and $\bra{\lambda}Cb_{\mf a}^{\dagger}\ket{\lambda}=0$
where $C$ is the charge conjugation matrix. The first one turns out
to be a consequence of the second, while the second can be calculated
to be of the form (see up to a conventional sign \cite{Berkovits:2000fe}
or again appendix D of \cite{Nh:2003cm}) 
\begin{equation}
\lambda^{+}\lambda_{\mf a}=\tfrac{1}{8}\epsilon_{\mf a\mf b_{1}\mf b_{2}\mf b_{3}\mf b_{4}}\lambda^{\mf b_{1}\mf b_{2}}\lambda^{\mf b_{3}\mf b_{4}}\label{U5ps-constraint}
\end{equation}
For $\lambda^{+}\neq0$ one can obviously solve for $\lambda_{\mf a}$.
A natural projection from the space of Weyl spinors $\rho^{\alpha}$
to pure Weyl spinors $\lambda^{\alpha}$ is then simply to replace
the general component $\rho_{\mf a}$ by the solution to the above
equation, i.e. $(\rho^{+},\rho^{\mf a_{1}\mf a_{2}},\rho_{\mf a})\mapsto(\rho^{+},\rho^{\mf a_{1}\mf a_{2}},\tfrac{1}{8\rho^{+}}\epsilon_{\mf a\mf b_{1}\mf b_{2}\mf b_{3}\mf b_{4}}\rho^{\mf b_{1}\mf b_{2}}\rho^{\mf b_{3}\mf b_{4}})$.
The claim is that there is some reference spinor $\bar{\chi}$, such
the previously discussed $P_{(\bar{\chi})}$ of equation (\ref{Pnoncov})
is precisely this projection: 
\begin{equation}
(\rho^{+},\rho^{\mf a_{1}\mf a_{2}},\rho_{\mf a})\stackrel{P_{(\bar{\chi})}}{\mapsto}P_{(\bar{\chi})}(\rho)\stackrel{!}{=}(\rho^{+},\rho^{\mf a_{1}\mf a_{2}},\tfrac{1}{8\rho^{+}}\epsilon_{\mf a\mf b_{1}\mf b_{2}\mf b_{3}\mf b_{4}}\rho^{\mf b_{1}\mf b_{2}}\rho^{\mf b_{3}\mf b_{4}})\equiv(\lambda^{+},\lambda^{\mf a_{1}\mf a_{2}},\lambda_{\mf a})\qquad\label{U5-proj}
\end{equation}
This should determine the reference spinor $\bar{\chi}_{\alpha}$.
As there appear no square roots in the above terms, it is reasonable
to assume that $\bar{\chi}_{\alpha}$ is a pure spinor and the general
form (\ref{Pnoncov}) of the projection reduces to (\ref{non-cov-integrated}),
i.e. 
\begin{equation}
P_{(\bar{\chi})}^{\alpha}(\rho)=\rho^{\alpha}-\tfrac{1}{4}\frac{(\rho\gamma^{a}\rho)(\gamma_{a}\bar{\chi})^{\alpha}}{(\rho\bar{\chi})}\quad({\rm pure}\:\bar{\chi})\label{non-cov-integrated-rep}
\end{equation}
  The antichiral spinor $\bar{\chi}_{\alpha}$ has in general only
a non-vanishing second line in the expansion (\ref{U5-exp}). Note
that the state (\ref{alternative-vac}) which appears in this expansion
is an alternative vacuum if one interchanges the role of annihilators
and creators, as it is annihilated by all $b_{\mf a}^{\dagger}$'s.
Being a vacuum it is automatically a pure spinor. So instead of making
a general ansatz for $\bar{\chi}_{\alpha}$, let us simply try this
alternative vacuum as reference spinor 
\begin{equation}
\left(\bar{\chi}_{+},\bar{\chi}_{\mf{aa}},\bar{\chi}^{\mf a}\right)=\left(1,0,0\right)\quad{\rm (ansatz})\label{ref-spinor-ansatz}
\end{equation}
The translation of (\ref{non-cov-integrated-rep}) into U(5) language
is then rather simple. First we note that due to the definitions (\ref{annihilator}),(\ref{creator})
we have $\gamma^{a}\otimes\gamma_{a}\to2b^{\mf a}\otimes b_{\mf a}^{\dagger}+2b_{\mf a}^{\dagger}\otimes b^{\mf a}$.
Applied to (\ref{non-cov-integrated-rep}), only the second term survives
with 
\begin{eqnarray}
b^{\mf a}\ket{\bar{\chi}} & = & \tfrac{1}{5!}\epsilon^{\mf b_{1}\mf b_{2}\mf b_{3}\mf b_{4}\mf b_{5}}b^{\mf a}\ket{\tief{\mf b_{1}\mf b_{2}\mf b_{3}\mf b_{4}\mf b_{5}}}=\tfrac{1}{4!}\epsilon^{\mf a\mf b_{1}\mf b_{2}\mf b_{3}\mf b_{4}}\ket{\tief{\mf b_{1}\mf b_{2}\mf b_{3}\mf b_{4}}}\\
\bra{\rho}Cb_{\mf a}^{\dagger}\ket{\rho} & = & 2\rho^{+}\rho_{\mf a}-\tfrac{1}{4}\epsilon_{\mf a\mf b_{1}\mf b_{2}\mf b_{3}\mf b_{4}}\rho^{\mf b_{1}\mf b_{2}}\rho^{\mf b_{3}\mf b_{4}}
\end{eqnarray}
So the only non-vanishing component of $b^{\mf a}\ket{\bar{\chi}}$
is $(b^{\mf a}\bar{\chi})_{\mf c}=\delta_{\mf c}^{\mf a}$ while $(b^{\mf a}\bar{\chi})^{+}=(b^{\mf a}\bar{\chi})^{\mf{cd}}=0$.
The projection (\ref{non-cov-integrated-rep}) thus becomes
\begin{eqnarray}
P_{(\bar{\chi})}^{+}(\rho) & = & \rho^{+}\\
P_{(\bar{\chi})}^{\mf{cd}}(\rho) & = & \rho^{\mf{cd}}\\
P_{(\bar{\chi})\mf c}(\rho) & = & \rho_{\mf c}-\tfrac{1}{2}\frac{\bra{\rho}Cb_{\mf a}^{\dagger}\ket{\rho}(b^{\mf a}\bar{\chi})_{\mf c}}{\rho^{+}}=\\
 & = & \rho_{\mf c}-\tfrac{1}{2}\frac{\left(2\rho^{+}\rho_{\mf a}-\tfrac{1}{4}\epsilon_{\mf a\mf b_{1}\mf b_{2}\mf b_{3}\mf b_{4}}\rho^{\mf b_{1}\mf b_{2}}\rho^{\mf b_{3}\mf b_{4}}\right)\delta_{\mf c}^{\mf a}}{\rho^{+}}=\\
 & = & \tfrac{1}{8\rho^{+}}\epsilon_{\mf c\mf b_{1}\mf b_{2}\mf b_{3}\mf b_{4}}\rho^{\mf b_{1}\mf b_{2}}\rho^{\mf b_{3}\mf b_{4}}
\end{eqnarray}
This is indeed the projection that we suggested in (\ref{U5-proj}).

Let us define for double indices the partial derivative such that
the variation comes with a factor $\frac{1}{2}$ which takes into
account the antisymmetry: 
\begin{equation}
\delta=\delta\rho^{+}\partial_{\rho^{+}}+\tfrac{1}{2}\delta\rho^{\mf{ab}}\partial_{\rho^{\mf{ab}}}+\delta\rho_{\mf a}\partial_{\rho_{\mf a}}\label{U5-variation}
\end{equation}
In other words $\partial_{\rho^{\mf{ab}}}$ is twice the naive partial
derivative. Let us think of the spinor index ${\scriptstyle \alpha}$
(in the U(5) basis) to be a collective index $\hoch{\alpha}\in\{\hoch{\mf +},\hoch{\mf a_{1}\mf a_{2}},\tief{\mf a}\}$.
The projection matrix $\Pi_{\bot}$ then becomes 
\begin{eqnarray}
\Pi_{\bot(\bar{\chi})}^{\quad\:\alpha}\tief{\beta}(\rho,\bar{\rho}) & \!\!\!\!= & \!\!\!\!\partiell{(P_{(\bar{\chi})}^{+}(\rho),P_{(\bar{\chi})}^{\mf a_{1}\mf a_{2}}(\rho),P_{(\bar{\chi})\mf a}(\rho))}{(\rho^{+},\rho^{\mf b_{1}\mf b_{2}},\rho_{\mf b})}=\\
 & \!\!\!\!= & \!\!\!\!\left(\!\!\begin{array}{ccc}
1 & 0 & 0\\
0 & \delta_{\mf b_{1}}^{\mf a_{1}}\delta_{\mf b_{2}}^{\mf a_{2}}-\delta_{\mf b_{2}}^{\mf a_{1}}\delta_{\mf b_{1}}^{\mf a_{2}} & 0\\
-\tfrac{1}{8(\rho^{+})^{2}}\epsilon_{\mf a\mf c_{1}\mf c_{2}\mf c_{3}\mf c_{4}}\rho^{\mf c_{1}\mf c_{2}}\rho^{\mf c_{3}\mf c_{4}} & \tfrac{1}{2\rho^{+}}\epsilon_{\mf a\mf b_{1}\mf b_{2}\mf c_{1}\mf c_{2}}\rho^{\mf c_{1}\mf c_{2}} & 0
\end{array}\!\!\right)\quad\label{U5-proj-matrix}
\end{eqnarray}
We will come back to this U(5) form of the projection matrix in subsection
\ref{sec:ghost-action-3} on page \pageref{sec:ghost-action-3}.

\section{Ghost action of the pure spinor string}

\label{sec:ghost-action}We are now ready to apply some of the mathematical
insight to the pure spinor string. Remember that in the introduction
the transpose of the linearized projection was claimed to project
the antighost field of the pure spinor string to its gauge invariant
part. In the following first subsection we will quickly recall the
ghost action of the so-called non-minimal formalism as a functional
of constrained fields and discuss the constrained variation and the
corresponding gauge transformations. The known projector to a gauge
invariant part of the antighost will be extended to the so-called
non-minimal fields. After that, in subsection \ref{sec:ghost-action-2},
we will replace the constrained variables by projections of unconstrained
variables and discuss the variation and the gauge symmetries of the
resulting constraint-free ghost action. And in subsection \ref{sec:ghost-action-3},
we will recall the minimal ghost action in the U(5) parametrization
and quickly review how in this formulation the antighosts automatically
combine to gauge invariant combinations. We will then provide the
explicit reference spinor $\bar{\chi}$ for which the non-covariant
projector $\bar{\Pi}_{\bot(\bar{\chi})}$ yields precisely these expressions.

\subsection{With constrained variables}

\label{sec:ghost-action-1}

The ghost action (left-moving sector) in the non-minimal formalism
\cite{Berkovits:2005bt} of Berkovits' pure spinor string theory is
given by 
\begin{eqnarray}
S_{{\rm gh}}[\lambda,\omega_{z},\bar{\lambda},\bar{\omega}_{z},\bs r,\bs s_{z}] & = & \int d^{2}z\:\Bigl(\bar{\partial}\lambda^{\alpha}\omega_{z\alpha}+\bar{\partial}\bar{\lambda}_{\alpha}\bar{\omega}_{z}^{\alpha}+\bar{\partial}\bs r_{\alpha}\bs s_{z}^{\alpha}\Bigr)\label{Sgh-orig}
\end{eqnarray}
together with the constraints 
\begin{equation}
(\lambda\gamma^{a}\lambda)=(\bar{\lambda}\gamma^{a}\bar{\lambda})=(\bar{\lambda}\gamma^{a}\bs r)=0\label{constraints}
\end{equation}
The variations consistent with these constraints accordingly have
to obey
\begin{equation}
(\delta\lambda\gamma^{a}\lambda)=(\delta\bar{\lambda}\gamma^{a}\bar{\lambda})=(\delta\bar{\lambda}\gamma^{a}\bs r)+(\bar{\lambda}\gamma^{a}\delta\bs r)=0\label{constraintsOnVariations}
\end{equation}
The constraints \eqref{constraints} generate via the Poisson bracket
gauge transformations of the form\vspace{-.3cm} 
\begin{eqnarray}
\delta_{(\mu)}\omega_{z\alpha} & \equiv & \mu_{za}(\gamma^{a}\lambda)_{\alpha}\label{gaugetrafoomegaOnsurf}\\
\delta_{(\bar{\mu},\sigma)}\bar{\omega}_{z}^{\alpha} & \equiv & \bar{\mu}_{z}^{a}(\gamma_{a}\bar{\lambda})^{\alpha}-\bs{\sigma}_{z}^{a}(\gamma_{a}\bs r)^{\alpha}\label{gaugetrafoomegabarOnsurf}\\
\delta_{(\sigma)}\bs s_{z}^{\alpha} & \equiv & \bs{\sigma}_{z}^{a}(\gamma_{a}\bar{\lambda})^{\alpha}\label{gaugetrafosOnsurf}
\end{eqnarray}
where $\mu_{za},\bar{\mu}_{z}^{a}$ are some even and $\bs{\sigma}_{z}^{a}$
are some odd gauge parameters. In equations (6) and (7) of \cite{Oda:2004bg}
within the minimal formalism (in the absence of $\bar{\lambda}_{\alpha},\bar{\omega}_{z}^{\alpha},\bs r_{\alpha}$
and $\bs s_{z}^{\alpha}$) the authors presented a linear projection
of the antighost $\omega_{z\alpha}$ to its gauge invariant part $\tilde{\omega}_{z\alpha}$,
using some reference spinor (see also (14) and (18) in \cite{Oda:2005sd}
or more recently (2.12) and (2.15) in \cite{Oda:2007ak}). Its covariantized
version (where $\bar{\lambda}_{\alpha}$ plays the role of the reference
spinor) is given by 
\begin{equation}
\tilde{\omega}_{z\alpha}\equiv(\bar{\Pi}_{\bot}\omega_{z})_{\alpha}=\omega_{z\alpha}-\tfrac{1}{2}\frac{(\gamma^{a}\lambda)_{\alpha}(\bar{\lambda}\gamma_{a}\omega_{z})}{(\lambda\bar{\lambda})}\label{gaugeinvOmegaOnsurf}
\end{equation}
and was presented in equations (2.9) and (2.11) of \cite{Berkovits:2010zz}.
The gauge transformation of $\bs s_{z}^{\alpha}$ in \eqref{gaugetrafosOnsurf}
is of the same type as the one of $\omega_{z\alpha}$, just that the
role of $\lambda^{\alpha}$ and $\bar{\lambda}_{\alpha}$ is interchanged.
It is therefore easy to guess also a projection to a gauge invariant
part of $\bs s_{z}^{\alpha}$:

\begin{equation}
\tilde{\bs s}_{z}^{\alpha}\equiv(\Pi_{\bot}\bs s_{z})^{\alpha}=\bs s_{z}^{\alpha}-\tfrac{1}{2}\frac{(\gamma_{a}\bar{\lambda})^{\alpha}(\lambda\gamma^{a}\bs s_{z})}{(\lambda\bar{\lambda})}\label{gaugeinvSonSurf}
\end{equation}
For $\bar{\omega}_{z}^{\alpha}$ at least the $\delta_{(\sigma)}$-part
of \eqref{gaugetrafoomegabarOnsurf} is of a slightly different form,
so that the same naive guess as above does not work. However, having
in mind the BRST transformation $\es\bs s_{z}^{\alpha}=\bar{\omega}_{z}^{\alpha}$
and the fact that the BRST differential forms together with the gauge
transformations a closed algebra%
\footnote{\label{fn:BRSTnonminimal}We have not presented the complete action
of the pure spinor string and neither the BRST transformations of
all its fields, but those of the non-minimal sector are given by \cite{Berkovits:2005bt}
\begin{eqnarray*}
\es\bs s_{z}^{\alpha} & = & \bar{\omega}_{z}^{\alpha},\quad\es\bar{\omega}_{z}^{\alpha}=0,\qquad\es\bar{\lambda}_{\alpha}=\bs r_{\alpha},\quad\es\bs r_{\alpha}=0
\end{eqnarray*}
These BRST transformations commute with the $\delta_{(\sigma)}$ gauge
transformations in \eqref{gaugetrafoomegabarOnsurf}-\eqref{gaugetrafosOnsurf}.
\[
[\es,\delta_{(\sigma)}]=0
\]
The commutator with $\delta_{(\bar{\mu})}$ instead is non-vanishing
but produces another $\delta_{(\sigma)}$-transformation:
\begin{eqnarray*}
[\es,\delta_{(\bar{\mu})}]\bar{\omega}_{z}^{\alpha} & = & \bar{\mu}_{z}^{a}(\gamma_{a}\bs r)^{\alpha}\\
{}[\es,\delta_{(\bar{\mu})}]\bs s_{z}^{\alpha} & = & -\bar{\mu}_{z}^{a}(\gamma_{a}\bar{\lambda})^{\alpha}
\end{eqnarray*}
So $[\delta_{(\bar{\mu})},\es]=\bs{\delta}_{(\tilde{\sigma})}$ with
$\tilde{\sigma}_{z}^{a}=\bar{\mu}_{z}^{a}$. The parameter has changed
parity, but that does not change invariance properties at linearized
level. $\quad\fussend$ %
}, a natural guess for the gauge invariant part of $\bar{\omega}_{z\alpha}$
is simply the BRST transformation of $\tilde{\bs s}_{z}^{\alpha}$
\begin{eqnarray}
\tilde{\bar{\omega}}_{z}^{\alpha} & \equiv & \es\tilde{\bs s}_{z}^{\alpha}=\label{tildeomegabarISsOftildes}\\
 & = & (\Pi_{\bot}\bar{\omega}_{z})^{\alpha}+\left(-(\gamma_{a}\bs r)^{\alpha}+(\lambda\bs r)\frac{(\gamma_{a}\bar{\lambda})^{\alpha}}{(\lambda\bar{\lambda})}\right)\frac{(\lambda\gamma^{a}\bs s_{z})}{2(\lambda\bar{\lambda})}=\\
 & = & (\Pi_{\bot}\bar{\omega}_{z})^{\alpha}+\left(-(\gamma_{a}\bs r)^{\alpha}+(\gamma_{a}\bar{\Pi}_{\lambda}\bs r)^{\alpha}\right)\frac{(\lambda\gamma^{a}\bs s_{z})}{2(\lambda\bar{\lambda})}
\end{eqnarray}
Using the identity $(\gamma_{a}\bar{\Pi}_{\lambda})^{\alpha\beta}(\lambda\gamma^{a})_{\gamma}=(\Pi_{\Vert}\gamma_{a})^{\alpha\beta}(\lambda\gamma^{a})_{\gamma}$
of equation \eqref{PiparZisZPibarlam2}, we obtain the linear projection%
\footnote{\label{fn:omegatildeProjprop}The expression for $\tilde{\bar{\omega}}_{z}$
in (\ref{gaugeinvOmegabarOnsurf}) is not everywhere linear in $\bar{\omega}_{z}^{\alpha}$,
but instead some terms are linear in $\es_{z}$. It should thus be
understood as a linear projection from the variables $(\bs s_{z}^{\alpha},\bar{\omega}_{z}^{\alpha})$
to $(\tilde{\bs s}_{z}^{\alpha},\tilde{\bar{\omega}}_{z}^{\alpha})$,
which indeed obeys the projection-property in addition to being gauge-invariant.
The projection property on the constraint surface is inherited from
$\Pi_{\bot}$ (Because of $\Pi_{\bot}^{2}=\Pi_{\bot}$ we have $(\es\Pi_{\bot})\Pi_{\bot}+\Pi_{\bot}(\es\Pi_{\bot})=\es\Pi_{\bot}$):\vspace{-.2cm}
\begin{eqnarray*}
\tilde{\bar{\omega}}_{z} & \stackrel{\eqref{tildeomegabarISsOftildes}}{=} & \Pi_{\bot}\bar{\omega}_{z}+(\es\Pi_{\bot})\bs s_{z}\\
\tilde{\tilde{\bar{\omega}}}_{z} & = & \Pi\tilde{\bar{\omega}}_{z}+(\es\Pi)\tilde{\bs s}_{z}=\\
 & = & \Pi\left(\Pi\bar{\omega}_{z}+(\es\Pi)\bs s_{z}\right)+(\es\Pi)\Pi\bs s_{z}=\\
 & \ous{\Pi^{2}=\Pi}={(\es\Pi)\Pi+\Pi(\es\Pi)=\es\Pi} & \qquad\bar{\omega}_{z}+(\es\Pi)\bs s_{z}=\tilde{\bar{\omega}}_{z}\qquad\fussend
\end{eqnarray*}
}
\begin{equation}
\tilde{\bar{\omega}}_{z}^{\alpha}=(\Pi_{\bot}\bar{\omega}_{z})^{\alpha}-(\Pi_{\bot}\gamma_{a}\bs r)^{\alpha}\frac{(\lambda\gamma^{a}\bs s_{z})}{2(\lambda\bar{\lambda})}\label{gaugeinvOmegabarOnsurf}
\end{equation}
Gauge invariance of this expression now follows from gauge invariance
of $\tilde{\bs s}_{z}^{\alpha}$ and the fact that the BRST differential
builds a closed algebra with the gauge transformations according to
footnote \ref{fn:BRSTnonminimal}:
\begin{eqnarray}
\delta_{(\sigma)}\tilde{\bar{\omega}}_{z}^{\alpha} & = & \delta_{(\sigma)}\es\tilde{\bs s}_{z}^{\alpha}\stackrel{{\rm fn\,}\eqref{fn:BRSTnonminimal}}{=}\es\underbrace{\delta_{(\sigma)}\tilde{\bs s}_{z}^{\alpha}}_{=0}=0\\
\delta_{(\bar{\mu})}\tilde{\bar{\omega}}_{z}^{\alpha} & = & \delta_{(\bar{\mu})}\es\tilde{\bs s}_{z}^{\alpha}\stackrel{{\rm fn\,}\eqref{fn:BRSTnonminimal}}{=}\es\underbrace{\delta_{(\bar{\mu})}\tilde{\bs s}_{z}^{\alpha}}_{=0}+\bs{\delta}_{(\tilde{\sigma})}\tilde{\bs s}_{z}^{\alpha}=0\qquad
\end{eqnarray}
 It can also easily been checked by direct calculation%
\footnote{Let us check the gauge invariance of (\ref{gaugeinvOmegabarOnsurf})
by direct calculation\vspace{-.5cm}
\begin{eqnarray*}
\delta_{(\bar{\mu},\sigma)}\tilde{\bar{\omega}}_{z}^{\alpha} & = & \bar{\mu}_{z}^{a}(\underbrace{\Pi_{\bot}\gamma_{a}\bar{\lambda}}_{=0})^{\alpha}-\bs{\sigma}_{z}^{a}(\Pi_{\bot}\gamma_{a}\bs r)^{\alpha}-(\Pi_{\bot}\gamma_{a}\bs r)^{\alpha}\frac{\bs{\sigma}_{z}^{b}(\lambda\overbrace{\gamma^{a}\gamma_{b}}^{-\gamma_{b}\gamma^{a}\lqn{{\scriptstyle +2\delta_{b}^{a}}}}\bar{\lambda})}{2(\lambda\bar{\lambda})}=\\
 & = & (\Pi_{\bot}\gamma_{a}\bs r)^{\alpha}\frac{\bs{\sigma}_{z}^{b}(\lambda\gamma_{b}\gamma^{a}\bar{\lambda})}{2(\lambda\bar{\lambda})}
\end{eqnarray*}
Using the Fierz identity $(\gamma_{a}\bs r)^{\alpha}(\gamma^{a}\bar{\lambda})^{\beta}=-(\bar{\lambda}\gamma_{a}\bs r)\gamma^{a\,\alpha\beta}-(\gamma_{a}\bar{\lambda})^{\alpha}(\gamma^{a}\bs r)^{\beta}$
and remembering the constraint $(\bar{\lambda}\gamma_{a}\bs r)=0$,
the above expression becomes 
\begin{eqnarray*}
\delta_{(\bar{\mu},\sigma)}\tilde{\bar{\omega}}_{z}^{\alpha} & = & (\underbrace{\Pi_{\bot}\gamma_{a}\bar{\lambda}}_{=0})^{\alpha}\frac{\bs{\sigma}_{z}^{b}(\lambda\gamma_{b}\gamma^{a}\bs r)}{2(\lambda\bar{\lambda})}=0\qquad\fussend
\end{eqnarray*}
}. To our knowledge the gauge invariant projections $\tilde{\es}_{z}^{\alpha}$
and $\tilde{\bar{\omega}}_{z}^{\alpha}$ for the non-minimal variables
had not yet been mentioned in the literature. Note that $\tilde{\omega}_{z\alpha}$,
$\tilde{\bs s}_{z}^{\alpha}$ and $\tilde{\bar{\omega}}_{z}^{\alpha}$
are gauge invariant only up to the constraints (\ref{constraints}).
We will remove this restriction a bit later in this section. Let us
now elaborate a bit further on the variation of the constrained variables
and recover the gauge invariant expressions in the equations of motion.
To this end, let us first note that by using our projection maps,
the constraints \eqref{constraints} can equivalently be written as
\begin{equation}
\lambda^{\alpha}=P_{(f)}^{\alpha}(\lambda,\bar{\lambda}),\quad\bar{\lambda}_{\alpha}=P_{(f)}^{\alpha}(\lambda,\bar{\lambda}),\quad\bs r_{\alpha}=\bar{\Pi}_{\bot\alpha}\hoch{\beta}\bs r_{\beta}\label{constraints2}
\end{equation}
Varying these constraints on both sides leads for the ghost $\lambda^{\alpha}$
to $\delta\lambda^{\alpha}=(\Pi_{\bot}\delta\lambda)^{\alpha}$ which
was already given in equation (\ref{deltalambda}). For $\bar{\lambda}_{\alpha}$
one obtains the complex conjugate relation, while for the $\bs r_{\alpha}$-variation
of the last constraint in \eqref{constraints} we need the variation
of $\Pi_{\bot\alpha}\hoch{\beta}$ given in (\ref{deltaPiperp}) (or
actually its complex conjugate). Using $(\bar{\lambda}\gamma_{a}\bs r)=0$
this yields
\begin{equation}
\delta\bs r_{\alpha}=\bar{\Pi}_{\bot\alpha}\hoch{\beta}\delta\bs r_{\beta}-\tfrac{(\gamma^{a}\lambda)_{\alpha}}{2(\lambda\bar{\lambda})}(\underbrace{\bs r\Pi_{\bot}}_{\bs r}\gamma_{a})^{\gamma}\underbrace{\delta\bar{\lambda}_{\gamma}}_{(\bar{\Pi}_{\bot}\delta\bar{\lambda})_{\gamma}}
\end{equation}
So altogether we end up with the following relations for the variations
\begin{equation}
\delta\lambda^{\alpha}=\Pi_{\bot\beta}^{\alpha}\delta\lambda^{\beta},\quad\delta\bar{\lambda}_{\alpha}=\bar{\Pi}_{\bot\alpha}\hoch{\beta}\delta\bar{\lambda}_{\beta},\quad\delta\bs r_{\alpha}=\bar{\Pi}_{\bot\alpha}\hoch{\beta}\delta\bs r_{\beta}-\tfrac{(\gamma^{a}\lambda)_{\alpha}}{2(\lambda\bar{\lambda})}(\bs r\gamma_{a}\bar{\Pi}_{\bot})^{\gamma}\delta\bar{\lambda}_{\gamma}\label{constraintsOnVariations2}
\end{equation}
These relations are on the one side just equivalent to the constraints
in (\ref{constraintsOnVariations}) but on the other hand they enable
to rewrite the variations in terms of projections that extract the
independent degrees of freedom of the variation. In particular the
variation of a function $f(\lambda,\bar{\lambda},\bs r)$ (instead
of a functional) can be rewritten as follows 
\begin{eqnarray}
\delta f(\lambda,\bar{\lambda},\bs r)\!\!\! & = & \!\!\left(\delta\lambda^{\alpha}\partial_{\lambda^{\alpha}}+\delta\bar{\lambda}_{\alpha}\partial_{\bar{\lambda}_{\alpha}}+\delta\bs r_{\alpha}\partial_{\bs r_{\alpha}}\right)f(\lambda,\bar{\lambda},\bs r)=\\
 & \stackrel{\mbox{\small\eqref{constraintsOnVariations2}}}{=} & \!\!\biggl\{\!\delta\lambda^{\alpha}\bar{\Pi}_{\bot\alpha}\hoch{\beta}\partial_{\lambda^{\beta}}\!+\!\delta\bar{\lambda}_{\alpha}\!\Bigl(\!\Pi_{\bot\beta}^{\alpha}\partial_{\bar{\lambda}_{\beta}}\!-\!(\underbrace{\bs r\gamma_{a}\bar{\Pi}_{\bot}}_{\Pi_{\bot}\gamma_{a}\bs r})^{\alpha}\tfrac{(\lambda\gamma^{a}\partial_{\bs r})}{2(\lambda\bar{\lambda})}\!\Bigr)\!+\qquad\nonumber \\
 &  & +\delta\bs r_{\alpha}\Pi_{\bot\beta}^{\alpha}\partial_{\bs r_{\beta}}\biggr\} f(\lambda,\bar{\lambda},\bs r)
\end{eqnarray}
and then naturally defines \textbf{covariant derivatives} in the sense
in which it was already discussed on page \pageref{deltalambda}:
\begin{equation}
D_{\lambda^{\alpha}}\equiv\bar{\Pi}_{\bot\alpha}\hoch{\beta}\partial_{\lambda^{\beta}}\:,\quad D_{\bar{\lambda}_{\alpha}}\equiv\Pi_{\bot\beta}^{\alpha}\partial_{\bar{\lambda}_{\beta}}-(\Pi_{\bot}\gamma_{a}\bs r)^{\alpha}\tfrac{(\lambda\gamma^{a}\partial_{\bs r})}{2(\lambda\bar{\lambda})}\:,\quad D_{\bs r_{\beta}}\equiv\Pi_{\bot\beta}^{\alpha}\partial_{\bs r_{\beta}}\qquad\label{covariant-derivatives}
\end{equation}
They annihilate the constraints (\ref{constraints}) and remarkably
reproduce the projector that we have proposed for $\bar{\omega}_{z}^{\alpha}$
in (\ref{gaugeinvOmegabarOnsurf}). The same insertions of projection
matrices using (\ref{constraintsOnVariations2}) should be done in
the constrained variation of functionals, in particular of the action
(\ref{Sgh-orig}):
\begin{eqnarray}
\lqn{\delta S_{{\rm gh}}[\lambda,\omega_{z},\bar{\lambda},\bar{\omega}_{z},\bs r,\bs s_{z}]=}\nonumber \\
 & = & \!\!\!\!\int d^{2}z\:\Bigl[\bar{\partial}\lambda^{\alpha}\delta\omega_{z\alpha}+\bar{\partial}\bar{\lambda}_{\alpha}\delta\bar{\omega}_{z}^{\alpha}+\bar{\partial}\bs r_{\alpha}\delta\bs s_{z}^{\alpha}+\nonumber \\
 &  & \!\!\!\!-\delta\lambda^{\alpha}\bar{\partial}\omega_{z\alpha}-\delta\bar{\lambda}_{\alpha}\bar{\partial}\bar{\omega}_{z}^{\alpha}-\delta\bs r_{\alpha}\bar{\partial}\bs s_{z}^{\alpha}+\nonumber \\
 &  & \!\!\!\!+\bar{\partial}\left(\delta\lambda^{\alpha}\omega_{z\alpha}+\delta\bar{\lambda}_{\alpha}\bar{\omega}_{z}^{\alpha}+\delta\bs r_{\alpha}\bs s_{z}^{\alpha}\right)\Bigr]=\label{deltaSorig}\\
 & \stackrel{\eqref{constraintsOnVariations2}}{=} & \!\!\!\!\int d^{2}z\:\Bigl[\bar{\partial}\lambda^{\alpha}\delta\omega_{z\alpha}+\bar{\partial}\bar{\lambda}_{\alpha}\delta\bar{\omega}_{z}^{\alpha}+\bar{\partial}\bs r_{\alpha}\delta\bs s_{z}^{\alpha}+\nonumber \\
 &  & \!\!\!\!-\delta\lambda^{\alpha}\bar{\Pi}_{\bot\alpha}\hoch{\beta}\bar{\partial}\omega_{z\beta}\!-\!\delta\bar{\lambda}_{\alpha}\!\left(\!\Pi_{\bot\beta}^{\alpha}\bar{\partial}\bar{\omega}_{z}^{\beta}\!-\!(\Pi_{\bot}\gamma_{a}\bs r)^{\alpha}\tfrac{(\lambda\gamma^{a}\bar{\partial}\bs s_{z})}{2(\lambda\bar{\lambda})}\!\right)\!-\!\delta\bs r_{\alpha}\Pi_{\bot\beta}^{\alpha}\bar{\partial}\bs s_{z}^{\beta}\!+\nonumber \\
 &  & \!\!\!\!\!\!+\bar{\partial}\!\left(\!\delta\lambda^{\alpha}\bar{\Pi}_{\bot\alpha}\hoch{\beta}\omega_{z\beta}\!+\!\delta\bar{\lambda}_{\alpha}\!\left(\!\Pi_{\bot\beta}^{\alpha}\bar{\omega}_{z}^{\beta}\!-\!(\Pi_{\bot}\gamma_{a}\bs r)^{\alpha}\tfrac{(\lambda\gamma^{a}\bs s_{z})}{2(\lambda\bar{\lambda})}\!\right)\!+\!\delta\bs r_{\alpha}\Pi_{\bot\beta}^{\alpha}\bs s_{z}^{\beta}\!\right)\!\Bigr]\qquad\label{deltaSorig2}
\end{eqnarray}
The equations of motion from the constrained variation are thus (we
suppress the worldsheet arguments of the functional derivative) 
\begin{eqnarray}
0 & \stackrel{{\rm eom}}{=} & \left(\drek{\funktl{\omega_{z\alpha}}S_{{\rm gh}}}{\funktl{\bar{\omega}_{z}^{\alpha}}S_{{\rm gh}}}{\funktl{\bs s_{z}^{\alpha}}S_{{\rm gh}}}\right)=\left(\drek{\bar{\partial}\lambda^{\alpha}}{\bar{\partial}\bar{\lambda}_{\alpha}}{-\bar{\partial}\bs r_{\alpha}}\right)\label{lambda-r-eoms}\\
0 & \stackrel{{\rm eom}}{=} & \left(\begin{array}{ccc}
\bar{\Pi}_{\bot\alpha}\hoch{\beta}\\
 & \Pi_{\bot\beta}^{\alpha} & -(\Pi_{\bot}\gamma_{a}\bs r)^{\alpha}\tfrac{(\lambda\gamma^{a})_{\beta}}{2(\lambda\bar{\lambda})}\\
 &  & \Pi_{\bot\beta}^{\alpha}
\end{array}\right)\left(\drek{\funktl{\lambda^{\beta}}S_{{\rm gh}}}{\funktl{\bar{\lambda}_{\beta}}S_{{\rm gh}}}{\funktl{\bs r_{\beta}}S_{{\rm gh}}}\right)=\label{lambarlamr-matrix}\\
 & = & -\left(\drek{\bar{\Pi}_{\bot\alpha}\hoch{\beta}\bar{\partial}\omega_{z\beta}}{\Pi_{\bot\beta}^{\alpha}\bar{\partial}\bar{\omega}_{z}^{\beta}-(\Pi_{\bot}\gamma_{a}\bs r)^{\alpha}\tfrac{(\lambda\gamma^{a}\bar{\partial}\bs s_{z})}{2(\lambda\bar{\lambda})}}{\Pi_{\bot\beta}^{\alpha}\bar{\partial}\bs s_{z}^{\beta}}\right)\label{omega-s-eoms}
\end{eqnarray}
In the same way the boundary conditions from the variational principle
(in the absence of an additional boundary action) become%
\footnote{In the presence of right-movers, the expressions would not be zero,
but coincide on the boundary with the corresponding projections for
the right-movers.$\quad\fussend$%
}
\begin{equation}
\bei{\bar{\Pi}_{\bot\alpha}\hoch{\beta}\omega_{z\beta}}{\partial\Sigma}\stackrel{{\rm bc}}{=}\bei{\left(\Pi_{\bot\beta}^{\alpha}\bar{\omega}_{z}^{\beta}-(\Pi_{\bot}\gamma_{a}\bs r)^{\alpha}\tfrac{(\lambda\gamma^{a}\bs s_{z})}{2(\lambda\bar{\lambda})}\right)}{\partial\Sigma}\stackrel{{\rm bc}}{=}\bei{\Pi_{\bot\beta}^{\alpha}\bs s_{z}^{\beta}}{\partial\Sigma}\stackrel{{\rm bc}}{=}0\quad
\end{equation}
The expressions in the equations of motion for the conjugate momenta
(\ref{omega-s-eoms}) already resemble the gauge invariant variables
(\ref{gaugeinvOmegaOnsurf}),(\ref{gaugeinvSonSurf}) and (\ref{gaugeinvOmegabarOnsurf}),
apart from the $\bar{\partial}$-derivative. But as the matrix $\Pi_{\bot}$
depends only on $\lambda$ and $\bar{\lambda}$, we have on-shell
(\ref{lambda-r-eoms}) $\bar{\partial}\Pi_{\bot}=0$ which allows
to pull in (\ref{omega-s-eoms}) the $\bar{\partial}$-derivative
to the front. The equations of motion then simply become 
\begin{equation}
\bar{\partial}\tilde{\omega}_{z\alpha}\stackrel{{\rm eom}}{=}\bar{\partial}\tilde{\bar{\omega}}_{z}^{\alpha}\stackrel{{\rm eom}}{=}\bar{\partial}\tilde{\bs s}_{z}^{\alpha}\stackrel{{\rm eom}}{=}0\label{omega-s-gaugeinv-eoms}
\end{equation}
Note finally that (not surprisingly) also the action itself can be
rewritten in terms of the gauge invariant variables. As the relations
(\ref{constraintsOnVariations2}) are for arbitrary variations, they
hold in particular for the worldsheet derivatives $\bar{\partial}$.
Using this fact in the action (\ref{Sgh-orig}), we directly obtain
the gauge invariant variables:
\begin{eqnarray}
\lqn{S_{{\rm gh}}[\lambda,\omega_{z},\bar{\lambda},\bar{\omega}_{z},\bs r,\bs s_{z}]=}\nonumber \\
 & \stackrel{\eqref{constraintsOnVariations2}}{=} & \int d^{2}z\:\Bigl[\bar{\partial}\lambda^{\alpha}\underbrace{\bar{\Pi}_{\bot\alpha}\hoch{\beta}\omega_{z\alpha}}_{\tilde{\omega}_{z\alpha}}+\bar{\partial}\bar{\lambda}_{\alpha}\underbrace{\left(\!\Pi_{\bot\beta}^{\alpha}\bar{\omega}_{z}^{\beta}\!-\!(\Pi_{\bot}\gamma_{a}\bs r)^{\alpha}\tfrac{(\lambda\gamma^{a}\bs s)}{2(\lambda\bar{\lambda})}\!\right)}_{\tilde{\bar{\omega}}_{z}^{\alpha}}+\quad\nonumber \\
 &  & +\bar{\partial}\bs r_{\alpha}\underbrace{\Pi_{\bot\beta}^{\alpha}\bs s_{z}^{\beta}}_{\tilde{\bs s}_{z}^{\alpha}}\:\Bigr]\label{actionInTermsOfGaugeInvVar}
\end{eqnarray}

\subsection{With unconstrained variables}

\label{sec:ghost-action-2}So far we have not really used our projection
$P_{(f)}$, but only its Jacobian matrix at the constraint surface
which was known already previously. However, having a projection at
hand, we can easily get rid of the constraint. The resulting action
is certainly not free and not very pleasant, but conceptionally it
is of some interest. In particular it will have additional gauge symmetries
that can be fixed to different constraints than those we started with.
So the idea to remove the pure spinor constraint is simply to replace
$\lambda^{\alpha}$ by $P_{(f)}^{\alpha}(\rho,\bar{\rho})$ of (\ref{ProjGeneral})
with an unconstrained spinor $\rho^{\alpha}$. 
\begin{equation}
\lambda^{\alpha}\equiv P_{(f)}^{\alpha}(\rho,\bar{\rho})
\end{equation}
We have a priori two natural options in order to remove also the constraint
$(\bar{\lambda}\gamma^{a}\bs r)=0$ on $\bs r_{\alpha}$. Either replace
it by 
\begin{equation}
\bs r_{\alpha}\stackrel{?}{\equiv}\bar{\Pi}_{(f)\bot}(\rho,\bar{\rho})_{\alpha}\hoch{\beta}\bs t_{\beta}
\end{equation}
 with some unconstrained $\bs t_{\beta}$ or by 
\begin{equation}
\bs r_{\alpha}\equiv(\bar{\Pi}_{\bot}\bs t)_{\alpha}\equiv\bar{\Pi}_{(f)\bot}\left(\lambda,\bar{\lambda}\right)_{\alpha}\hoch{\beta}\bs t_{\beta}\label{rt-rel}
\end{equation}
with of course still $\lambda^{\alpha}\equiv P_{(f)}^{\alpha}(\rho,\bar{\rho})$.
For both options, the result will be $\gamma$-orthogonal to $\bar{\lambda}$.
Only the latter choice is a projection by itself (with $\bar{\Pi}_{\bot}^{2}=\bar{\Pi}_{\bot}$),
but also the first is (part of) a projection, if not regarded independently,
but together with $\rho^{\alpha}\mapsto\lambda^{\alpha}\equiv P_{(f)}^{\alpha}(\rho,\bar{\rho}$)
, as we have discussed previously. One might want to prefer the first
one in order to get formally exactly the same equations of motion
for $\bs t_{\alpha}$ as for example for $\bar{\rho}_{\alpha}$. However,
the variation of $\bar{\Pi}_{(f)\bot}(\rho,\bar{\rho})$ can be quite
complicated to calculate, while the variation of $\bar{\Pi}_{\bot}\left(\lambda,\bar{\lambda}\right)=\one-\frac{(\gamma^{a}\lambda)\otimes(\gamma_{a}\bar{\lambda})}{2\left(\lambda\bar{\lambda}\right)}\quad\mbox{with }\lambda^{\alpha}\equiv P_{(f)}^{\alpha}(\rho,\bar{\rho})$
is relatively easy to perform and was provided in equation (\ref{deltaPiperp}).
We thus choose the second option and define a family of action functionals
of unconstrained variables via 
\begin{equation}
S_{{\rm gh}(f)}[\rho,\omega_{z},\bar{\rho},\bar{\omega}_{z},\bs t,\bs s_{z}]\equiv S_{{\rm gh}}[\underbrace{P_{(f)}(\rho,\bar{\rho})}_{\lambda},\omega_{z},\underbrace{\bar{P}_{(f)}(\rho,\bar{\rho})}_{\bar{\lambda}},\bar{\omega}_{z},\underbrace{\bar{\Pi}_{\bot}\bs t}_{\bs r},\bs s_{z}]
\end{equation}
Obviously the antighost gauge symmetries of the original action (\ref{gaugetrafoomegaOnsurf})-(\ref{gaugetrafosOnsurf})
will still be \textbf{gauge transformations} $\delta_{(\mu,\bar{\mu},\sigma)}$
of the new action $S_{{\rm gh(f)}}$. However, we expect new gauge
symmetries, namely all transformations of $\rho,\bar{\rho}$ and $\bs t$
that leave $\lambda,\bar{\lambda}$ and $\bs r$ unchanged. So let
us express the variation of the latter in terms of the former. For
$\delta\bs r$ we need the variation of the projection matrix $\delta\Pi_{\bot}$
given in (\ref{deltaPiperp}) (or better its complex conjugate) together
with (\ref{rt-rel}). This yields%
\footnote{If the variables on the righthand side of (\ref{delta-r}) are also
on the constraint surface, i.e. $\bs t_{\alpha}=\bs r_{\alpha}$,$\delta\bs t_{\alpha}=\delta\bs r_{\alpha}$,$\rho^{\alpha}=\lambda^{\alpha}$,$\delta\rho^{\alpha}=\delta\lambda^{\alpha}$,
it reduces to the constraint on $\delta\bs r$ given in (\ref{constraintsOnVariations2}).$\quad\fussend$%
} 
\begin{equation}
\delta\bs r_{\alpha}\equiv(\bar{\Pi}_{\bot}\delta\bs t)_{\alpha}-\frac{(\bs t\gamma_{a}\bar{\lambda})}{2(\lambda\bar{\lambda})}(\bar{\Pi}_{\bot}\gamma^{a}\delta\lambda)_{\alpha}-\frac{(\gamma^{a}\lambda)_{\alpha}}{2(\lambda\bar{\lambda})}(\underbrace{\bs t\Pi_{\bot}}_{\bs r}\gamma_{a}\delta\bar{\lambda})\label{delta-r}
\end{equation}
Using $\delta\lambda=\Pi_{(f)\bot}(\rho,\bar{\rho})\delta\rho+\rest_{(f)\bot}(\rho,\bar{\rho})\delta\bar{\rho}$
(\ref{bigMatrix}), the constrained variation of $\delta\lambda$
can be further expressed in terms of the free variation $\delta\rho$.
So altogether this yields in matrix notation 
\begin{equation}
\left(\!\!\drek{\delta\lambda}{\delta\bar{\lambda}}{\delta\bs r}\!\!\right)\!=\!\left(\begin{array}{ccc}
\Pi_{(f)\bot}(\rho,\bar{\rho}) & \rest_{(f)\bot}(\rho,\bar{\rho}) & 0\\
\bar{\rest}_{(f)\bot}(\rho,\bar{\rho}) & \bar{\Pi}_{(f)\bot}(\rho,\bar{\rho}) & 0\\
\!\!\!\!\!\biggl\{\!\!\!\!\zwek{-(\bs t\gamma_{a}\bar{\lambda})\frac{(\bar{\Pi}_{\bot}\gamma^{a}\Pi_{(f)\bot}(\rho,\bar{\rho}))}{2(\lambda\bar{\lambda})}+}{-\frac{(\gamma^{a}\lambda)\otimes(\bs t\Pi_{\bot}\gamma_{a}\bar{\rest}_{(f)\bot}(\rho,\bar{\rho}))}{2(\lambda\bar{\lambda})}}\!\!\!\!\biggr\}\!\! & \!\!\biggl\{\!\!\!\!\zwek{-(\bs t\gamma_{a}\bar{\lambda})\frac{(\bar{\Pi}_{\bot}\gamma^{a}\rest_{(f)\bot}(\rho,\bar{\rho}))}{2(\lambda\bar{\lambda})}+}{\,-\frac{(\gamma^{a}\lambda)\otimes(\bs t\Pi_{\bot}\gamma_{a}\bar{\Pi}_{(f)\bot}(\rho,\bar{\rho}))}{2(\lambda\bar{\lambda})}}\!\!\!\!\biggr\}\!\!\!\! & \bar{\Pi}_{\bot}
\end{array}\!\!\!\!\right)\!\!\!\left(\!\!\drek{\delta\rho}{\delta\bar{\rho}}{\delta\bs t}\!\!\right)\label{delta-lam-r}
\end{equation}
where the terms in the curly brackets where broken into two lines
just because of lack of space and thus have to be understood as being
summands in the same row of the block-matrix. Also for place reasons
$\lambda$ was used like many times before as a place holder for $P_{(f)}(\rho,\bar{\rho})$
and also the implicit dependence of $\Pi_{\bot}\equiv\Pi_{(f)\bot}(\lambda,\bar{\lambda})$
on $\lambda$ and thus $\rho$ is not indicated above. 

Equation (\ref{delta-lam-r}) is valid for general variations, but
as mentioned above we are looking for \textbf{new gauge transformations
}$\delta_{(\nu,\bar{\nu},\tau)}$ of $\rho,\bar{\rho}$ and $\bs t$
which lead to a vanishing variation of $\lambda,\bar{\lambda}$ and
$\bs r$ on the left-hand side of\textbf{ }(\ref{delta-lam-r}). On
the constraint surface $\rho=\lambda$ the matrix-action on $\delta\rho$
in (\ref{delta-lam-r}) reduces to $\delta\lambda=\Pi_{\bot}\delta\rho$.
Remember that on the constraint surface we have proper projection
properties of the form $\Pi_{\bot}^{2}=\Pi_{\bot}$ (\ref{PiParsquare})
implying $\Pi_{\bot}\Pi_{\Vert}=0$ (\ref{PiparPiperp}). This suggests
that the new gauge symmetry transformations for $\rho$ is (at the
constraint surface) of the form $\delta_{({\rm sym})}\rho\propto\Pi_{\Vert}$.
However, off the constraint surface we do not have $\Pi_{\bot}^{2}=\Pi_{\bot}$
but something of the form $\Pi_{(f)\bot}(\lambda,\bar{\lambda})\Pi_{(f)\bot}(\rho,\bar{\rho})=\Pi_{(f)\bot}(\rho,\bar{\rho})$
((\ref{LinProjProp}) and (\ref{linprojprop})) implying $\Pi_{(f)\Vert}(\lambda,\bar{\lambda})\Pi_{(f)\bot}(\rho,\bar{\rho})=0$.
Still choosing $\delta_{({\rm sym})}\rho\propto\Pi_{(f)\Vert}(\lambda,\bar{\lambda})$
would lead in (\ref{delta-lam-r}) to $\Pi_{(f)\bot}(\rho,\bar{\rho})\delta_{({\rm sym})}\rho\propto\Pi_{(f)\bot}(\rho,\bar{\rho})\Pi_{(f)\Vert}(\lambda,\bar{\lambda})$
which is the wrong order of matrices, if we want to use the just mentioned
$\Pi_{(f)\Vert}(\lambda,\bar{\lambda})\Pi_{(f)\bot}(\rho,\bar{\rho})=0$.
This is the main reason for us to discuss later the new gauge symmetries
only in the case $f(\auxi)=h(\auxi)\equiv\frac{1+\sqrt{1-\auxi}}{2\sqrt{1-\auxi}}$
where we get Hermitian projection matrices which have additional properties
that will resolve this problem. In spite of this issue, there will
certainly be also for general $f$ a new gauge symmetry taking care
of the artificial degrees of freedom that we introduced by using the
projection $P_{(f)}$. These symmetries are given by the 0-eigenvectors
of the matrix in (\ref{delta-lam-r}). For a general function $f$
it would thus require some work to determine this symmetry, while
in the Hermitian case the above naive guess will already work. 

We have now argued about the new gauge symmetries by just discussing
the variation of the constrained variables in terms of the unconstrained
variables. Let us now provide the general variation of the ghost action.
Starting from the variation (\ref{deltaSorig}) of the ghost action
and applying the relation in (\ref{delta-lam-r}) not only for the
general variation $\delta$, but also for the partial derivative $\delta\to\bar{\partial}$,
we obtain \rem{old version of this discussion hidden here in \LyX{}
file } 
\begin{eqnarray}
\lqn{\delta S_{{\rm gh}(f)}[\rho,\omega_{z},\bar{\rho},\bar{\omega}_{z},\bs t,\bs s_{z}]=}\nonumber \\
 & = & -\int_{\Sigma}d^{2}z\:\left(\delta\rho^{T},\delta\bar{\rho}^{T},\delta\bs t^{T}\right)\times\nonumber \\
 &  & \times\!\!\left(\!\!\begin{array}{ccc}
\Pi_{(f)\bot}^{T}(\rho,\bar{\rho}) & \bar{\rest}_{(f)\bot}^{T}(\rho,\bar{\rho}) & \biggl\{\zwek{-(\bs t\gamma_{a}\bar{\lambda})\frac{(\Pi_{(f)\bot}^{T}(\rho,\bar{\rho})\gamma^{a}\Pi_{\bot})}{2(\lambda\bar{\lambda})}+}{-\frac{(\bar{\rest}_{(f)\bot}^{T}(\rho,\bar{\rho})\gamma_{a}\Pi_{\bot}^{T}\bs t)\otimes(\lambda\gamma^{a})}{2(\lambda\bar{\lambda})}}\biggr\}\\
\rest_{(f)\bot}^{T}(\rho,\bar{\rho}) & \bar{\Pi}_{(f)\bot}^{T}(\rho,\bar{\rho}) & \biggl\{\zwek{-(\bs t\gamma_{a}\bar{\lambda})\frac{(\rest_{(f)\bot}^{T}(\rho,\bar{\rho})\gamma^{a}\Pi_{\bot})}{2(\lambda\bar{\lambda})}+}{-\frac{(\bar{\Pi}_{(f)\bot}^{T}(\rho,\bar{\rho})\gamma_{a}\Pi_{\bot}^{T}\bs t)\otimes(\lambda\gamma^{a})}{2(\lambda\bar{\lambda})}}\biggr\}\\
0 & 0 & \bar{\Pi}_{\bot}^{T}
\end{array}\!\!\right)\!\!\left(\!\begin{array}{c}
\bar{\partial}\omega_{z}\\
\bar{\partial}\bar{\omega}_{z}\\
\bar{\partial}\bs s_{z}
\end{array}\!\right)\nonumber \\
 &  & +\int_{\Sigma}d^{2}z\:\left(\delta\omega^{T},\delta\bar{\omega}^{T},\delta\bs s_{z}^{T}\right)\times\nonumber \\
 &  & \times\!\!\left(\begin{array}{ccc}
\Pi_{(f)\bot}(\rho,\bar{\rho}) & \rest_{(f)\bot}(\rho,\bar{\rho}) & 0\\
\bar{\rest}_{(f)\bot}(\rho,\bar{\rho}) & \bar{\Pi}_{(f)\bot}(\rho,\bar{\rho}) & 0\\
\!\!\!\!\biggl\{\!\!\zwek{(\bs t\gamma_{a}\bar{\lambda})\frac{(\Pi_{\bot}^{T}\gamma^{a}\Pi_{(f)\bot}(\rho,\bar{\rho}))}{2(\lambda\bar{\lambda})}+}{+\frac{(\gamma^{a}\lambda)\otimes(\bs t\Pi_{\bot}\gamma_{a}\bar{\rest}_{(f)\bot}(\rho,\bar{\rho}))}{2(\lambda\bar{\lambda})}}\!\!\biggr\}\!\! & \!\!\biggl\{\!\!\!\!\zwek{(\bs t\gamma_{a}\bar{\lambda})\frac{(\Pi_{\bot}^{T}\gamma^{a}\rest_{(f)\bot}(\rho,\bar{\rho}))}{2(\lambda\bar{\lambda})}+}{\,+\frac{(\gamma^{a}\lambda)\otimes(\bs t\Pi_{\bot}\gamma_{a}\bar{\Pi}_{(f)\bot}(\rho,\bar{\rho}))}{2(\lambda\bar{\lambda})}}\!\!\biggr\}\!\!\!\! & -\bar{\Pi}_{\bot}
\end{array}\!\!\!\right)\!\!\left(\!\begin{array}{c}
\bar{\partial}\rho\\
\bar{\partial}\bar{\rho}\\
\bar{\partial}\bs t
\end{array}\!\right)\nonumber \\
 &  & \quad\label{deltaS2}
\end{eqnarray}
Note that because of the projection properties (\ref{LinProjProp})
and (\ref{linprojprop}), the fact that the matrices have argument
$(\rho,\bar{\rho})$ is not in contradiction with (\ref{deltaSorig2}),
where they are evaluated at $(\lambda,\bar{\lambda})$. The equations
of motion for the conjugate momenta $\omega_{z},\bar{\omega}_{z}$
and $\bs s_{z}$ that one reads off from \eqref{deltaS2} look formally
different from \eqref{omega-s-eoms} but are in fact equivalent. The
new gauge symmetries that we have mentioned above will always allow
to fix $\rho^{\alpha}=\lambda^{\alpha}$ and $\bs t^{\alpha}=\bs r^{\alpha}$
and thus return to the original action. For this gauge the equations
of motion from \eqref{deltaS2} precisely reduce to the ones in \eqref{omega-s-eoms}. 

Note that from (\ref{deltaS2}) we can also re-derive the antighost
gauge transformations (\ref{gaugetrafoomegaOnsurf})-(\ref{gaugetrafosOnsurf})
in terms of the projector matrices. If we vary only $\omega_{z},\bar{\omega}_{z}$
and $\bs s_{z}$ and require the variation (\ref{deltaS2}) to vanish
for any $\bar{\partial}\lambda,\bar{\partial}\bar{\lambda}$ and $\bar{\partial}\bs t$,
we obtain 
\begin{equation}
\left(\!\!\!\begin{array}{ccc}
\Pi_{(f)\bot}^{T}(\rho,\bar{\rho}) & \bar{\rest}_{(f)\bot}^{T}(\rho,\bar{\rho}) & \biggl\{\!\!\zwek{-(\bs t\gamma_{a}\bar{\lambda})\frac{(\Pi_{(f)\bot}^{T}(\rho,\bar{\rho})\gamma^{a}\Pi_{\bot})}{2(\lambda\bar{\lambda})}+}{-\frac{(\bar{\rest}_{(f)\bot}^{T}(\rho,\bar{\rho})\gamma_{a}\Pi_{\bot}^{T}\bs t)\otimes(\lambda\gamma^{a})}{2(\lambda\bar{\lambda})}}\!\!\biggr\}\\
\rest_{(f)\bot}^{T}(\rho,\bar{\rho}) & \bar{\Pi}_{(f)\bot}^{T}(\rho,\bar{\rho}) & \biggl\{\!\!\zwek{-(\bs t\gamma_{a}\bar{\lambda})\frac{(\rest_{(f)\bot}^{T}(\rho,\bar{\rho})\gamma^{a}\Pi_{\bot})}{2(\lambda\bar{\lambda})}+}{-\frac{(\bar{\Pi}_{(f)\bot}^{T}(\rho,\bar{\rho})\gamma_{a}\Pi_{\bot}^{T}\bs t)\otimes(\lambda\gamma^{a})}{2(\lambda\bar{\lambda})}}\!\!\biggr\}\\
0 & 0 & \bar{\Pi}_{\bot}^{T}
\end{array}\!\!\!\!\right)\!\!\left(\!\begin{array}{c}
\delta_{{\rm (sym)}}\omega_{z}\\
\delta_{{\rm (sym)}}\bar{\omega}_{z}\\
\delta_{{\rm (sym)}}\bs s_{z}
\end{array}\!\right)\stackrel{!}{=}0
\end{equation}
Again on the constraint surface it is obvious that the symmetry transformation
$\delta_{{\rm (sym)}}\omega_{z\alpha}\equiv\delta_{(\mu)}\omega_{z\alpha}=(\Pi_{\Vert}^{T}\mu_{z})_{\alpha}$
with some spinorial gauge parameter $\mu_{z\alpha}$ obeys this condition.
And for the antighosts this even holds off the constraint surface,
because the order of the projection matrices is here due to the transposition
the correct one, i.e. $\Pi_{(f)\bot}^{T}(\rho,\bar{\rho})\Pi_{(f)\Vert}^{T}(\lambda,\bar{\lambda})=0\quad\forall f$
according to (\ref{LinProjProp}). For the variable $\bs s_{z}$ a
first guess would be $\delta_{(\bs{\sigma})}\bs s_{z}^{\alpha}=(\Pi_{\Vert}\bs{\sigma}_{z})^{\alpha}$,
which turns out to need a compensating transformation of $\bar{\omega}_{z}$.
One thus ends up with the following gauge transformations for the
antighosts 
\begin{eqnarray}
\delta_{(\mu)}\omega_{z\alpha} & = & (\bar{\Pi}_{\Vert}\mu_{z})_{\alpha}\label{gaugetrafoOmega}\\
\delta_{(\bar{\mu},\bs{\sigma})}\bar{\omega}_{z}^{\alpha} & = & (\Pi_{\Vert}\bar{\mu}_{z})^{\alpha}-(\lambda\gamma^{a}\Pi_{\Vert}\bs{\sigma}_{z})\frac{(\gamma_{a}\overbrace{\bar{\Pi}_{\bot}\bs t}^{\bs r})^{\alpha}}{2(\lambda\bar{\lambda})}\label{gaugetrafoOmegabar}\\
\delta_{(\bs{\sigma})}\bs s_{z}^{\alpha} & = & (\Pi_{\Vert}\bs{\sigma}_{z})^{\alpha}\label{gaugetrafoS}
\end{eqnarray}
with as usual $\Pi_{\Vert}\equiv\Pi_{(f)\Vert}(\lambda,\bar{\lambda})$
where $\lambda^{\alpha}\equiv P_{(f)}^{\alpha}(\rho,\bar{\rho})$
and spinorial gauge parameters $\mu_{z\alpha},\bar{\mu}_{z}^{\alpha}$
and $\bs{\sigma}_{z}^{\alpha}$.  Having in mind that $\Pi_{\Vert}=\frac{(\gamma_{a}\bar{\lambda})\otimes(\lambda\gamma^{a})}{2(\lambda\bar{\lambda})}$
and comparing with the form of the gauge transformations in (\ref{gaugetrafoomegaOnsurf})-(\ref{gaugetrafosOnsurf}),
we can deduce the relation between the spinorial gauge parameters
here and the vectorial ones there:%
\footnote{While in (\ref{mumubarsigma}) the expressions for $\mu_{za}$ and
$\bs{\sigma}_{z}^{a}$ can directly be read off from $\delta_{(\mu)}\omega_{z\alpha}$
and $\delta_{(\bs{\sigma})}\bs s_{z}^{\alpha}$, the one for $\bar{\mu}_{z}^{a}$
requires some Fierzing: 
\begin{eqnarray*}
\delta_{(\bar{\mu},\bs{\sigma})}\bar{\omega}_{z}^{\alpha} & \stackrel{\mbox{\tiny\eqref{gaugetrafoOmegabar}}}{=} & \frac{(\gamma_{a}\bar{\lambda})^{\alpha}(\lambda\gamma^{a}\bar{\mu}_{z})}{2(\lambda\bar{\lambda})}-\frac{(\lambda\overbrace{\gamma^{a}\gamma_{b}}^{-\gamma_{b}\gamma^{a}\lqn{{\scriptstyle +2\delta_{b}^{a}}}}\bar{\lambda})(\lambda\gamma^{b}\bs{\sigma}_{z})}{4(\lambda\bar{\lambda})^{2}}(\gamma_{a}\bs r)^{\alpha}=\\
 & \stackrel{{\rm Fierz}}{=} & \underbrace{\frac{(\lambda\gamma^{a}\bar{\mu}_{z})}{2(\lambda\bar{\lambda})}}_{\bar{\mu}_{z}^{a}}(\gamma_{a}\bar{\lambda})^{\alpha}-\underbrace{\frac{(\lambda\gamma^{a}\bs{\sigma}_{z})}{2(\lambda\bar{\lambda})}}_{\bs{\sigma}_{z}^{a}}(\gamma_{a}\bs r)^{\alpha}\qquad\fussend
\end{eqnarray*}
}
\begin{equation}
\mu_{za}=\frac{(\bar{\lambda}\gamma_{a})^{\alpha}\mu_{z\alpha}}{2(\lambda\bar{\lambda})}\,,\quad\bar{\mu}_{z}^{a}=\frac{(\lambda\gamma^{a})_{\alpha}\bar{\mu}_{z}^{\alpha}}{2(\lambda\bar{\lambda})}\,,\quad\bs{\sigma}_{z}^{a}=\frac{(\lambda\gamma^{a})_{\alpha}\bs{\sigma}_{z}^{\alpha}}{2(\lambda\bar{\lambda})}\quad\label{mumubarsigma}
\end{equation}
Both parametrizations, either with $\mu_{za},\bar{\mu}_{z}^{a}$ and
$\bs{\sigma}_{z}^{a}$ or with $\mu_{z\alpha},\bar{\mu}_{z}^{\alpha}$
and $\bs{\sigma}_{z}^{\alpha}$ are reducible. Remembering that $\tr\Pi_{\Vert}=5$,
it is apparent from (\ref{gaugetrafoOmega})-(\ref{gaugetrafoS})
that for each of these spinors only 5 components effectively enter
the gauge transformation.

\subsubsection*{Gauge transformations in the Hermitian case}

The antighost gauge transformations (\ref{gaugetrafoOmega})-(\ref{gaugetrafoS})
formally look the same for any particular choice of $f$ in our projection
$P_{(f)}$. The only way $f$ enters there is via the dependence of
$\lambda^{\alpha}\equiv P_{(f)}^{\alpha}(\rho,\bar{\rho})$ on it.
The situation is different for the new gauge transformations of the
ghosts $\rho^{\alpha},\bar{\rho}_{\alpha}$ and $\bs t_{\alpha}$,
because for the particular choice $f(\auxi)=h(\auxi)\equiv\frac{1+\sqrt{1-\auxi}}{2\sqrt{1-\auxi}}$
the matrix $\Pi_{\Vert}\equiv\Pi_{(h)\Vert}(\lambda,\bar{\lambda})$
on the constraint surface commutes with $\Pi_{(h)\Vert}(\rho,\bar{\rho})$
off the constraint surface (see (\ref{Proj-commutativity}) and (\ref{proj-commutativity})). 

So looking again at (\ref{delta-lam-r}) and having in mind (\ref{Proj-commutativity})
and (\ref{proj-commutativity}), it is obvious that 
\begin{equation}
\delta_{(\bs{\tau})}\bs t_{\alpha}\equiv(\bar{\Pi}_{\Vert}\bs{\tau})_{\alpha}
\end{equation}
for some spinorial gauge parameter $\bs{\tau}_{\alpha}$ will not
induce a variation of $\lambda^{\alpha},\bar{\lambda}_{\alpha}$ or
$\bs r_{\alpha}$ on the left-hand side of (\ref{delta-lam-r}) and
thus be a gauge transformation. Similarly one can try
\begin{eqnarray}
\delta_{(\nu)}\rho^{\alpha} & = & (\Pi_{\Vert}\nu)^{\alpha}\\
\delta_{(\bar{\nu})}\bar{\rho}_{\alpha} & = & (\bar{\Pi}_{\Vert}\bar{\nu})_{\alpha}
\end{eqnarray}
with spinorial gauge parameters $\nu^{\alpha}$ and $\bar{\nu}_{\alpha}$.
As in (\ref{delta-lam-r}) $\delta\rho$ and $\delta\bar{\rho}$ enter
not only in the variation of $\delta\lambda$ and $\delta\bar{\lambda}$,
one might suspect that one needs also some compensating transformation
of $\delta_{(\nu,\bar{\nu})}\bs t$. But one can easily check that
the transformations as they are written above are already a symmetry
and do not need such a compensation. \rembreak\rem{hidden here
in a \LyX{}-note: gauge variations of $\auxii^{a},\auxi$, $\Pi_{(h)\bot}(\rho,\bar{\rho})$
and $\rest_{(h)\bot}(\rho,\bar{\rho})$}

\paragraph{Gauge fixing }

One interesting aspect of this new artificial gauge symmetry is, that
we can try to find different gauge fixings that bring us from the
pure spinor constraint to some other constraint. 

In particular the role of $\rho^{\alpha}$ and $\omega_{z\alpha}$
can now be interchanged by fixing $\omega_{z\alpha}$ (with a constraint
linear in $\omega_{z\alpha}$) and leaving $\rho$ unconstrained.
Choosing a linear constraint for $\omega_{z\alpha}$ might give a
clue how to obtain such a ghost system from an underlying gauge freedom.
A natural gauge fixing constraint might be 
\begin{equation}
(\omega_{z}\gamma^{a}\bar{P}_{(h)}(\rho,\bar{\rho}))\stackrel{!}{=}0\quad\mbox{(gauge 1)}
\end{equation}
\rem{Is it possible to choose a gauge where even $\omega_{z}\gamma^{a}\bar{\rho}=0$?
Would be more appealing as it is only linear in each variable and
could thus also be used for corresponding gauge parameters. However,
it might be more restrictive: In the U5-formalism we have 
\begin{eqnarray*}
\bra{\Psi}Cb_{\mf e}^{\dagger}\ket{\Phi} & \stackrel{{\rm chiral}}{=} & \psi^{+}\phi_{\mf e}+\psi_{\mf e}\phi^{+}-\tfrac{1}{4}\psi^{\mf{cc}}\epsilon_{\mf{ccebb}}\phi^{\mf{bb}}\stackrel{!}{=}0\quad(I)\\
\bra{\Psi}Cb^{\mf e}\ket{\Phi} & \stackrel{{\rm chiral}}{=} & \psi_{\mf a}\phi^{\mf{ea}}+\psi^{\mf{ea}}\phi_{\mf a}\stackrel{!}{=}0\quad(II)
\end{eqnarray*}
Solving (I) $\phi_{\mf e}=-\tfrac{1}{\psi^{+}}\psi_{\mf e}\phi^{+}+\tfrac{1}{4\psi^{+}}\psi^{\mf{cc}}\epsilon_{\mf{ccebb}}\phi^{\mf{bb}}$
and plugging into (II) yields
\begin{eqnarray*}
0 & \stackrel{!}{=} & \psi_{\mf a}\phi^{\mf{ea}}+\psi^{\mf{ea}}\left(-\tfrac{1}{\psi^{+}}\psi_{\mf a}\phi^{+}+\tfrac{1}{4\psi^{+}}\psi^{\mf{cc}}\epsilon_{\mf{ccabb}}\phi^{\mf{bb}}\right)=\\
 & = & \psi_{\mf a}\left(\phi^{\mf{ea}}-\tfrac{\phi^{+}}{\psi^{+}}\psi^{\mf{ea}}\right)+\tfrac{1}{4\psi^{+}}\psi^{\mf{ea}}\psi^{\mf{cc}}\epsilon_{\mf{ccabb}}\phi^{\mf{bb}}
\end{eqnarray*}
In contrast to the pure-spinor case where $\Psi=\Phi$, here (II)
gives additional constraints.}\rembreak which is equivalent to 
\begin{equation}
\tilde{\omega}_{z\alpha}\equiv(\bar{\Pi}_{\bot}\omega_{z})_{\alpha}\stackrel{!}{=}\omega_{z\alpha}\quad\mbox{(gauge 1)}\label{gauge:interchange}
\end{equation}
with $\bar{\Pi}_{\bot}\equiv\bar{\Pi}_{(h)\bot}(P_{(h)}(\rho,\bar{\rho}),\bar{P}_{(h)}(\rho,\bar{\rho}))$.
It fixes only the gauge parameter $\mu_{z}^{a}$. Now the antighost
$\omega_{z\alpha}$ would obey a constraint, while the ghost $\rho^{\alpha}$
would be unconstrained, but not gauge invariant. $\lambda^{\alpha}\equiv P_{(h)}(\rho,\bar{\rho})$
would still be the gauge invariant combination. However, according
to (\ref{actionInTermsOfGaugeInvVar}) (with $\lambda$ replaced by
$P_{(h)}(\rho,\bar{\rho})$), this gauge would not lead to a free
action $S_{{\rm gh}(h)}$. 

Looking at (\ref{delta-lam-r}) it is clear that absorbing the transpose
of the matrix there into the antighosts would lead to another free
action (we use Hermiticity of $\Pi_{(h)\bot}$ and symmetry of $\rest_{(h)\bot}$):
\begin{align}
\left(\!\!\!\begin{array}{ccc}
\bar{\Pi}_{(h)\bot}(\rho,\bar{\rho}) & \bar{\rest}_{(h)\bot}(\rho,\bar{\rho}) & \biggl\{\zwek{-(\bs t\gamma_{a}\bar{\lambda})\frac{(\bar{\Pi}_{(h)\bot}(\rho,\bar{\rho})\gamma^{a}\Pi_{\bot})}{2(\lambda\bar{\lambda})}+}{-\frac{(\bar{\rest}_{(h)\bot}(\rho,\bar{\rho})\gamma_{a}\bar{\Pi}_{\bot}\bs t)\otimes(\lambda\gamma^{a})}{2(\lambda\bar{\lambda})}}\biggr\}\\
\rest_{(h)\bot}(\rho,\bar{\rho}) & \Pi_{(h)\bot}(\rho,\bar{\rho}) & \biggl\{\zwek{-(\bs t\gamma_{a}\bar{\lambda})\frac{(\rest_{(h)\bot}(\rho,\bar{\rho})\gamma^{a}\Pi_{\bot})}{2(\lambda\bar{\lambda})}+}{-\frac{(\Pi_{(h)\bot}(\rho,\bar{\rho})\gamma_{a}\bar{\Pi}_{\bot}\bs t)\otimes(\lambda\gamma^{a})}{2(\lambda\bar{\lambda})}}\biggr\}\\
0 & 0 & \Pi_{\bot}
\end{array}\!\!\!\!\right)\!\!\!\left(\!\!\begin{array}{c}
\omega_{z}\\
\bar{\omega}_{z}\\
\bs s_{z}
\end{array}\!\!\right)\!\stackrel{!}{=}\! & \!\left(\!\!\begin{array}{c}
\omega_{z}\\
\bar{\omega}_{z}\\
\bs s_{z}
\end{array}\!\!\right)\nonumber \\
\mbox{(gauge 2)}
\end{align}
The above (somewhat ugly) constraint would therefore replace in this
gauge the original constraints $(\lambda\gamma^{a}\lambda)=(\bar{\lambda}\gamma_{a}\bar{\lambda})=(\bar{\lambda}\gamma_{a}\bs r)=0$
while keeping a free action.

Conceptionally more interesting would be those gauges which are linear
in each spinorial variable, as $(\bar{\rho}\gamma^{a}\omega_{z})=0$.
This sort of constraint on the $\rho$-ghosts could be easily translated
into an equivalent constraint on the odd gauge parameters of some
underlying gauge symmetry. In contrast the pure spinor constraint
$(\lambda\gamma^{a}\lambda)=0$ would be trivial for anticommuting
parameters. However, it seems that the gauge $(\bar{\rho}\gamma^{a}\omega_{z})=0$
is too strong and would fix $10$ instead of five degrees of freedom. 

Another aim for a possible gauge fixing would be that $\rho^{\alpha}$
and $\bar{\rho}_{\alpha}$ as well as $\omega_{z\alpha}$ and $\bar{\omega}_{z}^{\alpha}$
are treated as complex conjugates in all transformations. This is
so far not the case in the BRST transformations, as well as in the
gauge transformations. Having in mind also the matter fields of the
pure spinor string, with BRST transformation $\es\tet^{\alpha}=\lambda^{\alpha}$,
this would require to identify some variable with $\bar{\tet}^{\alpha}$
(perhaps after turning the gauge parameter $\bar{\nu}_{\alpha}$ into
an anticommuting ghost). It is not yet clear if this would lead anywhere.

\subsection{Ghost action in the U(5) formalism}

\label{sec:ghost-action-3} We will not translate the complete previous
discussion into the U(5) formalism. Instead we just want to relate
some known fact from the literature to our presentation. For simplicity,
we even restrict to the minimal formalism, i.e. neglect the variables
$\bar{\lambda}_{\alpha},\bar{\omega}_{z}^{\alpha},\bs r_{\alpha},\bs s_{z}^{\alpha}$
from (\ref{Sgh-orig}). In the minimal formalism the ghost action
in U(5) coordinates thus reads
\begin{equation}
S[\lambda,\omega_{z}]=\int d^{2}z\quad\bar{\partial}\lambda^{+}\omega_{z+}+\tfrac{1}{2}\bar{\partial}\lambda^{\mf a_{1}\mf a_{2}}\omega_{z\mf a_{1}\mf a_{2}+}+\bar{\partial}\lambda_{\mf a}\omega_{z}^{\mf a}
\end{equation}
Using the constraint (\ref{U5ps-constraint}) we can replace (on the
patch $\lambda^{+}\neq0$) $\lambda^{\mf a}$ by 
\begin{equation}
\lambda_{\mf a}=\tfrac{1}{8\lambda^{+}}\epsilon_{\mf a\mf b_{1}\mf b_{2}\mf b_{3}\mf b_{4}}\lambda^{\mf b_{1}\mf b_{2}}\lambda^{\mf b_{3}\mf b_{4}}
\end{equation}
which leads to 
\begin{eqnarray}
S[\lambda,\omega_{z}] & = & \int d^{2}z\quad\bar{\partial}\lambda^{+}\left(\omega_{z+}-\tfrac{1}{8(\lambda^{+})^{2}}\epsilon_{\mf c\mf b_{1}\mf b_{2}\mf b_{3}\mf b_{4}}\lambda^{\mf b_{1}\mf b_{2}}\lambda^{\mf b_{3}\mf b_{4}}\omega_{z}^{\mf c}\right)+\nonumber \\
 &  & +\tfrac{1}{2}\bar{\partial}\lambda^{\mf a_{1}\mf a_{2}}\left(\omega_{z\mf a_{1}\mf a_{2}}+\tfrac{1}{2\lambda^{+}}\epsilon_{\mf a_{1}\mf a_{2}\mf b_{1}\mf b_{2}\mf c}\lambda^{\mf b_{1}\mf b_{2}}\omega_{z}^{\mf c}\right)
\end{eqnarray}
Now let us compare the antighost combinations in the brackets with
the gauge invariant antighost $\tilde{\omega}_{z}\equiv(\Pi^{T}\omega_{z})$
defined in (\ref{gaugeinvOmegaOnsurf}), but with its non-covariant
version (call it $\omega'_{z})$, where $\bar{\lambda}$ is replaced
by the reference spinor $\bar{\chi}=(1,0,0)$ of equation (\ref{ref-spinor-ansatz})
so that we can use the projection matrix (\ref{U5-proj-matrix}) evaluated
at $\rho=\lambda$ which we denote by $\Pi_{(\bar{\chi})}\equiv\Pi_{(\bar{\chi})}(\lambda,\bar{\lambda})$:
\begin{eqnarray}
\omega'_{z\alpha} & \!\!\!\!\equiv & \!\!\!\!(\Pi_{(\bar{\chi})}^{T})_{\alpha}\hoch{\beta}\omega_{z\beta}=\\
 & \!\!\!\!\stackrel{\mbox{\small\eqref{U5-proj-matrix}}}{=} & \!\!\!\!\left(\!\!\begin{array}{ccc}
1 & 0 & -\tfrac{1}{8(\lambda^{+})^{2}}\epsilon_{\mf b\mf c_{1}\mf c_{2}\mf c_{3}\mf c_{4}}\lambda^{\mf c_{1}\mf c_{2}}\lambda^{\mf c_{3}\mf c_{4}}\\
0 & \delta_{\mf a_{1}}^{\mf b_{1}}\delta_{\mf a_{2}}^{\mf b_{2}}-\delta_{\mf a_{2}}^{\mf b_{1}}\delta_{\mf a_{1}}^{\mf b_{2}} & \tfrac{1}{2\lambda^{+}}\epsilon_{\mf b\mf a_{1}\mf a_{2}\mf c_{1}\mf c_{2}}\lambda^{\mf c_{1}\mf c_{2}}\\
0 & 0 & 0
\end{array}\!\!\right)\!\!\left(\!\!\begin{array}{c}
\omega_{z+}\\
\omega_{z\mf b_{1}\mf b_{2}}\\
\omega_{z}^{\mf b}
\end{array}\!\!\right)\qquad
\end{eqnarray}
In order to be consistent with \eqref{U5-variation} the matrix multiplication
has to be understood as coming with a factor $\tfrac{1}{2}$ when
summing over the double-indices ${\scriptstyle \mf b_{1}\mf b_{2}}$.
We thus arrive at 
\begin{equation}
\omega'_{z\alpha}=\left(\begin{array}{c}
\omega_{z+}-\tfrac{1}{8(\lambda^{+})^{2}}\epsilon_{\mf c_{1}\mf c_{2}\mf c_{3}\mf c_{4}\mf b}\lambda^{\mf c_{1}\mf c_{2}}\lambda^{\mf c_{3}\mf c_{4}}\omega_{z}^{\mf b}\\
\omega_{z\mf a_{1}\mf a_{2}}+\tfrac{1}{2\lambda^{+}}\epsilon_{\mf a_{1}\mf a_{2}\mf c_{1}\mf c_{2}\mf b}\lambda^{\mf c_{1}\mf c_{2}}\omega_{z}^{\mf b}\\
0
\end{array}\right)
\end{equation}
These are precisely the terms that appeared naturally in the action
above, so we can write 
\begin{eqnarray}
S[\lambda,\omega_{z}] & = & \int d^{2}z\quad\bar{\partial}\lambda^{+}\omega'_{z+}+\tfrac{1}{2}\bar{\partial}\lambda^{\mf a_{1}\mf a_{2}}\omega'_{z\mf a_{1}\mf a_{2}}
\end{eqnarray}
This rewriting of the action was first presented in equations (26)
and (27) of \cite{Oda:2001zm} and our $\omega_{z}'$ agrees up to
a conventional sign of the components $\omega_{z}^{\mf a}$ (and $\lambda_{\mf a}$).
The upshot of the discussion in this subsection, however, is that
these gauge invariant combinations $\omega'_{z}$ fit into our more
general projection-picture if one chooses the reference spinor $\bar{\chi}=(1,0,0)$
of equation (\ref{ref-spinor-ansatz}).

\section{Volume form of the pure spinor space}

\label{sec:integration-measure}The holomorphic volume form of the
pure spinor space is already well known. Originally it was given in
\cite{Berkovits:2005bt} p.13,(4.4), or even in \cite{Berkovits:2004px}
p.14, (3.6). It was further elaborated on it in \cite{Mafra:2009wq}
p. 28,29 or p.33, (2.85)-(2.87) and \cite{Mafra:2009wi}p.5, (2.5)-(2.10)
where it takes the form 
\begin{equation}
[d^{11}\lambda]=\tfrac{1}{11!5!(\lambda\os{}{\bar{\lambda}})^{3}}(\bar{\lambda}\gamma_{a})^{\alpha_{1}}(\bar{\lambda}\gamma_{b})^{\alpha_{2}}(\bar{\lambda}\gamma_{c})^{\alpha_{3}}(\gamma^{abc})^{\alpha_{4}\alpha_{5}}\epsilon_{\alpha_{1}\ldots\alpha_{5}\beta_{1}\ldots\beta_{11}}\de\lambda^{\beta_{1}}\cdots\de\lambda^{\beta_{11}}\label{hol-volform}
\end{equation}
See also \cite{Gomez:2009qd} p.16 and \cite{Grassi:2010ca}. At first
sight the 11-form (\ref{hol-volform}) does not look holomorphic in
$\lambda^{\alpha}$. However, going to U(5)-coordinates, one can show
that all $\bar{\lambda}$-dependence disappears.\rem{Reference for
proof? If we want to show it ourselves: $\gamma^{abc}$ in U5-formalism
goes more into the U5-formalism than intended for this paper.} 

Wedging the holomorphic volume form with its complex conjugate yields
according to the above cited references
\begin{equation}
[d^{11}\lambda]\wedge[d^{11}\bar{\lambda}]=\tfrac{1}{11!(\lambda\os{}{\bar{\lambda}})^{3}}(\de\lambda^{\alpha}\de\bar{\lambda}_{\alpha})^{11}
\end{equation}
Let us think of the pure spinor space as being embedded in the ambient
space $\mathbb{C}^{16}$ at $\xi=0$. Can we obtain the above holomorphic
volume form from the holomorphic volume form of $\mathbb{C}^{16}$
by the variable transformations $\rho^{\alpha}\mapsto(\lambda^{\alpha},\auxii^{a})$
or $\rho^{\alpha}\mapsto(\lambda^{\alpha},\Auxii^{a})$ discussed
on page \pageref{var-trafo} and page \pageref{newz}, or by yet another
similar one? The canonical holomorphic volume form of the ambient
space $\mathbb{C}^{16}$ is given by 
\begin{equation}
[d^{16}\rho]\equiv\tfrac{1}{16!}\epsilon_{\alpha_{1}\ldots\alpha_{16}}\de\rho^{\alpha_{1}}\cdots\de\rho^{\alpha_{16}}\label{hol-target-volform}
\end{equation}
The corresponding measure of the complete space is therefore  
\begin{eqnarray}
[d^{16}\rho]\wedge[d^{16}\bar{\rho}] & = & \!\!\!\tfrac{1}{(16!)^{2}}\underbrace{\epsilon_{\alpha_{1}\ldots\alpha_{16}}\epsilon^{\beta_{1}\ldots\beta_{16}}}_{16!\delta_{\alpha_{1}\ldots\alpha_{16}}^{\beta_{1}\ldots\beta_{16}}}\de\rho^{\alpha_{1}}\cdots\de\rho^{\alpha_{16}}\de\bar{\rho}_{\beta_{1}}\cdots\de\bar{\rho}_{\beta_{16}}\!=\qquad\\
 & = & \!\!\!-\tfrac{1}{16!}(\de\rho^{\alpha}\de\bar{\rho}_{\alpha})^{16}\label{target-measure}
\end{eqnarray}
In the corresponding discussion within the toy-model in the appendix
around page \pageref{toy:Hauxi} it became clear that in order to
obtain the expected holomorphic volume form of the constrained space
after the variable transformation it was necessary to redefine $\auxii$
or $\Auxii$ once more in (\ref{toy:Hauxi}). We will try the same
here and define 
\begin{equation}
\Hauxii^{a}\equiv(\lambda\bar{\lambda})\Auxii^{a}\quad,\quad\Hauxi\equiv\tfrac{1}{2}\Hauxii^{a}\bar{\Hauxii}_{a}=(\lambda\bar{\lambda})^{2}\Auxi\label{Hauxii}
\end{equation}
where $\lambda^{\alpha}$ will be understood as being $P_{(h)}^{\alpha}(\rho,\bar{\rho})$.
We have expressed $\Hauxii^{a}$ in terms of $\Auxii^{a}$, as we
will need some equations which where particularly simple for $\Auxii^{a}$.
However, a very good motivation for this reparametrization (apart
from the results of the toy-model) is that $\Hauxii^{a}$ as a function
of $\rho^{\alpha}$ becomes holomorphic and very simple. In order
to see this, we have to use explicit form of $\Auxii^{a}$ in (\ref{newz})
and the fact that according to (\ref{PhiIsModSqu}) and (\ref{Potential})
we have $(\lambda\bar{\lambda})=\tfrac{1}{2}(1+\sqrt{1+\auxi})(\rho\bar{\rho})$.
Then altogether, we have the following variable transformation%
\footnote{\label{fn:Hauxii}With $(\rho\gamma^{a}\rho)$ being the most obvious
choice to parametrize the distance of $\rho^{\alpha}$ from being
a pure spinor, it is a fair question why we have not used $\Hauxii^{a}$
as it is given in (\ref{Hauxii-simple}) from the beginning. There
are two aspects. One is that the equivalence between (\ref{Hauxii-simple})
and (\ref{Hauxii}) holds only in the case $f=h$. For general $f$
the definition $\Hauxii^{a}\equiv(\lambda\bar{\lambda})\Auxii^{a}$
would be equivalent to $\Hauxii^{a}=2(\rho\bar{\rho})f(\auxi)^{2}\frac{(1-\auxi)\auxii^{a}}{(1+\sqrt{1-\auxi})^{2}}$.
In turn, defining $\Hauxii^{a}\equiv\tfrac{1}{2}(\rho\gamma^{a}\rho)$
would not be equivalent to (\ref{Hauxii}) any more and thus will
not have the simple inverse transformation (\ref{inverseWithHauxii}).
The other aspect is that although also $\auxii^{a}=\frac{(\rho\gamma^{a}\rho)}{(\rho\bar{\rho})}$
has a non-simple inverse (\ref{Proj-inv}), that variable appeared
very naturally in the original projection (\ref{ProjGeneral}). It
was introduced as a placeholder to make equations shorter significantly
without loosing the feeling for the $\rho^{\alpha}$-dependence within
the projection. If we had worked with $\Hauxii^{a}$, a lot of factors
$(\rho\bar{\rho})$ would be floating around everywhere. The variable
$\Auxii^{a}$ on the other hand was introduced in (\ref{newz}) to
get the most simple form of the inverse transformation, while $\Auxii^{a}$
as a function of $\rho^{\alpha}$ and $\bar{\rho}_{\alpha}$ is in
turn more complicated. $\quad\fussend$%
}:
\begin{eqnarray}
(\rho^{\alpha},\bar{\rho}_{\alpha}) & \mapsto & (\lambda^{\alpha},\bar{\lambda}_{\alpha},\Hauxii^{a},\bar{\Hauxii}_{a})\\
\mbox{with }\lambda^{\alpha} & \equiv & P_{(h)}^{\alpha}(\rho,\bar{\rho})\label{thesameistrue}\\
\Hauxii^{a} & = & \tfrac{1}{2}(\rho\gamma^{a}\rho)\label{Hauxii-simple}
\end{eqnarray}
Its inverse becomes (based on (\ref{Proj-inv-alt-h}) with $f=h$
and using (\ref{Hauxii}))
\begin{eqnarray}
\rho^{\alpha} & = & \lambda^{\alpha}+\tfrac{1}{2}\frac{\Hauxii^{a}(\bar{\lambda}\gamma_{a})^{\alpha}}{(\lambda\bar{\lambda})}\label{inverseWithHauxii}
\end{eqnarray}
whose differential reads
\begin{eqnarray}
\de\rho^{\alpha} & = & \de\lambda^{\beta}\Bigl(\delta_{\beta}^{\alpha}-\tfrac{1}{2}\frac{\Hauxii^{a}(\bar{\lambda}\gamma_{a})^{\alpha}\bar{\lambda}_{\beta}}{(\lambda\bar{\lambda})^{2}}\Bigr)+\tfrac{1}{2}\frac{\de\Hauxii^{a}(\bar{\lambda}\gamma_{a})^{\alpha}}{\lambda\bar{\lambda}}+\nonumber \\
 &  & +\tfrac{1}{2}\frac{\de\bar{\lambda}_{\beta}\Hauxii^{a}}{\lambda\bar{\lambda}}\Bigl(\gamma_{a}^{\alpha\beta}-\frac{(\bar{\lambda}\gamma_{a})^{\alpha}\lambda^{\beta}}{(\lambda\bar{\lambda})}\Bigr)\label{drhoWithHauxii}
\end{eqnarray}
It is still clear that the inverse transformation (\ref{inverseWithHauxii})
is not holomorphic so that it does not seem like a good idea to try
deriving a holomorphic volume form with it. But note that at least
on the constraint surface $\Hauxii^{a}=0$ it becomes holomorphic.
The same is true for the original transformation (\ref{thesameistrue})
because of $\bei{\partiell{\lambda^{\alpha}}{\bar{\rho}_{\beta}}}{\Hauxii^{a}=0}=\bei{\rest_{(h)\bot}^{\alpha\beta}}{\Hauxii^{a}=0}=0$.
Instead the one-forms in (\ref{drhoWithHauxii}) transform even on
the constraint surface non-holomorphically. But at least all $\de\bar{\lambda}_{\alpha}$-appearance
drops:
\begin{equation}
\bei{\de\rho^{\alpha}}{\Hauxii^{a}=0}=\de\lambda^{\alpha}+\tfrac{1}{2}\frac{\de\Hauxii^{a}\left(\bar{\lambda}\gamma_{a}\right)^{\alpha}}{\lambda\bar{\lambda}}\label{drhoOnsurface}
\end{equation}
This suggests that the variable transformation might lead to the holomorphic
$\lambda$-volume form at least on the surface $\Hauxii^{a}=0$. So
the hope is, that after the transformation we have a split of $[d^{16}\rho]$
into the holomorphic volume form $[d^{11}\lambda]$ of \eqref{hol-volform}
and some volume form for the $\Hauxii^{a}$-space.%
\footnote{\label{fn:zeta-volform}We could think about a suitable covariant
holomorphic volume form for the $\Hauxii^{a}$-space without going
through the transformation of the ambient space volume form. It has
to contain five $\de\Hauxii^{a}$. The five indices have to be saturated
and the only invariant tensors are the 10d epsilon tensor $\epsilon_{a_{1}\ldots a_{10}}$,
and the metric. Otherwise we could contract (if we do not allow $\lambda$-dependence)
only with $\Hauxii^{a}$ or its complex conjugate. Contracting with
the $\epsilon$-tensor is of no help, because it leads again to five
free indices and saturating them with $\Hauxii^{a}$ or $\bar{\Hauxii}_{a}$
always leads to vanishing results. It is therefore impossible to construct
a covariant holomorphic volume form for the $\Hauxii^{a}$-space which
is $\lambda^{\alpha}$-independent. If we instead allow $\lambda$-dependence,
a natural candidate would be $(\lambda\gamma_{a_{1}\ldots a_{5}}\lambda)$.
However, because of the constraint $\Hauxii_{a}(\gamma^{a}\lambda)_{\alpha}=0$
this would again vanish. So we need to use also the complex conjugate
$\bar{\lambda}_{\alpha}$ and hope that this non-holomorphic $\bar{\lambda}$-dependence
might drop when dividing by appropriate powers of $(\lambda\bar{\lambda})$
as it is the case in \eqref{hol-volform}. So the only ansatz which
has a chance to be a holomorphic volume form on the $\Hauxii$-space
is 
\[
[d^{5}\Hauxii]\stackrel{?}{=}\frac{1}{(\lambda\bar{\lambda})^{2}}(\bar{\lambda}\gamma_{a_{1}\ldots a_{5}}\bar{\lambda})\de\Hauxii^{a_{1}}\cdots\de\Hauxii^{a_{5}}\qquad\fussend
\]
} Using (\ref{drhoOnsurface}), the holomorphic volume form (\ref{hol-target-volform})
transforms on the constraint surface as 
\begin{equation}
\bei{[d^{16}\rho]}{\Hauxii^{a}=0}=\tfrac{\epsilon_{\alpha_{1}\ldots\alpha_{16}}}{16!}\bigl(\de\lambda^{\alpha_{1}}+\tfrac{\de\Hauxii^{a}(\bar{\lambda}\gamma_{a})^{\alpha_{1}}}{2(\lambda\bar{\lambda})}\bigr)\cdots\bigl(\de\lambda^{\alpha_{16}}+\tfrac{\de\Hauxii^{a}(\bar{\lambda}\gamma_{a})^{\alpha_{16}}}{2(\lambda\bar{\lambda})}\bigr)\qquad
\end{equation}
As $\lambda^{\alpha}$ and $\Hauxii^{a}$ are constrained variables,
they effectively have $11$ and $5$ components respectively and we
have the identities 
\begin{eqnarray}
\de\lambda^{\alpha_{1}}\cdots\de\lambda^{\alpha_{12}} & = & 0\\
\de\Hauxii^{a_{1}}\cdots\de\Hauxii^{a_{6}} & = & 0
\end{eqnarray}
Using these, the holomorphic volume form of the ambient space becomes
on the constraint surface
\begin{equation}
\bei{[d^{16}\rho]}{\Hauxii^{a}=0}\!\!\!=\!\!\tfrac{1}{32\cdot5!11!(\lambda\bar{\lambda})^{5}}\!(\bar{\lambda}\gamma_{a_{1}})^{\alpha_{1}}\!\!\cdots\!(\bar{\lambda}\gamma_{a_{5}})^{\alpha_{5}}\!\epsilon_{\alpha_{1}\ldots\alpha_{16}}\!\de\lambda^{\alpha_{6}}\cdots\de\lambda^{\alpha_{16}}\de\Hauxii^{a_{1}}\cdots\de\Hauxii^{a_{5}}\qquad
\end{equation}
Any matrix with two upper indices can be expanded in $\gamma_{a}^{\alpha\beta},\gamma_{abc}^{\alpha\beta}$
and $\gamma_{abcde}^{\alpha\beta}$ where only the middle one is antisymmetric.
This means that any antisymmetric matrix $A^{[\alpha\beta]}$ can
be written as (see e.g. (D.142) of \cite{Guttenberg:2008ic}) 
\begin{equation}
A^{[\alpha\beta]}=\tfrac{1}{16\cdot3!}\gamma_{abc}^{\alpha\beta}\gamma_{\gamma\delta}^{abc}A^{\gamma\delta}
\end{equation}
So we can in particular replace 
\begin{eqnarray}
\tfrac{1}{(\lambda\bar{\lambda})^{2}}(\bar{\lambda}\gamma_{a_{1}})^{[\alpha_{1}}(\bar{\lambda}\gamma_{a_{2}})^{\alpha_{2}]} & = & \tfrac{1}{16\cdot3!(\lambda\bar{\lambda})^{2}}(\bar{\lambda}\gamma_{a_{1}}\gamma^{bcd}\gamma_{a_{2}}\bar{\lambda})\gamma_{bcd}^{\alpha_{1}\alpha_{2}}
\end{eqnarray}
to obtain 
\begin{eqnarray}
\bei{[d^{16}\rho]}{\Hauxii^{a}=0}\!\!\! & \!\!= & \!\!\!\!\tfrac{1}{5!11!3!32\cdot16}\tfrac{1}{(\lambda\bar{\lambda})^{3}}\gamma_{bcd}^{\alpha_{1}\alpha_{2}}(\bar{\lambda}\gamma_{a_{3}})^{\alpha_{3}}\cdots(\bar{\lambda}\gamma_{a_{5}})^{\alpha_{5}}\epsilon_{\alpha_{1}\ldots\alpha_{16}}\de\lambda^{\alpha_{6}}\cdots\de\lambda^{\alpha_{16}}\times\nonumber \\
 &  & \!\!\!\!\times\tfrac{1}{(\lambda\bar{\lambda})^{2}}(\bar{\lambda}\gamma_{a_{1}}\hoch{bcd}\tief{a_{2}}\bar{\lambda})\de\Hauxii^{a_{1}}\cdots\de\Hauxii^{a_{5}}
\end{eqnarray}
This is already reasonably close to a split into (\ref{hol-volform})
and a volume form for the $\Hauxii^{a}$-space, in particular if one
has in mind footnote \ref{fn:zeta-volform} in which we argued that
the $\Hauxii^{a}$-volume form is expected to be $\lambda$-dependent.
There just remains the problem that the contraction of the indices
${\scriptstyle bcd}$ is between the two {}``factors'' so that factorization
is not yet obvious. We believe, however, that because of the constraints
on $\lambda^{\alpha}$ \textbf{and }on $\Hauxii^{a}$ there will be
identities similar to (\ref{toy:identity3}) for the toy-model which
will rearrange the contractions and thus solve this issue. We leave
this for further study.

\section{Conclusions}

In this article we have introduced a family $P_{(f)}^{\alpha}$ (parametrized
by a function $f$) of covariant non-linear projections to the pure
spinor space whose linearization on the constraint surface reduced
for all of them to the transpose of the known projector to the gauge
invariant part of the antighost $\omega_{z\alpha}$. In addition we
introduced similar linear projectors to gauge invariant parts for
the non-minimal variables $\bar{\omega}_{z}^{\alpha}$ and $\bs s_{z}^{\alpha}$. 

A priori the simple choice $f=1$ seems natural for the non-linear
projection, but does not lead to any particular properties. Instead
we presented a non-trivial function $h$ for which the projection
can be derived from a potential $\Phi$ and has Hermitian Jacobian
matrices. Hermiticity was essential to derive the explicit form for
the additional gauge transformations that one obtains if one replaces
the pure spinors in the string action by projections of an unconstrained
spinor. We have discussed possible gauge fixings but it remains to
see whether one can transform in this way to a formalism which has
some advantage to the original. 

Regarding the pure spinor space (parametrized by $\lambda^{\alpha}$)
as being embedded into an ambient space $\mathbb{C}^{16}$ (parametrized
by $\rho^{\alpha}$), the projection becomes part of a variable transformation:
$\rho^{\alpha}\mapsto(\lambda^{\alpha},\auxii^{a})$. Depending on
the application the latter variable was also redefined to either $\Auxii^{a}$
or $\Hauxii^{a}$. All of them have in common that they are constrained.
We have derived the transformation of the holomorphic volume form
under of $\mathbb{C}^{16}$ under this variable transformation, hoping
that it would factorize into the holomorphic volume form of the pure
spinor space and a volume form of the rest. In the toy-model such
a factorization worked perfectly. In the pure spinor case the result
is promising, but contains contractions between the two {}``factors''
which hopefully can be resolved via some identities. If this could
be completed, the corresponding reparametrization might by a way to
derive path-integral measures for higher genus. We leave this for
further study. 

Having a projection to the pure spinor space allows to obtain solutions
to the pure spinor constraint without switching to the U(5) formalism.
The latter is based on a Fock space representation whose vacuum has
to be a pure spinor itself. Having concrete pure spinors in the standard
representation should thus allow to derive concrete intertwiners between
the standard and the Fock space representation. 

One further application of the present formalism might be the projection
of unphysical contributions in the computation of the force between
D-branes. Indeed, this check in the context of the pure spinor formalism
has never been done. Similarly, the computation of the partition function
\cite{Aisaka:2008vw} can be probably expressed in closed form using
this non-linear projections. We hope to report on these applications
in the future.

\subsection*{Acknowledgements}

S.G. would like to thank Nicolas Orantin and Stavros Papadakis for
information about the website of sequences oeis.org \cite{oeis:2011}
(see footnote \ref{fn:series} on page \pageref{fn:series}) and further
Stavros for additional discussions and for testing numerically the
inequality $\auxi\leq1$ in (\ref{xin01}) or (\ref{xsmaller1}).
The work of S.G was financed in the initial phase by a {}``Scholarship
for action C 'Attracting foreign researchers' '' of the University
of Torino \rem{reference number?} and in the final phase by a research
fellowship of the Portuguese research foundation FCT (reference number
SFRH/BPD/63719/2009). 

\appendix

\section{Toy model}

\label{app:toy-model}\global\long\def\toydim{\mc N}
Many of the properties of pure spinor space can also be observed for
a simple toy model which helped us to find the projection that we
presented in the main part. In this appendix we will be less rigorous
in some statements and put more emphasis to present the logic that
we followed. 

Let us consider complex variables $\lambda^{I}\in\mathbb{C},\: I\in\{1,\ldots,\toydim\}$
obeying the quadratic constraint 
\begin{equation}
\lambda^{2}\equiv\lambda^{I}\delta_{IJ}\lambda^{J}=0\label{toy:constraint}
\end{equation}
We will denote the complex conjugate by $\bar{\lambda}_{I}$ and use
the metric $\delta_{IJ}$ to pull indices up or down, i.e. $\lambda_{I}=\delta_{IJ}\lambda^{J}$.
In principle there is thus no need to distinguish between upper and
lower indices, but we will still do so, because we understand the
Einstein summation convention to be applied only when an index is
repeated with opposite vertical position.

If one wants to explicitly solve the constraint, it is convenient
to introduce 'lightcone-coordinates' 
\begin{equation}
\lambda^{\pm}\equiv\tfrac{1}{\sqrt{2}}\left(\lambda^{\toydim-1}\pm i\lambda^{\toydim}\right)\label{toy:lightcone-coord}
\end{equation}
 in which the metric becomes off-diagonal for these two entries,
i.e. $\lambda^{+}=\lambda_{-}$, $\lambda^{-}=\lambda_{+}$ while
it remains diagonal for the remaining $\toydim-2$ indices ${\scriptstyle i}$,
i.e. $\lambda^{i}=\lambda_{i}$:
\begin{equation}
\lambda^{2}=2\lambda^{+}\lambda^{-}+\lambda^{i}\lambda_{i}=0
\end{equation}
For $\toydim=3$ this is the model discussed for example in section
3.1.3. of \cite{Grassi:2007va}. In a patch where $\lambda^{+}\neq0$,
one can thus solve the constraint $\lambda^{2}=0$ for $\lambda^{-}$
as a function of $\lambda^{+}$ and the $\lambda^{i}$'s. Let us give
this function the name $\lambda_{{\rm sol}}^{-}$, as we will later
refer to it: 
\begin{equation}
\lambda_{{\rm sol}}^{-}(\lambda^{+},\lambda^{i})\equiv-\tfrac{1}{2\lambda^{+}}\lambda^{i}\lambda_{i}\label{toy:lamsol}
\end{equation}
These coordinates correspond in the pure spinor case to the U(5)-covariant
coordinates. 

A covariant derivative%
\footnote{The covariant derivative $D_{\lambda^{I}}$ corresponds in the pure
spinor sigma model to the gauge invariant part of the antighost $\omega_{z}$.
The projection to this gauge invariant part as presented in \cite{Oda:2001zm}
and \cite{Berkovits:2010zz} was the starting point for this article.
$\quad\fussend$%
} which respects the constraint (\ref{toy:constraint}) can be constructed
using some reference vector $\bar{\chi}_{I}$:
\begin{eqnarray}
D_{\lambda^{I}} & \equiv & \partial_{\lambda^{I}}-\tfrac{1}{(\lambda\bar{\chi})}\lambda_{I}\bar{\chi}^{J}\partial_{\lambda^{J}}\equiv\Pi_{(\bar{\chi})\bot\, I}^{T}\hoch J\partial_{\lambda^{J}}
\end{eqnarray}
This defines the projection matrix $\Pi_{(\bar{\chi})\bot}^{T}$ whose
transpose 
\begin{equation}
\Pi_{(\bar{\chi})\bot}=\one-\tfrac{\bar{\chi}\otimes\lambda}{(\bar{\chi}\lambda)}
\end{equation}
 projects a general vector to one which is orthogonal (therefore the
subscript $\bot$) to $\lambda^{K}$ or equivalently 
\begin{eqnarray}
\lambda^{K}\delta_{KI}\Pi_{(\bar{\chi})\bot}\hoch I\tief J & = & 0
\end{eqnarray}
and in addition is idempotent, as one can easily check:
\begin{equation}
(\Pi_{(\bar{\chi})\bot})^{2}=\Pi_{(\bar{\chi})\bot}
\end{equation}
Note that in particular every variation of $\lambda^{I}$ consistent
with (\ref{toy:constraint}) is orthogonal to $\lambda^{I}$ (i.e.
$\lambda^{I}\delta_{IJ}\delta\lambda^{J}=0$). $\Pi_{(\bar{\chi})\bot}$
thus acts on $\delta\lambda$ like the identity:
\begin{equation}
\Pi_{(\bar{\chi})\bot}\hoch I\tief J\delta\lambda^{J}=\delta\lambda^{J}\label{toy:PiDeltalamIsDeltalam}
\end{equation}
A general variation of a function of $\lambda^{I}$ is thus of the
form 
\begin{equation}
\delta\lambda^{T}\partial_{\lambda}=(\Pi_{(\bar{\chi})\bot}\delta\lambda)^{T}\partial_{\lambda}=\delta\lambda^{T}(\Pi_{(\bar{\chi})\bot}^{T}\partial_{\lambda})=\delta\lambda^{T}D_{\lambda}
\end{equation}
and thus reproduces the covariant derivative from which we started. 

We will not restrict $\bar{\chi}$ to be only a constant reference
vector, but allow e.g. also $\bar{\chi}_{I}=\bar{\lambda}_{I}$. As
this will be the most important case, we denote 
\begin{equation}
\Pi_{\bot}\equiv\Pi_{(\bar{\lambda})\bot}=\one-\tfrac{\bar{\lambda}\otimes\lambda}{(\bar{\lambda}\lambda)}\label{toy:Pionshell}
\end{equation}
Particular properties of this projection matrix are its Hermiticity
and the fact that (seen as a function of $\lambda$ and $\bar{\lambda}$)
it is homogeneous of degree (0,0), i.e. 
\begin{equation}
\Pi_{\bot}(c\lambda,\bar{c}\bar{\lambda})=\Pi_{\bot}(\lambda,\bar{\lambda})\quad\forall c\in\mathbb{C}
\end{equation}
A natural question is now, whether one can integrate (\ref{toy:PiDeltalamIsDeltalam}).
Or in other words, whether we can find a projection 
\begin{equation}
P_{(\bar{\chi})}^{I}:\quad\rho^{I}\mapsto\lambda^{I}\equiv P_{(\bar{\chi})}^{I}(\rho)
\end{equation}
such that its variation just produces the linear projector $\Pi_{\bot}$.
I.e. if we define $\Pi_{(\bar{\chi})\bot}(\rho)$ via $\delta P_{(\bar{\chi})}^{I}(\rho)\equiv(\Pi_{(\bar{\chi})\bot}(\rho)\delta\rho)^{I}$
or equivalently 
\begin{equation}
\Pi_{(\bar{\chi})\bot}(\rho)^{I}\tief J\equiv\partial_{\rho^{J}}P_{(\bar{\chi})}^{I}(\rho)
\end{equation}
then this matrix should reduce on the constraint surface $\rho=\lambda$
to $\Pi_{(\bar{\chi})\bot}(\lambda)$. So the condition is 
\begin{equation}
\bei{\partial_{\rho^{J}}P_{(\bar{\chi})}^{I}(\rho)}{\rho=\lambda}\stackrel{!}{=}\delta_{J}^{I}-\tfrac{\bar{\chi}^{I}\lambda_{J}}{(\bar{\chi}\lambda)}\equiv\Pi_{(\bar{\chi})\bot}\hoch I\tief J\label{toy:to-solve}
\end{equation}
The naive extension of the righthand side off the constraint surface
to $\partial_{\rho^{J}}P_{(\bar{\chi})}^{I}(\rho)\stackrel{!}{=}\delta_{J}^{I}-\tfrac{\bar{\chi}^{I}\rho_{J}}{(\bar{\chi}\rho)}$
is not so easily integrable, but already the terms one obtains by
integrating it while treating the $\rho$-dependent denominator as
a constant leads for certain $\bar{\chi}$ to a solution of (\ref{toy:to-solve})
\begin{equation}
P_{(\bar{\chi})}^{I}(\rho)\equiv\rho^{I}-\frac{1}{2(\bar{\chi}\rho)}\rho^{2}\bar{\chi}^{I}\quad(\mbox{if }\bar{\chi}^{I}\bar{\chi}_{I}=0)\label{toy:PchiWithpurechi}
\end{equation}
However, the extra condition $\bar{\chi}^{I}\bar{\chi}_{I}=0$ is
necessary for $\lambda^{I}\equiv P_{(\bar{\chi})}^{I}(\rho)$ to obey
the constraint $\lambda^{2}=0$ for all $\rho$. It is clear that
we need additional terms in the case where $\bar{\chi}$ is not a
constrained vector, e.g. for the covariant choice $\bar{\chi}=\bar{\rho}$.
An idea would be to replace $\bar{\chi}$ in turn by a projection
that uses $\rho$ as reference spinor $\bar{\chi}_{I}\to\bar{\chi}_{I}-\frac{1}{2(\bar{\chi}\rho)}(\bar{\chi})^{2}\rho_{I}$
and then again replace $\rho_{I}$ by the projection (\ref{toy:PchiWithpurechi})
and then continue iteratively. The main insight from this procedure
is that we will obtain something of the form\enlargethispage*{1cm}
\begin{eqnarray}
P_{(\bar{\chi})}^{I}(\rho) & \equiv & f(\auxi)\rho^{I}-g(\auxi)\frac{\rho^{2}}{(\rho\bar{\chi})}\bar{\chi}^{I}\label{toy:projector-ansatz}\\
\mbox{with }\auxi & \equiv & \frac{\rho^{2}\bar{\chi}^{2}}{(\rho\bar{\chi})^{2}}\label{toy:xz-def}
\end{eqnarray}
with some functions $f,g$ to be determined. We will understand this
as an ansatz from now on and can forget about the iterative motivation%
\footnote{\label{fn:series}Coming from the iteration idea of the projection
(\ref{toy:PchiWithpurechi}), our original approach was via a power
series ansatz\vspace{-.2cm} 
\[
P_{(\bar{\chi})}^{I}(\rho)\equiv\rho^{I}\bigl({\scriptstyle \sum_{k=0}^{\infty}}\alpha_{k}\auxi^{k}\bigr)-\bar{\chi}^{I}\tfrac{\rho^{2}}{(\rho\bar{\chi})}\bigl({\scriptstyle \sum_{k=0}^{\infty}}\beta_{k}\auxi^{k}\bigr),\quad\alpha_{k},\beta_{k}\in\mathbb{C}
\]
In order to be a projection, one needs that for constrained $\lambda^{I}$
(i.e. with $\lambda^{2}=0$) it should reduce to the identity map
$P_{(\bar{\chi})}^{I}(\lambda)\stackrel{!}{=}\lambda^{I}$, which
fixes $\alpha_{0}=1$. From the condition that $(P_{(\bar{\chi})}(\rho))^{2}=0$
one obtains $\beta_{0}=\tfrac{1}{2}$ and in addition (using already
$\alpha_{0}=1$,$\beta_{0}=\tfrac{1}{2}$) the recursion relation
$\left(\alpha_{n}-2\beta_{n}\right)=-\sum_{k=1}^{n-1}\alpha_{k}\left(\alpha_{n-k}-2\beta_{n-k}\right)-\sum_{k=0}^{n-1}\beta_{k}\beta_{n-1-k}$. 

If one fixes by hand $\alpha_{k}=0\quad\forall k\geq1$ and makes
a reparametrization $\beta_{n}\equiv\frac{\gamma_{n}}{2^{2n+1}}$,
it leads to the new series $\gamma_{n+1}=\sum_{k=0}^{n}\gamma_{k}\gamma_{n-k}\:(\gamma_{0}=1)$,
whose first entries (starting with $\gamma_{0}$) are $\{1,1,2,5,14,42,132,429,1430,4862,...\}$.
 According to oeis.org \cite{oeis:2011} (thanks to Nicolas Orantin
and Stavros Papadakis for information about this web-site!), these
are the Catalan numbers or Segner numbers $\gamma_{n}=\frac{(2n)!}{n!(n+1)!}$.
For them it is well known that $\sum_{n=0}^{\infty}\gamma_{n}\auxi^{n}=\frac{2}{1+\sqrt{1-4\auxi}}$.
Coming back to $\beta_{n}=\frac{\gamma_{n}}{2^{2n+1}}$, we obtain
$\sum_{n=0}^{\infty}\beta_{n}\auxi^{n}=\frac{1}{1+\sqrt{1-\auxi}}$.
Therefore the power series ansatz for our projector with the manual
choice $\alpha_{k}=0\quad\forall k\geq1$ turns into 
\[
P_{(\bar{\chi})}^{I}(\rho)=\rho^{I}-\bar{\chi}^{I}\tfrac{\rho^{2}}{(\rho\bar{\chi})}\tfrac{1}{1+\sqrt{1-\auxi}}\qquad\fussend
\]
}. The variable $\auxi$ is homogeneous of degree $0$ w.r.t. $\rho$
and also w.r.t. $\bar{\chi}$. The above ansatz for $P_{(\bar{\chi})}^{I}(\rho)$
is therefore homogeneous of degree $1$ in $\rho$ and $0$ in $\bar{\chi}$.

For a vector $\lambda^{I}$ that lies already on the constraint-surface
$\lambda^{2}\equiv\lambda^{I}\delta_{IJ}\lambda^{J}=0$,  the map
$P_{(\bar{\chi})}^{I}$ (being a projection) should be the identity:
\begin{eqnarray}
\lambda^{I} & \stackrel{!}{=} & P_{(\bar{\chi})}^{I}(\lambda)=\lambda^{I}f\left(0\right)
\end{eqnarray}
This determines 
\begin{equation}
f(0)=1
\end{equation}
The other requirement for $P_{(\bar{\chi})}^{I}$ is that the image
$P_{(\bar{\chi})}^{I}(\rho)$ lies on the constraint surface for every
$\rho^{I}$, i.e.
\begin{eqnarray}
0 & \stackrel{!}{=} & P_{(\bar{\chi})}^{I}(\rho)P_{(\bar{\chi})I}(\rho)=\\
 & = & \rho^{2}\left(f(\auxi)^{2}-2f(\auxi)g(\auxi)+\auxi g(\auxi)^{2}\right)
\end{eqnarray}
So at least away from the constraint surface (i.e. for $\rho^{2}\neq0$)
we need the bracket to vanish and thus obtain a priori two solutions
for $f$ in terms of $g$ or vice versa 
\begin{eqnarray}
f(\auxi) & = & g(\auxi)\left(1\pm\sqrt{1-\auxi}\right)
\end{eqnarray}
The condition $f(0)=1$ fixes the sign to be the upper one. This
fixes $g$ uniquely in terms of $f$: 
\begin{equation}
g(\auxi)=\frac{f(\auxi)}{1+\sqrt{1-\auxi}},\quad f(0)=1
\end{equation}
Plugging this back into the ansatz (\ref{toy:projector-ansatz}),
we obtain a \textbf{family of projections} that still depends on the
choice of the reference spinor $\bar{\chi}$ and the function $f$:
\begin{eqnarray}
P_{(f,\bar{\chi})}^{I}(\rho) & \equiv & f(\auxi)\left(\rho^{I}-\frac{1}{1+\sqrt{1-\auxi}}\frac{\rho^{2}}{(\rho\bar{\chi})}\bar{\chi}^{I}\right)\quad\mbox{with }\auxi\equiv\frac{\rho^{2}\bar{\chi}^{2}}{(\rho\bar{\chi})^{2}}
\end{eqnarray}
For $\bar{\chi}=\bar{\rho}$, it is convenient to define 
\begin{equation}
\auxii\equiv\frac{\rho^{2}}{(\rho\bar{\rho})}\label{toy:z-def}
\end{equation}
which implies 
\begin{equation}
\auxi=\auxii\bar{\auxii}
\end{equation}
and the projector becomes%
\footnote{\label{fn:toy:xis1}When $\bar{\rho}_{I}$ is really the complex conjugate
of $\rho^{I}$, then we have $\auxi\leq1$ 
\begin{eqnarray*}
(\rho\bar{\rho})^{2} & = & \bigl({\scriptstyle \sum_{I}}\abs{\rho^{I}}^{2}\bigr)^{2}=\bigl({\scriptstyle \sum_{I}}\abs{(\rho^{I})^{2}}\bigr)^{2}\geq\abs{{\scriptstyle \sum_{I}}(\rho^{I})^{2}}^{2}=\abs{\rho^{2}}^{2}=\rho^{2}\bar{\rho}^{2}\\
\dann(\rho\bar{\rho})^{2} & \geq & \rho^{2}\bar{\rho}^{2}\quad\dann\auxi\geq1
\end{eqnarray*}
The fact $\auxi\leq1$ makes the projection (\ref{toy:projGeneral})
particularly well-behaved. But note that this statement is true only
if the reference spinor $\bar{\chi}=\bar{\rho}$ is really the complex
conjugate of $\rho$.$\quad\fussend$%
}$\hoch ,$%
\footnote{Let us discuss two special cases. 
\begin{itemize}
\item For $\toydim=1$ the constraint $0=\lambda^{2}=(\Re(\lambda)+i\Im(\lambda))^{2}=\Re(\lambda)^{2}-\Im(\lambda)^{2}+2i\Re(\lambda)\Im(\lambda)$
implies $\Re(\lambda)=\Im(\lambda)=0\dann\lambda=0$. Consistent with
this, the projection maps every complex number $\rho$ to $0$. Note
that at quantum level, or even in the ring of supernumbers, $\lambda^{2}=0$
has also nontrivial solutions. 
\item In dimension $\toydim=2$, the solution space to $0=\lambda^{2}=2\lambda^{+}\lambda^{-}$
is the union of $\lambda^{+}=0$ and $\lambda^{-}=0$. The projection
\eqref{toy:projGeneral} reads for $f=1$ 
\begin{eqnarray*}
P_{(f=1)}^{+}(\rho,\bar{\rho}) & = & \rho^{+}-\bar{\rho}^{+}\frac{2\rho^{+}\rho^{-}}{\rho^{+}\bar{\rho}_{+}+\rho^{-}\bar{\rho}_{-}+\sqrt{(\rho^{+}\bar{\rho}_{+}-\rho^{-}\bar{\rho}_{-})^{2}}}\\
 & = & \rho^{+}-\bar{\rho}_{-}\frac{2\rho^{+}\rho^{-}}{\rho^{+}\bar{\rho}_{+}+\rho^{-}\bar{\rho}_{-}+\abs{\underbrace{\rho^{+}\bar{\rho}_{+}-\rho^{-}\bar{\rho}_{-}}_{{\rm real}}}}=\\
 & = & \left\{ \begin{array}{ccc}
\rho^{+}-\frac{\bar{\rho}_{-}\rho^{-}}{\bar{\rho}_{+}} & {\rm for} & \abs{\rho^{+}}>\abs{\rho^{-}}\\
0 & {\rm for} & \abs{\rho^{+}}\leq\abs{\rho^{-}}
\end{array}\right.\\
P_{(f=1)}^{-}(\rho,\bar{\rho}) & = & \left\{ \begin{array}{ccc}
0 & {\rm for} & \abs{\rho^{+}}\geq\abs{\rho^{-}}\\
\rho^{-}-\frac{\bar{\rho}_{+}\rho^{+}}{\bar{\rho}_{-}} & {\rm for} & \abs{\rho^{+}}<\abs{\rho^{-}}
\end{array}\right.
\end{eqnarray*}
On the constraint surface this implies:
\[
\left(\zwek{P_{(f=1)}^{+}(\lambda^{+},0)}{P_{(f=1)}^{-}(\lambda^{+},0)}\right)=\left(\zwek{\lambda^{+}}0\right)\quad,\quad\left(\zwek{P_{(f=1)}^{+}(0,\lambda^{-})}{P_{(f=1)}^{-}(0,\lambda^{-})}\right)=\left(\zwek 0{\lambda^{-}}\right)\quad\fussend
\]
\end{itemize}
} 
\begin{equation}
\lambda^{I}\equiv P_{(f)}^{I}(\rho,\bar{\rho})\equiv f(\auxi)\left(\rho^{I}-\frac{\auxii}{1+\sqrt{1-\auxi}}\bar{\rho}^{I}\right)\label{toy:projGeneral}
\end{equation}
As long as $\bar{\rho}$ is treated independently from $\rho$ and
as long as we are not talking about reality properties, there is no
loss of generality when using this second form. We can simply see
$\bar{\rho}$ as an independent vector and are back at $P_{(f,\bar{\chi})}$.
Then $\bar{\auxii}$ is not the complex conjugate any more but it
is still formally defined in the same way as before, i.e. as $\bar{\auxii}\equiv\frac{\bar{\chi}^{2}}{(\rho\bar{\chi})}$.
The notation will thus reflect only which point of view we take. We
will stress it, as soon as it makes a fundamental difference.  The
absolute value square of the image of the projection is 

\begin{equation}
(\lambda\bar{\lambda})=2(\rho\bar{\rho})\abs{f(\auxi)}^{2}\frac{1-\auxi}{1+\sqrt{1-\auxi}}\label{toy:absval}
\end{equation}
The absolute value (\ref{toy:absval}) shows that if $f(1)$ is non-singular,
the projection (\ref{toy:projGeneral}) maps to the origin for $\auxi=1$.
\begin{equation}
\bei{P_{(f)}^{I}(\rho,\bar{\rho})}{\auxi=1}=0
\end{equation}
In particular when $f$ has no zero's (like the valid choice $f=1$),
$\auxi=1$ is the only case where the projection vanishes (apart from
$(\rho\bar{\rho})=0$ for which $\auxi$ is not well-defined). Using
this insight and plugging $\auxi=1$ and $\lambda^{I}=0$ back into
(\ref{toy:projGeneral}), we obtain $0=\rho^{I}-\auxii\bar{\rho}^{I}$
and therefore 
\begin{equation}
\auxi=1\quad\iff\quad\rho^{I}=\auxii\bar{\rho}^{I}\quad(\mbox{or }\bar{\rho}_{I}=\bar{\auxii}\rho_{I})
\end{equation}
\frem{The righthand side should further be equivalent to $\rho^{I}$
being a complex multiple of a real vector, i.e. $\rho^{I}=c\cdot\tilde{\rho}^{I}$
with $c\in\mathbb{C}$ and real $\tilde{\rho}^{I}$.}

\paragraph{Projection as part of a variable transformation}

One can partially invert the projection and express $\rho^{I}$ as
a linear combination of $\lambda^{I}$ and $\bar{\lambda}_{I}$ with
$\auxii$-dependent coefficients (having in mind that $\auxi=\auxii\bar{\auxii}$):

\begin{equation}
\rho^{I}(\lambda,\bar{\lambda},\auxii,\bar{\auxii})=\frac{1+\sqrt{1-\auxi}}{2f(\auxi)\sqrt{1-\auxi}}\lambda^{I}+\frac{\auxii}{2\bar{f}(\auxi)\sqrt{1-\auxi}}\bar{\lambda}^{I}\label{toy:inverse}
\end{equation}
 The unconstrained vector $\rho^{I}$ can thus be seen as a non-holomorphic
function of $\lambda^{I}$ and $\auxii$ and one can regard the projection
as part of a variable transformation $\rho^{I}\mapsto(\lambda^{I},\auxii)$
and the above equation as the inverse variable transformation. This
means that in (\ref{toy:inverse}) $\auxii$ is regarded as an independent
variable, while in (\ref{toy:projGeneral}) it was just a placeholder
for $\frac{\rho^{2}}{(\rho\bar{\rho})}$ and $\lambda^{I}$ was regarded
as a function of $\rho^{I}$ only. 

The inverse transformation of the absolute value squared is obviously
(looking at (\ref{toy:absval})) given by 
\begin{equation}
(\rho\bar{\rho})=\frac{(\lambda\bar{\lambda})\left(1+\sqrt{1-\auxi}\right)}{2\abs{f(\auxi)}^{2}\left(1-\auxi\right)}\label{toy:absval-inv}
\end{equation}

\paragraph{Alternative reparametrization with simple inverse}

If one rewrites the righthand side of equation (\ref{toy:inverse})
as $\frac{1+\sqrt{1-\auxi}}{2f(\auxi)\sqrt{1-\auxi}}\bigl(\lambda^{I}+\frac{f(\auxi)}{\os{}{\bar{f}(\auxi)}}\frac{\auxii}{1+\sqrt{1-\auxi}}\bar{\lambda}^{I}\bigr)$,
it becomes obvious that the inverse variable transformation would
become particularly simple upon choosing 
\begin{equation}
\Auxii\equiv\frac{f(\auxi)}{\bar{f}(\auxi)}\frac{\auxii}{1+\sqrt{1-\auxi}}=\frac{f(\auxi)}{\bar{f}(\auxi)}\frac{\rho^{2}}{(\rho\bar{\rho})+\sqrt{(\rho\bar{\rho})^{2}-\rho^{2}\bar{\rho}^{2}}}\label{toy:Auxii}
\end{equation}
as a new variable. The absolute value square of this variable is
\begin{equation}
\Auxi\equiv\Auxii\bar{\Auxii}=\frac{\auxi}{(1+\sqrt{1-\auxi})^{2}}=\frac{1-\sqrt{1-\auxi}}{1+\sqrt{1-\auxi}}\label{toy:Auxi}
\end{equation}
Because of $0\leq\auxi\leq1$ also $\Auxi$ obeys 
\begin{equation}
0\leq\Auxi\leq1
\end{equation}
In order to invert the relation between $\Auxii$ and $\auxii$ in
\eqref{toy:Auxii}, it is useful to invert first the above relation
\eqref{toy:Auxi} of their absolute values. Multiplying both sides
of \eqref{toy:Auxi} with the denominator $\left(2-\auxi+2\sqrt{1-\auxi}\right)$
and putting the term that contains the square root on the left-hand
side and the rest on the righthand side, one obtains $2\Auxi\sqrt{1-\auxi}=\auxi(1+\Auxi)-2\Auxi$.
Squaring and putting everything to the left yields $\auxi\left(\auxi(1+\Auxi)^{2}-4\Auxi\right)=0$.
At least for $\auxi\neq0$, the inverse transformation of the absolute
value square is thus given by 
\begin{equation}
\auxi=\frac{4\Auxi}{(1+\Auxi)^{2}}
\end{equation}
Together with $0\leq\Auxi\leq1$ this implies 
\begin{equation}
\sqrt{1-\auxi}=\frac{1-\Auxi}{1+\Auxi}
\end{equation}
Plugging this into \eqref{toy:Auxii} and solving for $\auxii$,
we arrive at 
\begin{eqnarray}
\auxii & = & \frac{\bar{\tilde{f}}(\tilde{\auxi})}{\tilde{f}(\tilde{\auxi})}\frac{2\Auxii}{(1+\Auxi)}
\end{eqnarray}
where 
\begin{equation}
\fa(\Auxi)\equiv f(\auxi(\Auxi))
\end{equation}
We introduced the new variable $\Auxii$ in order to obtain a simple
inverse transformation of $\rho^{I}\mapsto(\lambda^{I},\tilde{\zeta})$,
and indeed, after plugging the above expressions into \eqref{toy:inverse}
this inverse transformation becomes
\begin{eqnarray}
\rho^{I}(\lambda,\bar{\lambda},\Auxii,\bar{\Auxii}) & = & \frac{1}{\fa(\Auxi)(1-\Auxi)}\left(\lambda^{I}+\Auxii\bar{\lambda}^{I}\right)\label{toy:inverseAlt}
\end{eqnarray}
The absolute value squared is given by
\begin{eqnarray}
(\rho\bar{\rho}) & = & \frac{(1+\Auxi)}{\abs{\fa(\Auxi)}^{2}(1-\Auxi)^{2}}(\lambda\bar{\lambda})
\end{eqnarray}
Apparently the inverse transformation \eqref{toy:inverseAlt} simplifies
even further if the function $f(\auxi)$ (which determines the precise
form of the projection \eqref{toy:projGeneral} is chosen to coincide
with $\frac{1}{(1-\Auxi)}$ in the new variables. We will later rediscover
this function in the context of Hermiticity, so we give it the name
$h$: 
\begin{equation}
h(\auxi)\equiv\frac{1+\sqrt{1-\auxi}}{2\sqrt{1-\auxi}}=\frac{1}{(1-\Auxi)}
\end{equation}
In this case the inverse transformation \eqref{toy:inverseAlt} is
simply 
\begin{equation}
\rho^{I}(\lambda,\bar{\lambda},\Auxii,\bar{\Auxii})=\lambda^{I}+\Auxii\bar{\lambda}^{I}\quad\mbox{if }f=h\label{toy:inverseAltH}
\end{equation}
\rem{Starting from this form of the inverse transformation and with
the constraints on $\lambda^{I}$ and $\Auxii$, can we get back the
projection?}with $(\rho\bar{\rho})=(\lambda\bar{\lambda})\left(1+\Auxi\right)$.

It is interesting to see that from the simple inverse transformations
\eqref{toy:inverseAlt} (for general $f$) or \eqref{toy:inverseAltH}
(for $f=h$) one can indeed reconstruct the original (in terms of
$\rho^{I}$ quite complicated) form of the projection \eqref{toy:projGeneral},
just exploiting the constraint $\lambda^{2}=0$. As this is basically
inverting the inverse transformation, it is clear that it has to work.
The reader might perhaps convince himself for the special case $f=h$
that plugging \eqref{toy:inverseAltH} into \eqref{toy:Auxii} (with
$f=h$ written in terms of $\rho,\bar{\rho}$) indeed gives $\Auxii$
and that plugging \eqref{toy:inverseAltH} into \eqref{toy:projGeneral}
indeed gives $\lambda^{I}$.

\paragraph{Jacobian matrices }

The projection $P^{I}$ in the manifold is naturally accompanied by
push-forward maps on vectors in the tangent space, which are given
by the Jacobian matrix of the projection: 
\begin{eqnarray}
\delta\lambda^{I} & \equiv & \delta P_{(f)}^{I}(\rho,\bar{\rho})\equiv\Pi_{(f)\bot K}^{I}\delta\rho^{K}+\rest_{(f)\bot}^{IK}\delta\bar{\rho}_{K}\\
\Pi_{(f)\bot K}^{I}(\rho,\bar{\rho}) & \equiv & \partial_{\rho^{K}}P_{(f)}^{I}(\rho,\bar{\rho})\label{toy:Pidef}\\
\rest_{(f)\bot}^{IK}(\rho,\bar{\rho}) & \equiv & \partial_{\bar{\rho}_{K}}P_{(f)}^{I}(\rho,\bar{\rho})\label{toy:pidef}
\end{eqnarray}
The subscript '$\bot$' indicates the fact that the matrices $\Pi_{(f)\bot}$
and $\rest_{(f)\bot}$ are mapping to subspaces which are orthogonal
to $\lambda$
\begin{equation}
\lambda_{I}\delta\lambda^{I}=0\,,\quad P_{I}(\rho,\bar{\rho})\Pi_{(f)\bot K}^{I}(\rho,\bar{\rho})=P_{I}(\rho,\bar{\rho})\rest_{(f)\bot}^{IK}(\rho,\bar{\rho})=0\quad\forall\rho\label{toy:orth}
\end{equation}
In order to calculate $\Pi_{(f)\bot}$ and $\rest_{(f)\bot}$ we repeatedly
need the following partial derivatives:%
\footnote{\label{fn:toy:rho-derivs-lam}Using (\ref{toy:inverse}) and (\ref{toy:absval-inv}),
the derivatives (\ref{toy:zbyrho}) to (\ref{toy:xbyrho}) can be
rewritten as 
\begin{eqnarray*}
\partl{\rho^{K}}\auxii & = & \bar{f}(\auxi)\sqrt{1-\auxi}\left(1+\sqrt{1-\auxi}\right)\frac{\lambda_{K}}{(\lambda\bar{\lambda})}+f(\auxi)\frac{\auxii\sqrt{1-\auxi}\left(1-\sqrt{1-\auxi}\right)}{1+\sqrt{1-\auxi}}\frac{\bar{\lambda}_{K}}{(\lambda\bar{\lambda})}\\
\partl{\bar{\rho}_{K}}\auxii & = & -\auxii\sqrt{1-\auxi}\bar{f}(\auxi)\frac{\lambda^{K}}{(\lambda\bar{\lambda})}-\auxii\frac{f(\auxi)\sqrt{1-\auxi}}{\left(1+\sqrt{1-\auxi}\right)}\frac{\auxii\bar{\lambda}^{K}}{(\lambda\bar{\lambda})}\\
\partl{\rho^{K}}\auxi & = & 2\bar{f}(\auxi)\left(1-\auxi\right)\frac{\bar{\auxii}\lambda_{K}}{(\lambda\bar{\lambda})}-2f(\auxi)\frac{\auxi\left(1-\auxi\right)}{1+\sqrt{1-\auxi}}\frac{\bar{\lambda}_{K}}{(\lambda\bar{\lambda})}\qquad\fussend
\end{eqnarray*}
}
\begin{eqnarray}
\partl{\rho^{K}}\auxii & = & 2\frac{\rho_{K}}{(\rho\bar{\rho})}-\auxii\frac{\bar{\rho}_{K}}{(\rho\bar{\rho})}\label{toy:zbyrho}\\
\partl{\bar{\rho}_{K}}\auxii & = & -\auxii\frac{\rho^{K}}{(\rho\bar{\rho})}\label{toy:zbyrhobar}\\
\partl{\rho^{K}}\auxi & = & 2\frac{\bar{\auxii}\rho_{K}}{(\rho\bar{\rho})}-2\auxi\frac{\bar{\rho}_{K}}{(\rho\bar{\rho})}\label{toy:xbyrho}
\end{eqnarray}
Using these, we obtain
\begin{eqnarray}
\lqn{\Pi_{(f)\bot K}^{I}(\rho,\bar{\rho})=}\nonumber \\
 & = & f\left(\auxi\right)\delta_{K}^{I}+2\bar{\auxii}f'(\auxi)\tfrac{\rho^{I}\rho_{K}}{(\rho\bar{\rho})}-2\auxi f'(\auxi)\tfrac{\rho^{I}\bar{\rho}_{K}}{(\rho\bar{\rho})}-\left(\frac{f(\auxi)}{\sqrt{1-\auxi}}+\frac{2\auxi f'(\auxi)}{1+\sqrt{1-\auxi}}\right)\tfrac{\bar{\rho}^{I}\rho_{K}}{(\rho\bar{\rho})}+\nonumber \\
 &  & +2\auxii\left(\frac{f(\auxi)}{2\sqrt{1-\auxi}\left(1+\sqrt{1-\auxi}\right)}+\left(1-\sqrt{1-\auxi}\right)f'(\auxi)\right)\tfrac{\bar{\rho}^{I}\bar{\rho}_{K}}{(\rho\bar{\rho})}\qquad\label{toy:Pi}
\end{eqnarray}
with trace
\begin{eqnarray}
\tr\Pi_{(f)\bot}(\rho,\bar{\rho}) & = & f\left(\auxi\right)\left(N-1\right)-2(1-\auxi)\left(1-\sqrt{1-\auxi}\right)f'(\auxi)\label{toy:trace}
\end{eqnarray}
and 
\begin{eqnarray}
\lqn{\rest_{(f)\bot}^{IK}(\rho,\bar{\rho})=}\nonumber \\
 & = & -\frac{f(\auxi)}{\left(1+\sqrt{1-\auxi}\right)}\auxii\delta^{IK}-2\auxi f'(\auxi)\tfrac{\rho^{I}\rho^{K}}{(\rho\bar{\rho})}+2\auxii f'(\auxi)\tfrac{\rho^{I}\bar{\rho}^{K}}{(\rho\bar{\rho})}+\nonumber \\
 &  & +2\auxii\left(\frac{f(\auxi)}{2\sqrt{1-\auxi}\left(1+\sqrt{1-\auxi}\right)}+\left(1-\sqrt{1-\auxi}\right)f'(\auxi)\right)\tfrac{\bar{\rho}^{I}\rho^{K}}{(\rho\bar{\rho})}+\nonumber \\
 &  & -\frac{2\auxii^{2}}{1+\sqrt{1-\auxi}}\left(\frac{f(\auxi)}{2\sqrt{1-\auxi}\left(1+\sqrt{1-\auxi}\right)}+f'(\auxi)\right)\tfrac{\bar{\rho}^{I}\bar{\rho}^{K}}{(\rho\bar{\rho})}\qquad\label{toy:pi}
\end{eqnarray}
Neither the matrix $\Pi_{(f)\bot}(\rho,\bar{\rho})$ nor the more
appropriate full matrix 
\begin{equation}
\left(\begin{array}{cc}
\Pi_{(f)\bot}(\rho,\bar{\rho}) & \rest_{(f)\bot}(\rho,\bar{\rho})\\
\bar{\rest}_{(f)\bot}(\rho,\bar{\rho}) & \bar{\Pi}_{(f)\bot}(\rho,\bar{\rho})
\end{array}\right)\label{toy:full-matrix}
\end{equation}
are in general projection matrices by themselves in the sense that
they are idempotent for all $\rho$. So in general $\Pi_{(f)\bot}^{2}(\rho,\bar{\rho})\neq\Pi_{(f)\bot}(\rho,\bar{\rho})$
and the square of the matrix (\ref{toy:full-matrix}) 
\begin{equation}
\left(\begin{array}{cc}
\Pi_{(f)\bot}^{2}+\rest_{(f)\bot}\bar{\rest}_{(f)\bot} & \Pi_{(f)\bot}\rest_{(f)\bot}+\rest_{(f)\bot}\bar{\Pi}_{(f)\bot}\\
\bar{\rest}_{(f)\bot}\Pi_{(f)\bot}+\bar{\Pi}_{(f)\bot}\bar{\rest}_{(f)\bot} & \bar{\Pi}_{(f)\bot}^{2}+\bar{\rest}_{(f)\bot}\rest_{(f)\bot}
\end{array}\right)\label{toy:full-matrix-squared}
\end{equation}
(for notational simplicity, we have suppressed here the argument $(\rho,\bar{\rho})$)
is in general not equal to (\ref{toy:full-matrix}). Matrices obeying
this strict kind of projection property would have an integer trace,
namely the dimension of the projected subspace. Looking at the general
trace in (\ref{toy:trace}), $f=1$ is an obvious choice where at
least the trace becomes integer
\begin{eqnarray}
\tr\Pi_{(1)\bot}(\rho,\bar{\rho}) & = & N-1
\end{eqnarray}
and it turns out that indeed $\Pi_{(1)\bot}(\rho,\bar{\rho})$ is
a proper projection matrix 
\begin{equation}
\Pi_{(1)\bot}^{2}(\rho,\bar{\rho})=\Pi_{(1)\bot}(\rho,\bar{\rho})\label{Pi1quad}
\end{equation}
It is therefore worth to spell out this special case, in which the
matrices also become particularly simple:
\begin{eqnarray}
\Pi_{(1)\bot K}^{I}(\rho,\bar{\rho}) & = & \delta_{K}^{I}-\frac{1}{\sqrt{1-\auxi}}\tfrac{\bar{\rho}^{I}\rho_{K}}{(\rho\bar{\rho})}+\frac{\auxii}{\sqrt{1-\auxi}\left(1+\sqrt{1-\auxi}\right)}\tfrac{\bar{\rho}^{I}\bar{\rho}_{K}}{(\rho\bar{\rho})}\qquad\\
\rest_{(1)\bot}^{IK}(\rho,\bar{\rho}) & = & -\frac{\auxii}{1+\sqrt{1-\auxi}}\Pi_{(1)\bot}^{IK}(\rho,\bar{\rho})
\end{eqnarray}
$\Pi_{(1)\bot K}^{I}(\rho,\bar{\rho})$ can be even further simplified
by noticing that some of the terms combine to the projection $\lambda_{K}\equiv P_{(1)K}(\rho,\bar{\rho})$:
\begin{equation}
\Pi_{(1)\bot K}^{I}(\rho,\bar{\rho})=\delta_{K}^{I}-\frac{1}{\sqrt{1-\auxi}}\tfrac{\bar{\rho}^{I}\lambda_{K}}{(\rho\bar{\rho})}
\end{equation}
Note that $f=1$ for the pure spinor case does not lead to $\Pi_{(1)\bot}^{2}(\rho,\bar{\rho})=\Pi_{(1)\bot}(\rho,\bar{\rho})$
nor to an integer trace as one can see in (\ref{LinProjTrace}) on
page \pageref{LinProjTrace} (See also footnote \ref{fn:differentialEqForf}
on page \pageref{fn:differentialEqForf}). If $f\neq1$, the matrices
can only be seen as part of a tangent bundle projection, namely 
\begin{eqnarray}
\mc P_{(f)}:\quad(\rho^{I},\bar{\rho}_{I}) & \mapsto & (\lambda^{I},\bar{\lambda}_{I})\equiv(P_{(f)}^{I}(\rho,\bar{\rho}),\bar{P}_{(f)I}(\rho,\bar{\rho}))\\
\left(\!\!\zwek{\delta\rho^{I}}{\delta\bar{\rho}_{I}}\!\!\right) & \mapsto & \!\!\left(\!\!\zwek{\delta\lambda^{I}}{\delta\bar{\lambda}_{I}}\!\!\right)=\!\left(\!\!\begin{array}{cc}
\Pi_{(f)\bot}(\rho,\bar{\rho}) & \rest_{(f)\bot}(\rho,\bar{\rho})\\
\bar{\rest}_{(f)\bot}(\rho,\bar{\rho}) & \bar{\Pi}_{(f)\bot}(\rho,\bar{\rho})
\end{array}\!\!\right)\!\!\left(\!\!\zwek{\delta\rho^{I}}{\delta\bar{\rho}_{I}}\!\!\right)\qquad
\end{eqnarray}
Its projection property 
\begin{equation}
\mc P_{(f)}\circ\mc P_{(f)}=\mc P_{(f)}
\end{equation}
implies for the matrices
\begin{eqnarray}
\Pi_{(f)\bot K}^{I}(\lambda,\bar{\lambda})\Pi_{(f)\bot J}^{K}(\rho,\bar{\rho})+\underbrace{\rest_{(f)\bot}^{IK}(\lambda,\bar{\lambda})}_{=0}\bar{\rest}_{(f)\bot KJ}(\rho,\bar{\rho})\!\! & = & \!\!\Pi_{(f)\bot J}^{I}(\rho,\bar{\rho})\qquad\label{projection-matrix-mod1}\\
\underbrace{\rest_{(f)\bot}^{IK}(\lambda,\bar{\lambda})}_{=0}\bar{\Pi}_{K}^{\bot J}(\rho,\bar{\rho})+\Pi_{\bot K}^{I}(\lambda,\bar{\lambda})\rest_{(f)\bot}^{KJ}(\rho,\bar{\rho})\!\! & = & \!\!\rest_{(f)\bot}^{IJ}(\rho,\bar{\rho})\label{projection-matrix-mod2}
\end{eqnarray}
Comparing the left side of these two equations with (\ref{toy:full-matrix-squared}),
the difference is only that here the first factor of each term is
evaluated at $\lambda^{I}\equiv P_{(f)}^{I}(\rho,\bar{\rho})$ and
only the second at $\rho^{I}$ while in (\ref{toy:full-matrix-squared})
both have argument $\rho$ as mentioned in the line below (\ref{toy:full-matrix-squared}). 

The Jacobian matrices are mapping to a subspace of vectors that is
perpendicular to $\lambda^{I}\equiv P_{(f)}^{I}(\rho,\bar{\rho})$:
\begin{eqnarray}
\lambda_{I}\Pi_{(f)\bot K}^{I}(\rho,\bar{\rho}) & = & 0\label{PibotOrthtoP}\\
\lambda_{I}\rest_{(f)\bot}^{IK}(\rho,\bar{\rho}) & = & 0\label{PirestOrthtoP}
\end{eqnarray}

On the constraint surface where $\rho=\lambda$ with $\lambda^{2}=0$
(and thus $\auxii=\auxi=0$, $f(0)=1$), the matrix $\Pi_{(f)\bot K}^{I}(\rho,\bar{\rho})$
reduces for all $f$ to the projection matrix we started with (as
it was asked for in (\ref{toy:to-solve}), just now with $\bar{\chi}=\bar{\rho}$)
while $\rest_{(f)\bot}^{IK}(\rho,\bar{\rho})$ vanishes there: 
\begin{eqnarray}
\Pi_{\bot}^{I}\tief K\equiv\Pi_{(f)\bot}^{I}\tief K(\lambda,\bar{\lambda}) & = & \delta_{K}^{I}-\tfrac{\bar{\lambda}^{I}\lambda_{K}}{(\lambda\bar{\lambda})}\\
\rest_{(f)\bot}^{IK}(\lambda,\bar{\lambda}) & = & 0
\end{eqnarray}
The matrix $\Pi_{\bot}\equiv\Pi_{(f)\bot}(\lambda,\bar{\lambda})$
depends on the function $f$ only implicitly via the dependence of
$\lambda$ on $f$. Note that $\Pi_{(f)\bot}$ evaluated on the constraint
surface, i.e. $\Pi_{\bot}\equiv\Pi_{(f)\bot}(\lambda,\bar{\lambda})$,
is a proper projection matrix for any $f$
\begin{eqnarray}
\Pi_{\bot}^{2} & = & \Pi_{\bot}\\
\tr\Pi_{\bot} & = & N-1
\end{eqnarray}
Furthermore this matrix is Hermitian
\begin{equation}
\Pi_{\bot}^{\dagger}=\Pi_{\bot}
\end{equation}
It is thus a natural question whether this Hermiticity can be extended
from $\Pi_{\bot}=\Pi_{(f)\bot}(\lambda,\bar{\lambda})$ off the constraint
surface to $\Pi_{(f)\bot}(\rho,\bar{\rho})$.

\paragraph{Hermitian Jacobian matrices }

The answer to the above question is that $\Pi_{(f)\bot}(\rho,\bar{\rho})$
is Hermitian for all $\rho$ if $f=h$, where 
\begin{eqnarray}
h(\auxi) & \equiv & \frac{1+\sqrt{1-\auxi}}{2\sqrt{1-\auxi}},\quad(h(0)=1)
\end{eqnarray}
In fact for this choice the matrices become
\begin{eqnarray}
\Pi_{(h)\bot K}^{I}(\rho,\bar{\rho})\!\!\!\!\! & = & \!\!\!\!\!\frac{1}{2\sqrt{1-\auxi}}\left(1+\sqrt{1-\auxi}\right)\delta_{K}^{I}-\frac{2-\auxi}{2\sqrt{1-\auxi}^{3}}\tfrac{\bar{\rho}^{I}\rho_{K}}{(\rho\bar{\rho})}+\nonumber \\
 &  & +\frac{\bar{\auxii}}{2\sqrt{1-\auxi}^{3}}\tfrac{\rho^{I}\rho_{K}}{(\rho\bar{\rho})}-\frac{\auxi}{2\sqrt{1-\auxi}^{3}}\tfrac{\rho^{I}\bar{\rho}_{K}}{(\rho\bar{\rho})}+\frac{\auxii}{2\sqrt{1-\auxi}^{3}}\tfrac{\bar{\rho}^{I}\bar{\rho}_{K}}{(\rho\bar{\rho})}\qquad\qquad\\
\rest_{(h)\bot}^{IK}(\rho,\bar{\rho}) & = & -\frac{\auxi}{2\sqrt{1-\auxi}^{3}}\tfrac{\rho^{I}\rho^{K}}{(\rho\bar{\rho})}+\frac{\auxii}{2\sqrt{1-\auxi}^{3}}\tfrac{\rho^{I}\bar{\rho}^{K}}{(\rho\bar{\rho})}+\frac{\auxii}{2\sqrt{1-\auxi}^{3}}\tfrac{\bar{\rho}^{I}\rho^{K}}{(\rho\bar{\rho})}+\nonumber \\
 &  & -\frac{\auxii^{2}}{2\sqrt{1-\auxi}^{3}}\tfrac{\bar{\rho}^{I}\bar{\rho}^{K}}{(\rho\bar{\rho})}-\frac{\auxii}{2\sqrt{1-\auxi}}\delta^{IK}
\end{eqnarray}
This shows that not only $\Pi_{(h)\bot}(\rho,\bar{\rho})$ is Hermitian,
but in addition $\rest_{(h)\bot}$ is symmetric 
\begin{eqnarray}
\Pi_{(h)\bot}^{\dagger}(\rho,\bar{\rho}) & = & \Pi_{(h)\bot}(\rho,\bar{\rho})\\
\rest_{(h)\bot}^{T}(\rho,\bar{\rho}) & = & \rest_{(h)\bot}(\rho,\bar{\rho})
\end{eqnarray}
These equations are in turn equivalent to the statement that the complete
block matrix {\tiny$\left(\begin{array}{cc}
\Pi_{(h)\bot}(\rho,\bar{\rho}) & \rest_{(h)\bot}(\rho,\bar{\rho})\\
\bar{\rest}_{(h)\bot}(\rho,\bar{\rho}) & \bar{\Pi}_{(h)\bot}(\rho,\bar{\rho})
\end{array}\right)$} is Hermitian.

The non-linear projection \eqref{toy:projGeneral} itself becomes
for $f=h$
\begin{equation}
\lambda^{I}\equiv P_{(h)}^{I}(\rho,\bar{\rho})\equiv\frac{1+\sqrt{1-\auxi}}{2\sqrt{1-\auxi}}\rho^{I}-\frac{\auxii}{2\sqrt{1-\auxi}}\bar{\rho}^{I}
\end{equation}
\rem{Also in the toy model it is hard to decide whether a limit $\auxi\to1$
exists}\rembreak with 'inverse' transformation (\ref{toy:inverse})

\begin{equation}
\rho^{I}=\lambda^{I}+\frac{\auxii}{1+\sqrt{1-\auxi}}\bar{\lambda}^{I}\label{toy:herm-inverse}
\end{equation}
In the alternative parametrization (\ref{toy:Auxii}) this becomes
according to (\ref{toy:inverseAltH}) simply $\rho^{I}=\lambda^{I}+\Auxii\bar{\lambda}^{I}$.

In particular in the Hermitian case it turns out to be convenient
for calculations to use the above inverse transformation, in order
to express the Jacobian matrices in terms of $\lambda^{I}$, $\auxii$
and their complex conjugates: 
\begin{eqnarray}
\Pi_{(h)\bot K}^{I}(\rho,\bar{\rho}) & = & \frac{1+\sqrt{1-\auxi}}{2\sqrt{1-\auxi}}\biggl(\delta_{K}^{I}-\frac{\bar{\lambda}^{I}\lambda_{K}}{(\lambda\bar{\lambda})}-\frac{1-\sqrt{1-\auxi}}{1+\sqrt{1-\auxi}}\frac{\lambda^{I}\bar{\lambda}_{K}}{(\lambda\bar{\lambda})}\biggr)\qquad\\
 & \stackrel{\auxi=0}{=} & \delta_{K}^{I}-\frac{\bar{\lambda}^{I}\lambda_{K}}{(\lambda\bar{\lambda})}\quad(\equiv\Pi_{\bot K}^{I})\\
\rest_{(h)\bot}^{IK}(\rho,\bar{\rho}) & = & -\frac{\auxii}{2\sqrt{1-\auxi}}\Bigl(\underbrace{\delta^{IK}-\tfrac{\bar{\lambda}^{I}\lambda^{K}}{(\lambda\bar{\lambda})}}_{\Pi_{\bot}}-\underbrace{\tfrac{\lambda^{I}\bar{\lambda}^{K}}{(\lambda\bar{\lambda})}}_{\equiv\Pi_{\Vert}^{T}}\Bigr)\\
 & \stackrel{\auxii=\auxi=0}{=} & 0
\end{eqnarray}
The $\rho$-dependence of these matrices is now only implicitly, with
$\auxii,\auxi$ and $\lambda^{I}\equiv P_{(h)}^{I}(\rho,\bar{\rho})$
being functions of $\rho$. The trace of the first matrix is given
by 
\begin{equation}
\tr\Pi_{(h)\bot}(\rho,\bar{\rho})=\frac{1+\sqrt{1-\auxi}}{2\sqrt{1-\auxi}}N-\frac{1}{\sqrt{1-\auxi}}
\end{equation}
It is useful to note that in the Hermitian case $\bar{\rest}_{(h)\bot}$
is proportional to $\rest_{(h)\bot}$ 
\begin{equation}
\bar{\rest}_{(h)\bot IK}(\rho,\bar{\rho})=\frac{\bar{\auxii}}{\auxii}\rest_{(h)\bot IK}(\rho,\bar{\rho})
\end{equation}
If we denote just for a moment $\partial_{I}\equiv\partial_{\rho^{I}}$
and $\bar{\partial}^{I}\equiv\partial_{\bar{\rho}_{I}}$, then we
have due to Hermiticity\vspace{-.4cm} 
\begin{equation}
\partial_{[I|}\bar{P}_{(h)|K]}=\bar{\partial}^{[I}P_{(h)}^{K]}=\partial_{I}P_{(h)}^{K}-\bar{\partial}^{K}\bar{P}_{(h)I}=0
\end{equation}
This in turn implies that $P^{I}\de\bar{\rho}_{I}+\bar{P}_{I}\de\rho^{I}$
is a closed 1-form and thus locally exact. I.e. there exists a scalar
function $\Phi(\rho,\bar{\rho})$ such that 
\begin{equation}
\lambda^{I}\equiv P^{I}(\rho,\bar{\rho})=\bar{\partial}^{I}\Phi(\rho,\bar{\rho})
\end{equation}
Indeed an explicit solution for the {}``potential'' $\Phi$ is given
by 
\begin{equation}
\Phi(\rho,\bar{\rho})=\frac{(\rho\bar{\rho})}{2}\left(1+\sqrt{1-\auxi}\right)\label{toy:potential}
\end{equation}
One can further check by direct calculation that 
\begin{equation}
(\lambda\bar{\lambda})\equiv P^{K}(\rho,\bar{\rho})\bar{P}_{K}(\rho,\bar{\rho})=\Phi(\rho,\bar{\rho})\label{toy:PhiIslambdalambdabar}
\end{equation}

\subsubsection*{Holomorphic volume form and measure}

\paragraph{Ambient space }

The canonical holomorphic volume form in the unconstrained ambient
space $\mathbb{C}^{\toydim}$ is given by 

\begin{eqnarray}
[d^{\mc N}\rho] & = & \de\rho^{1}\wedge\ldots\wedge\de\rho^{\toydim}=\tfrac{1}{\toydim!}\epsilon_{K_{1}\ldots K_{\toydim}}\de\rho^{K_{1}}\wedge\ldots\wedge\de\rho^{K_{\toydim}}\label{toy:targspaceVolform}
\end{eqnarray}
Wedged with its complex conjugate and using $\epsilon_{K_{1}\ldots K_{\toydim}}\epsilon^{L_{1}\ldots L_{\toydim}}=\toydim!\delta_{K_{1}\ldots K_{\toydim}}^{L_{1}\ldots L_{\toydim}}$
and for the signs ${\scriptstyle \sum_{k=0}^{\toydim-1}k=\frac{\toydim(\toydim-1)}{2}}$,
one obtains the canonical integration measure of the ambient space
\begin{equation}
[d^{\mc N}\rho]\wedge[d^{\mc N}\bar{\rho}]=(-)^{\frac{\toydim(\toydim-1)}{2}}\tfrac{1}{\toydim!}(\de\rho^{I}\de\bar{\rho}_{I})^{\toydim}\label{toy:targspaceMeasure}
\end{equation}
where on the righthand side we have omitted the wedge symbol $\wedge$
and will omit it also in the following.

\paragraph{$\lambda$-space}

In the $\lambda$-space (which we think of being embedded in the ambient
$\rho$-space at $\auxii=0$) things complicate a bit, because of
the constraint $\lambda^{2}=0$, in particular when trying to build
a covariant holomorphic volume form. So let us first collect a few
identities which we will need in the following. They are all related
to the fact that $\lambda^{K}$ and $\de\lambda^{K}$ effectively
have only $\toydim-1$ independent components and thus the antisymmetrization
of $\toydim$ vectors vanishes.
\begin{eqnarray}
\de\lambda^{I_{1}}\cdots\de\lambda^{I_{\toydim}} & = & 0\label{toy:identity1}\\
\de\lambda^{[I_{1}}\cdots\de\lambda^{I_{\toydim-1}}\lambda^{I_{\toydim}]} & = & 0\label{toy:identity2}\\
\toydim\de\lambda^{[1}\cdots\de\lambda^{\toydim-2}\de\lambda^{+}v^{-]} & = & \tfrac{1}{\lambda^{+}}\de\lambda^{1}\cdots\de\lambda^{\toydim-2}\de\lambda^{+}v_{I}\lambda^{I}\label{toy:identity2-gen}
\end{eqnarray}
We will prove these identities in a footnote\footnotemark  on the
next page. This will require at least for the last identity (\ref{toy:identity2-gen})
the explicit solution (\ref{toy:lamsol}) of the constraint together
with its derivatives, which obey 
\begin{equation}
\partial_{+}\lambda_{{\rm sol}}^{-}=-\frac{\lambda_{{\rm sol}}^{-}}{\lambda^{+}}\quad,\quad\partial_{i}\lambda_{{\rm sol}}^{-}=-\frac{\lambda_{i}}{\lambda^{+}}\label{toy:anotherusefulrelation}
\end{equation}
The holomorphic volume form that we are looking for has to contain
$\toydim-1$ powers of $\de\lambda^{I}$. Their $\toydim-1$ indices
have to be contracted covariantly with $\epsilon_{I_{1}\ldots I_{\toydim}}$,
$\delta_{IJ}$ or the only available vectors, namely $\lambda^{I}$
and $\bar{\lambda}_{I}$. It is clear that we cannot do without the
$\epsilon$-tensor. All contractions avoiding it would vanish. But
after contracting the $\toydim-1$ one-forms with the covariant $\epsilon$-tensor
of the ambient space $\epsilon_{I_{1}\ldots I_{\toydim-1}I_{\toydim}}\de\lambda^{I_{1}}\cdots\de\lambda^{I_{\toydim-1}}$,
there is one index ${\scriptstyle I_{\toydim}}$ remaining, which
needs to be saturated. So we have to contract it with another vector,
so either with $\lambda^{I}$ or with $\bar{\lambda}^{I}$. Contracting
with $\lambda^{I}$ yields a vanishing result due to (\ref{toy:identity2}).
Contracting it instead with $\bar{\lambda}^{I}$ seems to be nonsense
a priori, as we want a holomorphic form. However, from (\ref{toy:identity2-gen})
it is clear that this $\bar{\lambda}^{I}$ will actually be contracted
with a $\lambda^{I}$ so that this non-holomorphic factor can be divided
out. So the following covariant expression is indeed holomorphic:
\begin{eqnarray}
[d\lambda^{\toydim-1}] & \equiv & \frac{1}{(\toydim-1)!(\lambda\bar{\lambda})}\epsilon_{I_{1}\ldots I_{\toydim}}\de\lambda^{I_{1}}\cdots\de\lambda^{I_{\toydim-1}}\bar{\lambda}^{I_{\mc N}}\label{toy:holomVolform}
\end{eqnarray}
\footnotetext{\label{fn:toy:identity}Proof of (\ref{toy:identity1})-(\ref{toy:identity2-gen}):
Remember that the constraint on $\lambda^{I}$ is quadratic, i.e.
$\lambda^{I}\lambda_{I}=2\lambda^{+}\lambda^{-}+\lambda^{i}\lambda_{i}=0$,
while the one on $\de\lambda^{I}$ is linear in the $\de\lambda$'s,
i.e. $\lambda_{I}\de\lambda^{I}=0$. Solving the latter equation for
one of the $\de\lambda$'s, e.g. for $\de\lambda^{\toydim}$ yields
an expression linear in the other $\de\lambda$'s, so that \eqref{toy:identity1}
is rather obvious. Instead, in order to show \eqref{toy:identity2}
or (\ref{toy:identity2-gen}), it is essential to observe that the
constraint-function $\lambda^{I}\lambda_{I}$ is homogeneous of degree
2 (or in other words has definite {}``ghost number'' 2). Solving
for any of the $\lambda$'s, e.g. for $\lambda^{-}$, yields some
function of the other variables 
\[
\lambda^{-}=\lambda_{{\rm sol}}^{-}(\lambda^{1},\ldots,\lambda^{\toydim-2},\lambda^{+})\quad(*)
\]
At least for showing \eqref{toy:identity2}, it is actually not essential
here to switch to {}``lightcone coordinates'' with $\lambda^{+},\lambda^{-}$.
One could also solve for $\lambda_{{\rm sol}}^{\toydim}(\lambda^{1},\ldots,\lambda^{\toydim-1})$,
although the explicit form of the solution $\lambda_{{\rm sol}}^{\toydim}=\pm\sqrt{-(\lambda^{1})^{2}-\ldots-(\lambda^{\toydim-1})^{2}}$
is not unique and is also not so nice as the one in lightcone coordinates
$\lambda_{{\rm sol}}^{-}=-\tfrac{1}{2\lambda^{+}}\lambda^{i}\lambda_{i}$.
However, what matters for the argument is only homogeneity: Because
of the homogeneity of the constraint function also the functions $\lambda_{{\rm sol}}^{\toydim}$
or $\lambda_{{\rm sol}}^{-}$ are homogeneous, but now of degree 1
(They have ghost number 1), i.e. 
\[
\left(\lambda^{i}\partial_{i}+\lambda^{+}\partial_{+}\right)\lambda_{{\rm sol}}^{-}=\lambda_{{\rm sol}}^{-}\quad(\#)
\]
In order to show \eqref{toy:identity2} and (\ref{toy:identity2-gen}),
first take the differential of ({*}): 
\[
\de\lambda_{{\rm sol}}^{-}=\left(\de\lambda^{i}\partial_{i}+\de\lambda^{+}\partial_{+}\right)\lambda_{{\rm sol}}^{-}\quad(**)
\]
Now we can start with the left-hand side of (\ref{toy:identity2-gen})
(for a particular permutation of the indices), make the antisymmetrization
of the index `$-$' explicit and then replace $\de\lambda^{-}$ using
({*}{*}). 
\begin{eqnarray*}
\lqn{\toydim\underbrace{\de\lambda^{[1}\cdots\de\lambda^{\toydim-2}\de\lambda^{+}}_{\toydim-1}v^{-]}=}\\
 & = & \de\lambda^{1}\cdots\de\lambda^{\toydim-2}\de\lambda^{+}v^{-}-\de\lambda^{1}\cdots\de\lambda^{\toydim-2}\de\lambda^{-}v^{+}+(\toydim-2)\de\lambda^{+}\de\lambda^{-}\de\lambda^{[1}\cdots\de\lambda^{\toydim-3}v^{\toydim-2]}=\\
 & \stackrel{(**)}{=} & \de\lambda^{1}\cdots\de\lambda^{\toydim-2}\de\lambda^{+}v^{-}-\de\lambda^{1}\cdots\de\lambda^{\toydim-2}\left(\de\lambda^{i}\partial_{i}\lambda_{{\rm sol}}^{-}+\de\lambda^{+}\partial_{+}\lambda_{{\rm sol}}^{-}\right)v^{+}+\\
 &  & +(\toydim-2)\de\lambda^{+}\left(\de\lambda^{i}\partial_{i}\lambda_{{\rm sol}}^{-}+\de\lambda^{+}\partial_{+}\lambda_{{\rm sol}}^{-}\right)\de\lambda^{[1}\cdots\de\lambda^{\toydim-3}v^{\toydim-2]}=\\
 & = & \de\lambda^{1}\cdots\de\lambda^{\toydim-2}\de\lambda^{+}v^{-}-\de\lambda^{1}\cdots\de\lambda^{\toydim-2}\de\lambda^{+}\partial_{+}\lambda_{{\rm sol}}^{-}v^{+}+\\
 &  & +(-)^{\mc N-3}\de\lambda^{+}\de\lambda^{1}\cdots\de\lambda^{\mc N-2}v^{i}\partial_{i}\lambda_{{\rm sol}}^{-}=\\
 & = & \de\lambda^{1}\cdots\de\lambda^{\toydim-2}\de\lambda^{+}\bigl(v^{-}-v^{+}\partial_{+}\lambda_{{\rm sol}}^{-}-v^{i}\partial_{i}\lambda_{{\rm sol}}^{-}\bigr)
\end{eqnarray*}
If at this point we replace the general vector $v^{I}$ by $\lambda^{I}$,
then because of (\#) the last line indeed vanishes, which proves \eqref{toy:identity2}
without making use of the explicit form of the constraint. In order
to prove (\ref{toy:identity2-gen}), however, we need the explicit
form (\ref{toy:lamsol}) or in particular its derivatives (\ref{toy:anotherusefulrelation}).
Using these derivatives for the above equation with general $v^{I}$
yields: 
\begin{eqnarray*}
\toydim\de\lambda^{[1}\cdots\de\lambda^{\toydim-2}\de\lambda^{+}v^{-]} & = & \tfrac{1}{\lambda^{+}}\de\lambda^{1}\cdots\de\lambda^{\toydim-2}\de\lambda^{+}\left(v^{-}\lambda^{+}+v^{+}\lambda^{-}+v^{i}\lambda_{i}\right)=\\
 & = & \tfrac{1}{\lambda^{+}}\de\lambda^{1}\cdots\de\lambda^{\toydim-2}\de\lambda^{+}v_{I}\lambda^{I}\quad\fussend
\end{eqnarray*}
}The corresponding measure of the total space is therefore
\begin{eqnarray}
\lqn{[d\lambda^{\toydim-1}]\wedge[d\bar{\lambda}^{\toydim-1}]=}\\
 & = & \tfrac{\bar{\lambda}^{K}\lambda_{L}}{(\toydim-1)!^{2}(\lambda\os{}{\bar{\lambda}})^{2}}\underbrace{\epsilon_{I_{1}\ldots I_{\toydim-1}K}\epsilon^{J_{1}\ldots J_{\toydim-1}L}}_{\toydim!\delta_{I_{1}\ldots I_{\toydim-1}K}^{J_{1}\ldots J_{\toydim-1}L}}\de\lambda^{I_{1}}\cdots\de\lambda^{I_{\toydim-1}}\de\bar{\lambda}_{J_{1}}\cdots\de\bar{\lambda}_{J_{\toydim-1}}\nonumber 
\end{eqnarray}
Using ${\scriptstyle \toydim\delta_{I_{1}\ldots I_{\toydim-1}K}^{J_{1}\ldots J_{\toydim-1}L}=\delta_{I_{1}\ldots I_{\toydim-1}}^{J_{1}\ldots J_{\toydim-1}}\delta_{K}^{L}-(\toydim-1)\delta_{[I_{1}\ldots I_{\toydim-2}|K}^{J_{1}\ldots J_{\toydim-1}}\delta_{|I_{\toydim-1}]}^{L}}$
and $\lambda_{L}\de\lambda^{L}=0$ and for the signs ${\scriptstyle \sum_{k=0}^{\toydim-2}k=\frac{(\toydim-1)(\toydim-2)}{2}}$,
this becomes
\begin{equation}
[d\lambda^{\toydim-1}]\wedge[d\bar{\lambda}^{\toydim-1}]=\tfrac{1}{(\toydim-1)!(\lambda\os{}{\bar{\lambda}})}(-)^{\frac{(\toydim-1)(\toydim-2)}{2}}(\de\lambda^{I}\de\bar{\lambda}_{I})^{\toydim-1}\label{toy:measure}
\end{equation}
From the total space point of view, the prefactor $\frac{1}{(\lambda\bar{\lambda})}$
is not really necessary. It comes only from the requirement that we
want it to factorize into holomorphic and anti-holomorphic volume
form, so from compatibility with the complex structure.

\paragraph{From the ambient space to the $\lambda$-space}

Thinking of the constrained $\lambda$-space as being embedded at
$\auxii=0$ (or $\Auxii=0$ in the parametrization of (\ref{toy:Auxii}))
into the ambient space, it is a natural question if the above volume
form and measure can be derived from the ambient space. The idea is
to make a variable transformation from $\rho^{I}$ to $(\lambda^{I},\auxii)$
or $(\lambda^{I},\Auxii$) and see if the transformed volume form
factorizes. Let us consider the case $f=h$ where the parametrization
with the variables $\Auxii$ of (\ref{toy:Auxii}) becomes according
to \eqref{toy:inverseAltH} simply 
\begin{equation}
\rho^{I}=\lambda^{I}+\Auxii\bar{\lambda}^{I}\label{toy:rho}
\end{equation}
It is obvious that this transformation is not holomorphic which might
spoil the idea to obtain a holomorphic $\lambda$-volume form after
the transformation. However, on the constraint surface $\Auxii=0$,
the transformation is holomorphic. On the other hand, the transformation
of one-forms, given by
\begin{equation}
\de\rho^{I}=\de\lambda^{I}+\Auxii\de\bar{\lambda}^{I}+\de\Auxii\bar{\lambda}^{I}\label{toy:drho}
\end{equation}
is not holomorphic even at $\Auxii=0$. But this does not exclude
the possibility that the volume form might still transform holomorphically
at $\Auxii=0$ and this is what we are going to test in the following. 

But before we discuss the transformation of the holomorphic volume
form (\ref{toy:targspaceVolform}), let us first consider the transformation
of the total measure (\ref{toy:targspaceMeasure}) where the problem
of holomorphicity does not yet appear. For the measure it remains
to see, if we can factorize the result covariantly and whether one
factor coincides with (\ref{toy:measure}). So let us start by replacing
in (\ref{toy:targspaceMeasure}) all $\de\rho^{I}$ using (\ref{toy:drho}):
\begin{eqnarray}
\lqn{(\de\rho^{I}\de\bar{\rho}_{I})^{\toydim}=}\nonumber \\
 & = & \left(\left(\de\lambda^{I}+\Auxii\de\bar{\lambda}^{I}+\de\Auxii\bar{\lambda}^{I}\right)\left(\de\bar{\lambda}_{I}+\bar{\Auxii}\de\lambda_{I}+\de\bar{\Auxii}\lambda_{I}\right)\right)^{\toydim}\!\!\!=\qquad\\
 & \stackrel{\de\lambda^{I}\lambda_{I}=0}{=} & \Bigl\{\de\lambda^{I}\de\bar{\lambda}_{I}\left(1-\Auxi\right)+\Auxii\lambda_{I}\de\bar{\lambda}^{I}\de\bar{\Auxii}+\nonumber \\
 &  & -\bar{\Auxii}\bar{\lambda}^{I}\de\lambda_{I}\de\Auxii+(\lambda\bar{\lambda})\de\Auxii\de\bar{\Auxii}\Bigr\}^{\toydim}
\end{eqnarray}
 Now we can use the identity (\ref{toy:identity1}) and its complex
conjugate which imply that only terms survive which have $\toydim-1$
factors of $\de\lambda^{I}$ and of $\de\bar{\lambda}_{I}$ respectively,
as well as one factor of $\de\Auxii$ and $\de\bar{\Auxii}$. In addition
we are using the constraints $\lambda_{I}\lambda^{I}=\lambda_{I}\de\lambda^{I}=\de\lambda_{I}\de\lambda^{I}=0$
wherever possible:
\begin{eqnarray}
(\de\rho^{I}\de\bar{\rho}_{I})^{\toydim} & \stackrel{\mbox{\small\eqref{toy:identity1}}}{=} & \toydim\left(\de\lambda^{I}\de\bar{\lambda}_{I}(1-\Auxi)\right)^{\toydim-1}(\lambda\bar{\lambda})\de\Auxii\de\bar{\Auxii}+\label{toy:intermed}\\
 &  & +\toydim(\toydim-1)\left(\de\lambda^{I}\de\bar{\lambda}_{I}(1-\Auxi)\right)^{\toydim-2}\bar{\lambda}^{J}\de\lambda_{J}\lambda_{K}\de\bar{\lambda}^{K}\Auxi\de\Auxii\de\bar{\Auxii}\nonumber 
\end{eqnarray}
In order to combine the two terms which have different contractions
among the $\de\lambda$'s, we can use the identity \eqref{toy:identity2}
in the following way:
\begin{eqnarray}
0\!\! & \!\!\!\!\!\stackrel{\mbox{\small\eqref{toy:identity2}}}{=} & \!\!\!\!\toydim\de\lambda^{[I_{1}}\de\bar{\lambda}_{I_{1}}\cdots\de\lambda^{I_{\toydim-2}}\de\bar{\lambda}_{I_{\toydim-2}}\de\lambda^{J}\de\bar{\lambda}_{K}\lambda^{K]}\bar{\lambda}_{J}=\\
 & \!\!\!\!= & \!\!\!\!\!\!\de\lambda^{I_{1}}\de\bar{\lambda}_{I_{1}}\cdots\de\lambda^{I_{\toydim-2}}\de\bar{\lambda}_{I_{\toydim-2}}\de\lambda^{J}\de\bar{\lambda}_{K}\lambda^{K}\bar{\lambda}_{J}+\nonumber \\
 &  & \!\!\!\!\!\!\!\!-\de\lambda^{I_{1}}\de\bar{\lambda}_{I_{1}}\cdots\de\lambda^{I_{\toydim-2}}\de\bar{\lambda}_{I_{\toydim-2}}\de\lambda^{K}\de\bar{\lambda}_{K}\lambda^{J}\bar{\lambda}_{J}+\nonumber \\
 &  & \!\!\!\!\!\!\!\!+{\scriptstyle (\toydim-2)}\de\lambda^{I_{1}}\de\bar{\lambda}_{I_{1}}\!\cdots\!\de\lambda^{I_{\toydim-3}}\de\bar{\lambda}_{I_{\toydim-3}}\de\lambda^{J}\de\bar{\lambda}_{I_{\toydim-2}}\de\lambda^{K}\de\bar{\lambda}_{K}\lambda^{I_{\toydim-2}}\bar{\lambda}_{J}\!\!=\quad\qquad\\
 & \!\!\!\!= & \!\!\!\!\!\!{\scriptstyle (\toydim-1)}(\de\lambda^{I}\de\bar{\lambda}_{I})^{\toydim-2}\de\lambda^{J}\de\bar{\lambda}_{K}\lambda^{K}\bar{\lambda}_{J}-(\lambda\bar{\lambda})(\de\lambda^{I}\de\bar{\lambda}_{I})^{\toydim-1}
\end{eqnarray}
This implies the identity 
\begin{equation}
{\scriptstyle (\toydim-1)}(\de\lambda^{I}\de\bar{\lambda}_{I})^{\toydim-2}\de\lambda^{J}\de\bar{\lambda}_{K}\lambda^{K}\bar{\lambda}_{J}=(\lambda\bar{\lambda})(\de\lambda^{I}\de\bar{\lambda}_{I})^{\toydim-1}\label{toy:identity3}
\end{equation}
Plugging this back into \eqref{toy:intermed} yields
\begin{eqnarray}
(\de\rho^{I}\de\bar{\rho}_{I})^{\toydim} & = & \toydim\cdot(\lambda\bar{\lambda})(\de\lambda^{I}\de\bar{\lambda}_{I})^{\toydim-1}\times(1-\Auxi)^{\toydim-2}\de\Auxii\de\bar{\Auxii}
\end{eqnarray}
This, as indicated, factorizes into a $\lambda$- and a $\Auxii$-volume
form. However, the $\lambda$-volume form $\toydim\cdot(\lambda\bar{\lambda})(\de\lambda^{I}\de\bar{\lambda}_{I})^{\toydim-1}$
does not come with the expected negative power of $(\lambda\bar{\lambda})$
as in (\ref{toy:measure}).

This problem disappears if we redefine 
\begin{equation}
\Hauxii\equiv(\lambda\bar{\lambda})\Auxii\quad,\quad\Hauxi\equiv\Hauxii\bar{\Hauxii}=(\lambda\bar{\lambda})^{2}\Auxi\label{toy:Hauxi}
\end{equation}
Note, however, that then the factorization breaks down in general
(the $\Hauxii$-measure becomes $\lambda$-dependent) and we need
to restrict to the constraint surface $\Hauxi=0$, if we want to restore
it:
\begin{eqnarray}
(\de\rho^{I}\de\bar{\rho}_{I})^{\toydim}\!\!\! & = & \!\!\!\!\tfrac{\toydim}{(\lambda\os{}{\bar{\lambda}})}\cdot(\de\lambda^{I}\de\bar{\lambda}_{I})^{\toydim-1}\times(1-\tfrac{1}{(\lambda\os{}{\bar{\lambda}})^{2}}\Hauxi)^{\toydim-2}\de\Hauxii\de\bar{\Hauxii}=\qquad\\
 & \stackrel{\Hauxi=0}{=} & \!\!\!\!\tfrac{\toydim}{(\lambda\os{}{\bar{\lambda}})}\cdot(\de\lambda^{I}\de\bar{\lambda}_{I})^{\toydim-1}\times\de\Hauxii\de\bar{\Hauxii}
\end{eqnarray}
So using the variables (\ref{toy:Hauxi}) we have achieved to arrive
at the measure (\ref{toy:measure}) on the constraint surface.  Note
finally that the new variables $\Hauxii$ turn out to be very simple
in terms of $\rho$, if one uses that according to (\ref{toy:potential}),
(\ref{toy:PhiIslambdalambdabar}), (\ref{toy:Auxii}) and (\ref{toy:z-def})
it is nothing else but 
\begin{equation}
\Hauxii=\tfrac{1}{2}\rho^{2}
\end{equation}
This means in particular that in contrast to $\auxii$ or $\Auxii$
it depends holomorphically on $\rho^{\alpha}$.

\paragraph{Transformation of the target space holomorphic volume form}

Now let us see if we can achieve the same for the holomorphic volume
form. Using the redefined variables $\Hauxii$ of (\ref{toy:Hauxi}),
the reparametrization of $\rho^{I}$ and the corresponding transformation
of $\de\rho^{I}$ read
\begin{eqnarray}
\rho^{I} & = & \lambda^{I}+\Hauxii\frac{\bar{\lambda}^{I}}{(\lambda\bar{\lambda})}\label{toy:rho-ito-Hauxi}\\
\de\rho^{I} & = & \Bigl(\delta_{J}^{I}-\Hauxii\frac{\bar{\lambda}^{I}\bar{\lambda}_{J}}{(\lambda\bar{\lambda})^{2}}\Bigr)\de\lambda^{J}+\frac{\bar{\lambda}^{I}}{(\lambda\bar{\lambda})}\de\Hauxii+\frac{\Hauxii}{(\lambda\bar{\lambda})}\Bigl(\delta_{J}^{I}-\frac{\bar{\lambda}^{I}\lambda_{J}}{(\lambda\bar{\lambda})}\Bigr)\de\bar{\lambda}^{J}\qquad
\end{eqnarray}
For the transformation of the holomorphic volume form (\ref{toy:targspaceVolform})
of the ambient space, we will now immediately restrict to the constraint
surface $\Hauxii=0$ where $\de\rho^{I}=\de\lambda^{I}+\frac{\bar{\lambda}^{I}}{(\lambda\bar{\lambda})}\de\Hauxii$.
This yields
\begin{eqnarray}
\bei{[d^{\mc N}\rho]}{\Hauxii=0} & \!\!\!= & \!\!\!\!\!\!\tfrac{1}{\toydim!}\epsilon_{K_{1}\ldots K_{\toydim}}\Bigl(\de\lambda^{K_{1}}+\frac{\bar{\lambda}^{K_{1}}}{(\lambda\bar{\lambda})}\de\Hauxii\Bigr)\cdots\Bigl(\de\lambda^{K_{\toydim}}+\frac{\bar{\lambda}^{K_{\toydim}}}{(\lambda\bar{\lambda})}\de\Hauxii\Bigr)\!\!=\quad\qquad\\
 & \!\!\!\stackrel{\eqref{toy:identity1}}{=} & \!\!\!\!\!\!\tfrac{1}{(\lambda\bar{\lambda})(\toydim-1)!}\epsilon_{K_{1}\ldots K_{\toydim}}\de\lambda^{K_{1}}\cdots\de\lambda^{K_{\toydim-1}}\bar{\lambda}^{K_{\toydim}}\times\de\Hauxii
\end{eqnarray}
This indeed factorizes into the holomorphic $\lambda$-measure (\ref{toy:holomVolform})
and $\de\Hauxii$. So on the constraint surface the non-holomorphic
transformation (\ref{toy:rho-ito-Hauxi}) seems to transform the volume
form holomorphically.

\subsubsection*{From toy-model formulas to pure spinor formulas}

\label{sub:non-rigorous}The toy-model equations have been proven
to be a valuable guide for the pure spinor case, although they do
not make use of something like the Clifford algebra nor of Fierz identities.
Nevertheless many equations look almost identical. A recipe to go
from pure spinor equations back to toy-model equations is to make
the following replacements 
\begin{eqnarray*}
\rho^{\alpha},\bar{\rho}_{\alpha} & \to & \rho^{I},\bar{\rho}_{I}\\
\auxii^{a},\bar{\auxii}_{a} & \to & \auxii,\bar{\auxii}\\
\gamma_{\alpha\beta}^{a} & \to & \delta_{IJ}\\
\sum_{a}(\ldots) & \to & 2\cdot(\ldots)\\
\auxi & \to & \auxi
\end{eqnarray*}
There are a few deviations from this rule which is therefore not rigorous
but rather serves as a guideline. For example the last term of \eqref{HermLinProjLam}
does not appear in the toy-model. It would be interesting to see if
one can modify the recipe in a way which also correctly covers such
terms. \rem{These deviations from made me think that I had some mistakes
for a long time, but I have rechecked carefully:

As mentioned, the last term of \eqref{HermLinProjLam}, i.e. $-\frac{1}{16\sqrt{1-\auxi}\left(1+\sqrt{1-\auxi}\right)}\frac{\bar{\auxii}_{c}(\gamma^{c}\gamma_{b}\lambda)^{\alpha}\auxii^{d}(\bar{\lambda}\gamma^{b}\gamma_{d})_{\beta}}{(\lambda\bar{\lambda})}$,
does not appear in the toy-model. Following the above rule, it should
correspond to $-\frac{\auxi}{2\sqrt{1-\auxi}\left(1+\sqrt{1-\auxi}\right)}\frac{\lambda^{\alpha}\bar{\lambda}_{\beta}}{(\lambda\bar{\lambda})}=-\frac{1-\sqrt{1-\auxi}}{2\sqrt{1-\auxi}}\frac{\lambda^{\alpha}\bar{\lambda}_{\beta}}{(\lambda\bar{\lambda})}$.

It is confusing that in \eqref{HermLinProj} there is no such mismatch
to the toy-model, so it seems to appear only with the inverse transformation
(\eqref{hermInverse}). This by itself however also doesn't have a
mismatch. Instead, the expressions in footnote \eqref{fn:xzrho-derivsLamHerm}
do have again the same kind of mismatch. Would be nice to understand
this, in order to be really confident about the equations. 

The mismatches imply also that in the pure spinor case the choice
$f(\auxi)\equiv1\quad\forall\auxi$ does not have the property $\Pi_{(1)\bot}^{2}=\Pi_{(1)\bot}$
that it had in the toy-model. Related to that, its trace is not 11
(in the toy-model it was $N-1)$. It would be 11, however, precisely
when one drops the disturbing term $-\frac{1}{16\sqrt{1-\auxi}\left(1+\sqrt{1-\auxi}\right)}\frac{\bar{\auxii}_{c}(\gamma^{c}\gamma_{b}\lambda)^{\alpha}\auxii^{d}(\bar{\lambda}\gamma^{b}\gamma_{d})_{\beta}}{(\lambda\bar{\lambda})}$
which really made me think for a long time that it shouldn't be there.

A try to go even in the other direction was the idea to always group
$\auxii$ with a Kronecker $\delta^{IJ}$ and then use the mapping
\[
\auxii\delta^{IJ}\to\tfrac{1}{2}\auxii^{a}\gamma_{a}^{\alpha\beta}
\]
This leads however already for $\Pi_{\bot}=\delta_{K}^{I}-\frac{\bar{\lambda}^{I}\lambda_{K}}{(\lambda\bar{\lambda})}=\delta_{K}^{I}-\frac{\auxii\delta^{IJ}\bar{\auxii}\delta_{KL}\bar{\lambda}_{J}\lambda^{L}}{\auxi(\lambda\bar{\lambda})}$
to the wrong result $\delta_{\beta}^{\alpha}-\tfrac{1}{4\auxi}\frac{\auxii^{a}(\gamma_{a}\bar{\lambda})^{\alpha}\bar{\auxii}_{a}(\gamma^{a}\lambda)_{\beta}}{(\lambda\bar{\lambda})}$.
Interestingly, its trace would be $15=(16-1)$ instead of $11$, which
makes it a too close analogy to the toy-model. }\newpage

\section{Proofs}

\subsection{Proof of proposition 1}

\label{app:proof1}\begin{proof} Let us start from the end with the
properties of the auxiliary variables:

5. The equations in (\ref{zsquarezero}), i.e. $\auxii^{a}\auxii_{a}=0$
and $\auxii^{a}(\gamma_{a}\rho)_{\alpha}=0$, follow directly from
the Fierz identity.%
\footnote{\label{fn:Fierz}The 10d Fierz identity that we will use most frequently
is 
\[
\gamma_{a(\alpha\beta}\gamma_{\gamma)\delta}^{a}=0\quad\mbox{or }\gamma^{a(\alpha\beta}\gamma_{a}^{\gamma)\delta}=0
\]
So in particular 
\[
(\rho\gamma_{a}\rho)(\gamma^{a}\rho)_{\alpha}=0
\]
or, when we work with pure spinors $\lambda$:
\[
(\gamma_{a}\lambda)_{\alpha}(\gamma^{a}\lambda)_{\beta}=\tfrac{1}{2}(\lambda\gamma_{a}\lambda)\gamma_{\alpha\beta}^{a}=0
\]

The less well known Fierz identity given in footnote \ref{fn:Fierz-id}
can be found in e.g. \cite{Berkovits:2010zz}, equation (2.12)\frem{Is
there a standard reference for such Fierz identities?}. One can find
how to derive this kind of chiral Fierz identities e.g. around page
179 of \cite{Guttenberg:2008ic}. The first identity of above is there
given in equation (D.160) (p.180), while a symmetrized version of
the identity given in footnote \ref{fn:Fierz-id} can there be found
in equation (D.166), which would actually already be enough for our
purposes here. Nevertheless, in order to get the full one of footnote
\ref{fn:Fierz-id}, one can take there a linear combination of equations
(D.154) on page 179 and (D.70) on page 174 (taking the chiral block
with the index structure that matches the identity in footnote \ref{fn:Fierz-id}).$\quad\fussend$%
} That $\auxi$ and $\auxii^{a}$ vanish for pure spinors away from
the origin (\ref{xzvanishforpure}) follows directly from their definitions
in (\ref{zxDefandfzero}). Equations (\ref{xzscaling}) are also obvious
from the definitions. 

If $\rho$ and its complex conjugate are 'proportional' in the sense
$\rho^{\alpha}\propto(\gamma_{b}\bar{\rho})^{\alpha}$, or more explicitly
$\rho^{\alpha}=\alpha^{b}(\rho,\bar{\rho})(\gamma_{b}\bar{\rho})^{\alpha}$
for some complex-valued $\alpha^{b}(\rho,\bar{\rho})$, or equivalently
$\bar{\rho}_{\alpha}=\bar{\alpha}^{b}(\gamma_{b}\rho)_{\alpha}$ and
$(\rho\bar{\rho})\neq0$ (which implies with the preceding assumption
that $\bar{\alpha}_{b}(\rho,\bar{\rho})(\rho\gamma^{b}\rho)\neq0$),
then we have 
\begin{eqnarray}
\tfrac{1}{2}\auxii^{a}(\gamma_{a}\bar{\rho})^{\alpha} & = & \tfrac{1}{2}\frac{\rho\gamma^{a}\rho}{\bar{\alpha}^{c}(\rho\gamma_{c}\rho)}(\underbrace{\gamma_{a}\gamma_{b}}_{-\gamma_{b}\gamma_{a}\lqn{{\scriptstyle +2\delta_{ab}}}}\rho)^{\alpha}\bar{\alpha}^{b}=\\
 & \stackrel{{\rm Fierz}}{=} & \frac{(\rho\gamma_{b}\rho)\bar{\alpha}^{b}}{(\rho\gamma_{c}\rho)\bar{\alpha}^{c}}\rho^{\alpha}=\rho^{\alpha}\quad\surd
\end{eqnarray}
This shows that if $\rho^{\alpha}$ is 'proportional' to $\bar{\rho}_{\alpha}$
in the sense $\rho^{\alpha}=\alpha^{b}(\gamma_{b}\bar{\rho})^{\alpha}$,
then the expansion coefficients $\alpha^{b}$ can always be chosen
to be 
\begin{equation}
\alpha^{b}=\tfrac{1}{2}\auxii^{b}=\tfrac{\rho\gamma^{b}\rho}{2(\rho\bar{\rho})}\label{alphafixed}
\end{equation}
and thus proves the second equivalence relation in (\ref{xoneifreal}).

In order to show the first equivalence relation in (\ref{xoneifreal}),
we will show for the absolute value square of the difference $\abs{\rho-\tfrac{1}{2}\auxii^{a}(\gamma_{a}\bar{\rho})}^{2}=0\iff\auxi=1\:\mbox{or }\rho\bar{\rho}=0$:
\begin{eqnarray}
\hspace{-1cm}\abs{\rho-\tfrac{1}{2}\auxii^{a}(\gamma_{a}\bar{\rho})}^{2} & = & \left(\rho-\tfrac{1}{2}\auxii^{a}(\gamma_{a}\bar{\rho})\right)\left(\bar{\rho}-\tfrac{1}{2}\bar{\auxii}^{b}(\gamma_{b}\rho)\right)=\\
 & = & \rho\bar{\rho}-\tfrac{1}{2}\auxii^{a}(\bar{\rho}\gamma_{a}\bar{\rho})-\tfrac{1}{2}\bar{\auxii}^{a}(\rho\gamma_{a}\rho)+\tfrac{1}{4}\auxii^{a}\bar{\auxii}^{b}(\bar{\rho}\underbrace{\gamma_{a}\gamma_{b}}_{-\gamma_{b}\gamma_{a}\lqn{{\scriptstyle +2\delta_{ab}}}}\rho)=\qquad\\
 & = & \rho\bar{\rho}\left(1-\auxi\right)
\end{eqnarray}
This completes the proof of (\ref{xoneifreal}). 

In order to prove (\ref{xin01}), we remember the definition $\auxi=\frac{1}{2}\auxii^{a}\bar{\auxii}_{a}$
and thus 
\begin{equation}
\mathbb{R}\ni\auxi\geq0\quad\mbox{with }\auxi=0\,\mbox{only if }\auxii^{a}=0
\end{equation}
Now assume that $\auxi>1$ and consider the spinor $\tfrac{1}{2}\frac{\auxii^{a}(\gamma_{a}\bar{\rho})^{\alpha}}{1+i\sqrt{\auxi-1}}$.
The choice of this spinor is suggested by the appearance of such a
term in the projector map (\ref{ProjGeneral}) (with $\sqrt{-1}\equiv i$),
but the following argument is independent from our motivation to look
at this particular spinor. We claim that this spinor would be equal
to $\rho^{\alpha}$ for every $\auxi>1$ by calculating the modulus
squared of the difference: 
\begin{eqnarray}
\lqn{\Abs{\rho^{\alpha}-\tfrac{1}{2}\frac{\auxii^{a}(\gamma_{a}\bar{\rho})^{\alpha}}{1+i\sqrt{\auxi-1}}}^{2}=}\nonumber \\
 & \stackrel{\auxi>1}{=} & \left(\rho^{\alpha}-\tfrac{1}{2}\frac{\auxii^{a}(\bar{\rho}\gamma_{a})^{\alpha}}{1+i\sqrt{\auxi-1}}\right)\left(\bar{\rho}_{\alpha}-\tfrac{1}{2}\frac{\bar{\auxii}_{b}(\gamma^{b}\rho)_{\alpha}}{1-i\sqrt{\auxi-1}}\right)=\\
 & = & (\rho\bar{\rho})-\tfrac{1}{2}\frac{\bar{\auxii}_{b}(\rho\gamma^{b}\rho)}{1-i\sqrt{\auxi-1}}-\tfrac{1}{2}\frac{\auxii^{a}(\bar{\rho}\gamma_{a}\bar{\rho})}{1+i\sqrt{\auxi-1}}+\tfrac{1}{4}\frac{\auxii^{a}(\bar{\rho}\overbrace{\gamma_{a}\gamma^{b}}^{-\gamma^{b}\gamma_{a}\lqn{{\scriptstyle +2\delta_{a}^{b}}}}\rho)\bar{\auxii}_{b}}{1+\auxi-1}=\\
 & \stackrel{(\ref{zsquarezero})}{=} & (\rho\bar{\rho})\left(2-\frac{\auxi}{1-i\sqrt{\auxi-1}}-\frac{\auxi}{1+i\sqrt{\auxi-1}}\right)=\\
 & = & 0
\end{eqnarray}
However, a few lines above we have seen that a linear combination
of $(\gamma_{a}\bar{\rho})^{\alpha}$ can be equal to $\rho$ if and
only if $\auxi=1$ which disproves the assumption $\auxi>1$ and thus
shows that%
\footnote{\label{fn:Fierz-id}This was a very indirect proof of $\auxi\leq1$.
One might think that a more direct way is to use the following Fierz
identity (see also footnote \ref{fn:Fierz} on page \pageref{fn:Fierz}):
\begin{eqnarray*}
4\gamma_{\delta\beta}^{a}\gamma_{a}^{\alpha\gamma} & = & 8\delta_{\delta}^{\alpha}\delta_{\beta}^{\gamma}+2\delta_{\beta}^{\alpha}\delta_{\delta}^{\gamma}-\gamma^{ab\alpha}\tief{\beta}\gamma_{ba}\hoch{\gamma}\tief{\delta}
\end{eqnarray*}
Contracting it with two $\rho'$s and two $\bar{\rho}$'s, we obtain
\begin{eqnarray*}
4(\rho\gamma^{a}\rho)(\bar{\rho}\gamma_{a}\bar{\rho}) & = & 10(\rho\bar{\rho})^{2}-(\bar{\rho}\gamma^{ab}\rho)(\bar{\rho}\gamma_{ba}\rho)\leq10(\rho\bar{\rho})^{2}
\end{eqnarray*}
Strange enough, this shows only $\auxi\leq\tfrac{5}{4}$ which is
weaker than what we had before. The bound $\auxi=\tfrac{5}{4}$ would
be reached only for $(\bar{\rho}\gamma^{ab}\rho)=0$. Together with
the proven $\auxi\leq1$ this implies 
\[
(\bar{\rho}\gamma^{ab}\rho)=0\iff\rho\bar{\rho}=0\qquad\fussend
\]
} 
\begin{equation}
\auxi\leq1\label{xsmaller1}
\end{equation}
The proof of the properties of the auxiliary variables is now complete.
From now on let us stick to the order in the proposition and continue
with its 1st statement.$\quad\hfill{\scriptstyle \square}$

1. The second projection property (\ref{proj-prop2}) follows directly
from (\ref{xzvanishforpure}) and the fact that we required $f(0)=1$.
What remains to show is that the spinor indeed becomes pure after
the mapping:
\begin{eqnarray}
\lqn{P_{(f)}^{\alpha}(\rho,\bar{\rho})\gamma_{\alpha\beta}^{c}P_{(f)}^{\beta}(\rho,\bar{\rho})=}\nonumber \\
 & = & f\left(\auxi\right)^{2}\left(\rho^{\alpha}-\tfrac{1}{2}\frac{\auxii^{a}(\bar{\rho}\gamma_{a})^{\alpha}}{1+\sqrt{1-\auxi}}\right)\gamma_{\alpha\beta}^{c}\left(\rho^{\beta}-\tfrac{1}{2}\frac{\auxii^{b}(\bar{\rho}\gamma_{b})^{\beta}}{1+\sqrt{1-\auxi}}\right)=\\
 & = & f\left(\auxi\right)^{2}\Biggl((\rho\gamma^{c}\rho)-\frac{\auxii^{a}(\rho\overbrace{\gamma^{c}\gamma_{a}}^{-\gamma_{a}\gamma^{c}\lqn{{\scriptstyle +2\delta_{a}^{c}}}}\bar{\rho})}{1+\sqrt{1-\auxi}}+\tfrac{1}{4}\frac{\auxii^{a}\auxii^{b}(\bar{\rho}\gamma_{a}\overbrace{\gamma^{c}\gamma_{b}}^{-\gamma_{b}\gamma^{c}\lqn{{\scriptstyle +2\delta_{b}^{c}}}}\bar{\rho})}{\left(1+\sqrt{1-\auxi}\right)^{2}}\Biggr)=\\
 & \stackrel{(\ref{zsquarezero})}{=} & \frac{f\left(\auxi\right)^{2}}{\rho\bar{\rho}}\Biggl(\auxii^{c}-\frac{2\auxii^{c}}{1+\sqrt{1-\auxi}}+\frac{\auxii^{c}\auxi}{\left(1+\sqrt{1-\auxi}\right)^{2}}\Biggr)=0\quad\surd
\end{eqnarray}
This proves (\ref{proj-prop2}).$\quad\hfill{\scriptstyle \square}$

2. The homogeneity (\ref{homogenous}) follows directly from the scaling
behaviour (\ref{xzscaling}) of the variables $\auxi$ and $\auxii^{a}$.
The modulus square of the projected spinor is
\begin{eqnarray}
\lqn{P_{(f)}^{\alpha}(\rho,\bar{\rho})\bar{P}_{(f)\alpha}(\rho,\bar{\rho})=}\nonumber \\
 & = & f\left(\auxi\right)\bar{f}\left(\auxi\right)\left(\rho^{\alpha}-\tfrac{1}{2}\frac{\auxii^{a}(\bar{\rho}\gamma_{a})^{\alpha}}{1+\sqrt{1-\auxi}}\right)\left(\bar{\rho}_{\alpha}-\tfrac{1}{2}\frac{\bar{\auxii}^{b}(\gamma_{b}\rho)_{\alpha}}{1+\sqrt{1-\auxi}}\right)=\\
 & = & \abs{f\left(\auxi\right)}^{2}\biggl((\rho\bar{\rho})-\tfrac{1}{2}\tfrac{\bar{\auxii}^{b}(\rho\gamma_{b}\rho)}{1+\sqrt{1-\auxi}}-\tfrac{1}{2}\tfrac{\auxii^{a}(\bar{\rho}\gamma_{a}\bar{\rho})}{1+\sqrt{1-\auxi}}+\tfrac{1}{4}\tfrac{\auxii^{a}(\bar{\rho}\overbrace{\gamma_{a}\gamma_{b}}^{-\gamma_{b}\gamma_{a}\lqn{{\scriptscriptstyle +2\delta_{ab}}}}\rho)\bar{\auxii}^{b}}{\left(1+\sqrt{1-\auxi}\right)^{2}}\biggr)=\\
 & \!\!\!\!\stackrel{(\ref{zsquarezero})}{=}\!\!\!\! & \abs{f\left(\auxi\right)}^{2}(\rho\bar{\rho})\biggl(\!1\!-\tfrac{1}{2}\tfrac{\bar{\auxii}^{b}\auxii_{b}}{1+\sqrt{1-\auxi}}-\tfrac{1}{2}\tfrac{\auxii^{a}\bar{\auxii}_{a}}{1+\sqrt{1-\auxi}}+\tfrac{1}{2}\tfrac{\auxii^{a}\bar{\auxii}_{a}}{\left(1+\sqrt{1-\auxi}\right)^{2}}\biggr)=\qquad\\
 & = & 2(\rho\bar{\rho})\abs{f\left(\auxi\right)}^{2}\left(\frac{1-\auxi}{1+\sqrt{1-\auxi}}\right)
\end{eqnarray}
This agrees with the claim in (\ref{Proj-modulus}). At the origin
$(\rho\bar{\rho})=0$ the variable $\auxi$ is ill-defined, but when
we assume $f(\auxi)$ to be continuous on the whole interval $[0,1]$
then it will clearly be bounded and we have a well defined limit when
approaching the origin:
\begin{eqnarray}
\lqn{\lim_{\abs{\rho}\to0}P_{(f)}^{\alpha}(\rho,\bar{\rho})\bar{P}_{(f)\alpha}(\rho,\bar{\rho})=}\nonumber \\
 & = & \lim_{\abs{\rho}\to0}2\abs{\rho}^{2}\abs{f\left(\auxi\right)}^{2}\left(\frac{1-\auxi}{1+\sqrt{1-\auxi}}\right)\leq\\
 & \stackrel{0\leq\auxi\leq1}{\leq} & \lim_{\abs{\rho}\to0}2\abs{\rho}^{2}\abs{f\left(\auxi\right)}^{2}=\\
 & \stackrel{f(\auxi){\rm bounded}}{=} & 0
\end{eqnarray}
 This proves (\ref{limitToZero}). Note that the projection properties
(\ref{proj-prop1}) and (\ref{proj-prop2}) also obviously hold in
this limit.$\quad\hfill{\scriptstyle \square}$

3. From the above result we see that the zero-vector is in the zero-locus
of $P_{(f)}^{\alpha}$, at least if the latter is analytically continued
to that point. Looking at the definition (\ref{ProjGeneral}) of the
projector in the remaining regime, it is obvious that it vanishes
if and only if either $f(\auxi)=0$ or the term in the bracket, i.e.
$\rho^{\alpha}-\tfrac{1}{2}\frac{\auxii^{a}(\bar{\rho}\gamma_{a})^{\alpha}}{1+\sqrt{1-\auxi}}$
is zero. From (\ref{xoneifreal}) we know that this is the case if
and only if $\rho^{\alpha}=\tfrac{1}{2}\auxii^{a}(\bar{\rho}\gamma_{a})^{\alpha}$
(or $\auxi=1$). This completes the proof of (\ref{kerP}). Note that
writing $(\rho\bar{\rho})\rho^{\alpha}=\tfrac{1}{2}(\rho\gamma^{a}\rho)(\bar{\rho}\gamma_{a})^{\alpha}$
instead of $\rho^{\alpha}=\tfrac{1}{2}\auxii^{a}(\bar{\rho}\gamma_{a})^{\alpha}$
would include also the case $\rho^{\alpha}=0$ into this set. Now
let us assume that $\rho^{\alpha}$ is real, which is a non-covariant
statement: 
\begin{equation}
\underbrace{\bar{\rho}_{\alpha}}_{(\rho^{\alpha})^{*}}=\rho^{\alpha}\quad({\rm assumption})\label{real-assump}
\end{equation}
Remember from footnote \ref{fn:10dgamma} on page \pageref{fn:10dgamma}
that we have $\gamma^{10\alpha\beta}=-\gamma_{\alpha\beta}^{10}=i\delta_{\alpha\beta}$
($\gamma_{\alpha\beta}^{0}=-\gamma^{0\alpha\beta}=\gamma_{0}^{\alpha\beta}=-\delta_{\alpha\beta}$).
Using the above assumption for the following expression, we thus obtain
\begin{eqnarray}
\tfrac{\auxii^{a}}{2}(\gamma_{a}\bar{\rho})^{\alpha} & = & \sum_{i=1}^{9}\tfrac{(\rho\gamma^{i}\rho)}{2}\underbrace{(\gamma_{i}\bar{\rho})^{\alpha}}_{(\gamma_{i}\rho)_{\alpha}}+\tfrac{(\rho^{\gamma}\gamma_{\gamma\delta}^{10}\rho^{\delta})}{2(\rho\bar{\rho})}(\underbrace{\gamma_{10}^{\alpha\beta}}_{-\gamma_{\alpha\beta}^{10}}\underbrace{\bar{\rho}_{\beta}}_{\rho^{\beta}})=\label{same-calc-I}\\
 & \stackrel{(\ref{real-assump})}{=} & \tfrac{(\rho\gamma^{a}\rho)}{2(\rho\rho)}(\gamma_{a}\rho)^{\alpha}-\tfrac{(\rho^{\gamma}\overbrace{\gamma_{\gamma\delta}^{10}}^{-i\delta_{\gamma\delta}}\rho^{\delta})}{(\rho\rho)}(\underbrace{\gamma_{\alpha\beta}^{10}}_{-i\delta_{\alpha\beta}}\rho^{\beta})=\\
 & = & \rho^{\alpha}\label{same-calc-III}
\end{eqnarray}
From what we have proved before, this shows that $\auxi=1$ if $\rho$
is real.%
\footnote{Note that the same calculation (\ref{same-calc-I})-(\ref{same-calc-III})
goes through if $\rho^{\alpha}$ is only almost real in the sense
$\bar{\rho}_{\alpha}=c\rho^{\alpha}$ for some $c\in\mathbb{C}$ (implying
that $\rho^{\alpha}$ is a complex multiple of some real spinor $\tilde{\rho}^{\alpha}$,
i.e. $\rho^{\alpha}=\tilde{c}\tilde{\rho}^{\alpha},\quad\tilde{c}\in\mathbb{C}$).
Also in this case $\auxi=1.\quad\mbox{\ensuremath{\fussend}}$%
} The complete zero-locus is SO(10) covariant. Any SO(10) rotation
of a real spinor (the result of the rotation is in general not real
any longer) thus also has to lie in the zero-locus. The corresponding
(infinitesimally) rotated reality condition would be 
\[
(L_{ab}\gamma^{ab}\tief{\alpha}\hoch{\beta}\bar{\rho}_{\beta})=(L_{ab}\gamma^{ab\:\alpha}\tief{\beta}\rho^{\beta})
\]
with any antisymmetric parameter $L_{ab}$.$\quad\hfill{\scriptstyle \square}$

4. Continuity is obvious. If $f$ is differentiable, then there are
only two possible problems for differentiability of the projector.
One is at $\rho^{\alpha}=0$, where $\auxi$ and $\auxii^{a}$ are
not well defined. And the other is at $\auxi=1$, as the square root
is not differentiable at $0$. We will later study the variations
(and thus the derivatives) of the projector and will see in footnote
\ref{fn:endOfproof4} that they have a pole for $\abs{\rho}\to0$
and also in general for $\auxi=1$, but that the latter can be avoided
by choosing $f(\auxi)=\tilde{f}(\auxi)\left(1-\auxi\right)^{1+r}$
with $\tilde{f}$ differentiable everywhere. This will complete the
proof of proposition \ref{prop:covproj}.\end{proof}

\subsection{Proof of proposition 2}

\label{app:proof2}\begin{proof}[Proof of Proposition \ref{prop:linProj}.]Remember
$\auxii^{a}=\frac{\rho\gamma^{a}\rho}{\rho\bar{\rho}},\:\auxi=\tfrac{1}{2}\auxii^{a}\bar{\auxii}_{a}$
of equation \eqref{zxDefandfzero}. Their derivatives read%
\footnote{\label{fn:zx-derivs-lambda}Using (\ref{Proj-inv}), the derivatives
of the auxiliary variables can also be written in terms of $\lambda^{\alpha}\equiv P_{(f)}^{\alpha}(\rho,\bar{\rho})$:
\begin{eqnarray*}
\partial_{\rho^{\beta}}\auxii^{a} & = & \frac{\sqrt{1-\auxi}}{(\lambda\bar{\lambda})}\Bigl\{2\bar{f}(\auxi)\left(\lambda\gamma^{a}\right)_{\beta}-\frac{\bar{f}(\auxi)}{2\left(1+\sqrt{1-\auxi}\right)}\auxii^{a}\bar{\auxii}_{b}\left(\lambda\gamma^{b}\right)_{\beta}+\\
 &  & -\frac{f(\auxi)}{1+\sqrt{1-\auxi}}\auxii^{b}\left(\gamma_{b}\gamma^{a}\bar{\lambda}\right)_{\beta}+\frac{1-\sqrt{1-\auxi}}{1+\sqrt{1-\auxi}}\auxii^{a}f(\auxi)\bar{\lambda}_{\beta}\Bigr\}\qquad\\
\partial_{\bar{\rho}_{\beta}}\auxii^{a} & = & -\auxii^{a}\frac{\sqrt{1-\auxi}}{(\lambda\bar{\lambda})}\Bigl\{\bar{f}(\auxi)\lambda^{\beta}+\frac{f(\auxi)}{2\left(1+\sqrt{1-\auxi}\right)}\auxii^{b}\left(\bar{\lambda}\gamma_{b}\right)^{\beta}\Bigr\}\\
\partial_{\rho^{\beta}}\auxi & = & \frac{\left(1-\auxi\right)}{(\lambda\bar{\lambda})}\Bigl\{\bar{f}(\auxi)\bar{\auxii}_{c}\left(\gamma^{c}\lambda\right)_{\beta}-\frac{2\auxi}{1+\sqrt{1-\auxi}}f(\auxi)\bar{\lambda}_{\beta}\Bigr\}
\end{eqnarray*}
Written in this form it is particularly easy to see that on the subspace
where $\auxi=1$ (basically the zero-locus of the projection), the
$\rho$-derivative of $\auxi$ vanishes
\[
\bei{\partial_{\rho^{\beta}}\auxi}{\auxi=1}=0
\]
while at the constraint surface (where $\auxi=\auxii^{a}=\bar{\auxii}_{a}=0$)
we have even 
\[
\bei{\partial_{\rho^{\beta}}\auxi}{\auxi=\auxii^{a}=0}=\bei{\partial_{\bar{\rho}_{\beta}}\auxii^{a}}{\auxi=\auxii^{a}=0}=0,\quad\bei{\partial_{\rho^{\beta}}\auxii^{a}}{\auxi=\auxii^{a}=0}=\frac{2\left(\lambda\gamma^{a}\right)_{\beta}}{(\lambda\bar{\lambda})}
\]
So in both subspaces the $\rho$-derivative of $\auxi$ vanishes which
means that $\auxi$ is stationary (so e.g. extremal) at these values.
This agrees with the observation that $0\leq\auxi\leq1$.$\qquad\fussend$%
} 
\begin{eqnarray}
\partial_{\rho^{\beta}}\auxii^{a} & = & \frac{2(\gamma^{a}\rho)_{\beta}-\auxii^{a}\bar{\rho}_{\beta}}{\rho\bar{\rho}}\label{zbyrho-derivative}\\
\partial_{\bar{\rho}_{\beta}}\auxii^{a} & = & -\frac{\auxii^{a}\rho^{\beta}}{\rho\bar{\rho}}\label{zbyrhobar-derivative}\\
\partial_{\rho^{\beta}}\auxi & = & \tfrac{1}{2}\partial_{\rho^{\beta}}\auxii^{a}\bar{\auxii}_{a}+\tfrac{1}{2}\auxii^{a}\partial_{\rho^{\beta}}\bar{\auxii}_{a}=\frac{\bar{\auxii}_{a}(\gamma^{a}\rho)_{\beta}-2\auxi\bar{\rho}_{\beta}}{\rho\bar{\rho}}\label{xbyrhoderivative}
\end{eqnarray}

1. Now we can apply the $\rho$-derivative to the projector (\ref{ProjGeneral}):
\begin{eqnarray}
\lqn{\partial_{\rho^{\beta}}P_{(f)}^{\alpha}(\rho,\bar{\rho})=}\nonumber \\
 & = & \!\!\!\!\partial_{\rho^{\beta}}\auxi\cdot f'\left(\auxi\right)\left(\rho^{\alpha}-\tfrac{1}{2}\frac{\auxii^{a}(\bar{\rho}\gamma_{a})^{\alpha}}{1+\sqrt{1-\auxi}}\right)+\nonumber \\
 &  & \!\!\!\!+f\left(\auxi\right)\left(\delta_{\beta}^{\alpha}-\tfrac{1}{2}\frac{\partial_{\rho^{\beta}}\auxii^{a}(\bar{\rho}\gamma_{a})^{\alpha}}{1+\sqrt{1-\auxi}}-\tfrac{1}{2}\frac{\partial_{\rho^{\beta}}\auxi\cdot\auxii^{a}(\bar{\rho}\gamma_{a})^{\alpha}}{2\sqrt{1-\auxi}\left(1+\sqrt{1-\auxi}\right)^{2}}\right)=\qquad\label{partialPintermed}\\
 & \stackrel{\mbox{\tiny\eqref{zbyrho-derivative}\eqref{xbyrhoderivative}}}{=} & f\left(\auxi\right)\delta_{\beta}^{\alpha}+f'\left(\auxi\right)\frac{\bar{\auxii}_{a}\rho^{\alpha}(\gamma^{a}\rho)_{\beta}}{\rho\bar{\rho}}-2f'\left(\auxi\right)\auxi\frac{\rho^{\alpha}\bar{\rho}_{\beta}}{\rho\bar{\rho}}+\nonumber \\
 &  & \hspace{-1.5cm}-\!\!\left(\!\!\frac{f\left(\auxi\right)}{1+\sqrt{1-\auxi}}\delta_{b}^{a}\!+\!\frac{f\left(\auxi\right)\auxii^{a}\bar{\auxii}_{b}}{4\sqrt{1-\auxi}\left(1+\sqrt{1-\auxi}\right)^{2}}\!+\!\tfrac{1}{2}\frac{f'\left(\auxi\right)\auxii^{a}\bar{\auxii}_{b}}{1+\sqrt{1-\auxi}}\!\!\right)\!\!\frac{(\gamma_{a}\bar{\rho})^{\alpha}(\rho\gamma^{b})_{\beta}}{\rho\bar{\rho}}+\nonumber \\
 &  & \hspace{-1.5cm}+\!\!\left(\!\!\frac{f'\left(\auxi\right)\auxi}{1+\sqrt{1-\auxi}}\!+\!\frac{f\left(\auxi\right)}{2\left(1+\sqrt{1-\auxi}\right)}\!+\!\frac{f\left(\auxi\right)\auxi}{2\sqrt{1-\auxi}\left(1+\sqrt{1-\auxi}\right)^{2}}\!\!\right)\!\!\frac{\auxii^{a}(\bar{\rho}\gamma_{a})^{\alpha}\bar{\rho}_{\beta}}{\rho\bar{\rho}}\qquad\quad
\end{eqnarray}
After a slight rewriting this agrees with (\ref{LinProjGen}). The
derivative with respect to $\bar{\rho}_{\beta}$ instead has the form
\begin{eqnarray}
\lqn{\partial_{\bar{\rho}_{\beta}}P_{(f)}^{\alpha}(\rho,\bar{\rho})=}\nonumber \\
 & \hspace{-1cm}= & \partial_{\bar{\rho}_{\beta}}\auxi\cdot f'\left(\auxi\right)\left(\rho^{\alpha}-\tfrac{1}{2}\frac{\auxii^{a}(\bar{\rho}\gamma_{a})^{\alpha}}{1+\sqrt{1-\auxi}}\right)+\nonumber \\
 &  & \hspace{-1cm}+f\!\left(\auxi\right)\!\left(\!\!-\tfrac{1}{2}\frac{\partial_{\bar{\rho}_{\beta}}\auxii^{a}(\bar{\rho}\gamma_{a})^{\alpha}}{1+\sqrt{1-\auxi}}-\tfrac{1}{2}\frac{\auxii^{a}\gamma_{a}^{\alpha\beta}}{1+\sqrt{1-\auxi}}-\tfrac{1}{2}\frac{\partial_{\bar{\rho}_{\beta}}\auxi\auxii^{a}(\bar{\rho}\gamma_{a})^{\alpha}}{2\sqrt{1-\auxi}\left(1+\sqrt{1-\auxi}\right)^{2}}\!\!\right)\!=\qquad\label{barpartialPinter}\\
 & \hspace{-1cm}\stackrel{\tiny\eqref{zbyrhobar-derivative}\eqref{xbyrhoderivative}}{=} & \frac{\auxii^{a}(\gamma_{a}\bar{\rho})^{\beta}-2\auxi\rho^{\beta}}{\rho\bar{\rho}}\times\nonumber \\
 &  & \hspace{-1cm}\times\left(f'\left(\auxi\right)\left(\rho^{\alpha}-\tfrac{1}{2}\frac{\auxii^{a}(\bar{\rho}\gamma_{a})^{\alpha}}{1+\sqrt{1-\auxi}}\right)-\tfrac{1}{2}f\left(\auxi\right)\frac{\auxii^{a}(\bar{\rho}\gamma_{a})^{\alpha}}{2\sqrt{1-\auxi}\left(1+\sqrt{1-\auxi}\right)^{2}}\right)+\nonumber \\
 &  & \hspace{-1cm}+f\left(\auxi\right)\left(-\tfrac{1}{2}\frac{-\frac{\auxii^{a}\rho^{\beta}}{\rho\bar{\rho}}(\bar{\rho}\gamma_{a})^{\alpha}}{1+\sqrt{1-\auxi}}-\tfrac{1}{2}\frac{\auxii^{a}\gamma_{a}^{\alpha\beta}}{1+\sqrt{1-\auxi}}\right)
\end{eqnarray}
Collecting the terms, one arrives at (\ref{linprojgen})%
\footnote{\label{fn:endOfproof4} Looking at the denominators of (\ref{LinProjGen})
and (\ref{linprojgen}) it becomes clear that the only possible poles
of $\Pi_{(f)\bot}$ and $\rest_{(f)\bot}$ are either at $(\rho\bar{\rho})=0$,
at $\auxi=1$ or at some poles of $f$ itself. This shows that $P_{(f)}^{\alpha}$
is differentiable everywhere but at $\{0\}\cup\{\rho^{\alpha}|\auxi=1\}$
if $f$ is everywhere differentiable. This was one of the statements
of the 4th point of proposition \ref{prop:covproj}. Furthermore the
singularities $\frac{1}{\sqrt{1-\auxi}}$ come with $f(\auxi)$ and
could be removed by $f(\auxi)\propto\sqrt{1-\auxi}$ which, however,
would introduce new singularities in other terms because of $f'(\auxi)\propto\frac{1}{\sqrt{1-\auxi}}$.
A save choice instead would be 
\begin{equation}
f(\auxi)=\tilde{f}(\auxi)(1-\auxi)^{1+r}\,,\quad\tilde{f}(0)=1\,,\quad r\geq0
\end{equation}
with $\tilde{f}$ differentiable everywhere. The exponent $(1+r)$
guarantees that the derivative 
\begin{equation}
f'(\auxi)=\tilde{f}'(\auxi)(1-\auxi)^{1+r}-(1+r)\tilde{f}(\auxi)(1-\auxi)^{r}
\end{equation}
does not have singularities, while at the same time the factor $(1-\auxi)^{1+r}$
will cure the singularities of the form $\frac{1}{\sqrt{1-\auxi}}$
which come along with $f$ only. This completes the proof of point
4 of proposition \ref{prop:covproj}.$\quad\fussend$%
}. $\quad\hfill{\scriptstyle \square}$\newpage 

2. Using (\ref{Proj-inv}) and (\ref{Proj-modulus-repeated}), we
can rewrite the the linear projection matrix (\ref{LinProjGen}) in
terms of $\lambda^{\alpha}$ and its complex conjugate:{\small 
\begin{eqnarray}
\lqn{\Pi_{(f)\bot\beta}^{\alpha}(\rho,\bar{\rho})=f\left(\auxi\right)\delta_{\beta}^{\alpha}+}\nonumber \\
 &  & \hspace{-0.5cm}+f'\!(\auxi)\!\frac{\!\!\left(\!\frac{1+\sqrt{1-\auxi}}{2f(\auxi)\sqrt{1-\auxi}}\lambda^{\alpha}\!+\!\frac{1}{4\bar{f}(\auxi)\sqrt{1-\auxi}}\auxii^{c}\left(\bar{\lambda}\gamma_{c}\right)^{\alpha}\!\right)\!\!\Bigl(\!\frac{1+\sqrt{1-\auxi}}{2f(\auxi)\sqrt{1-\auxi}}\bar{\auxii}_{b}(\lambda\gamma^{b})_{\beta}\!+\!\frac{1}{4\os{}{\bar{f}(\auxi)}\sqrt{1-\auxi}}\auxii^{d}\bar{\auxii}_{b}(\bar{\lambda}\overbrace{\gamma_{d}\gamma^{b}}^{\hspace{-1cm}-\gamma^{b}\gamma_{d}\lqn{{\scriptstyle +2\delta_{d}^{b}}}})_{\beta}\!\Bigr)\!\!}{\frac{\left(1+\sqrt{1-\auxi}\right)}{2\abs{f(\auxi)}^{2}\left(1-\auxi\right)}(\lambda\bar{\lambda})}+\nonumber \\
 &  & \hspace{-0.5cm}-2\auxi f'\!\left(\auxi\right)\!\frac{\!\!\left(\!\frac{1+\sqrt{1-\auxi}}{2f(\auxi)\sqrt{1-\auxi}}\lambda^{\alpha}\!+\!\frac{1}{4\bar{f}(\auxi)\sqrt{1-\auxi}}\auxii^{c}\left(\bar{\lambda}\gamma_{c}\right)^{\alpha}\!\right)\!\!\left(\!\frac{1+\sqrt{1-\auxi}}{2\bar{f}(\auxi)\sqrt{1-\auxi}}\bar{\lambda}_{\beta}\!+\!\frac{1}{4f(\auxi)\sqrt{1-\auxi}}\bar{\auxii}_{d}\left(\lambda\gamma^{d}\right)_{\beta}\!\right)\!\!}{\frac{\left(1+\sqrt{1-\auxi}\right)}{2\abs{f(\auxi)}^{2}\left(1-\auxi\right)}(\lambda\bar{\lambda})}+\nonumber \\
 &  & \hspace{-0.5cm}-\tfrac{1}{1+\sqrt{1-\auxi}}\left(f\left(\auxi\right)\delta_{b}^{a}+\Bigl(\tfrac{f\left(\auxi\right)}{4\sqrt{1-\auxi}\left(1+\sqrt{1-\auxi}\right)}+\tfrac{1}{2}f'(\auxi)\Bigr)\auxii^{a}\bar{\auxii}_{b}\right)\times\nonumber \\
 &  & \hspace{-0.5cm}\times\frac{\!\!\Bigl(\!\frac{1+\sqrt{1-\auxi}}{2\bar{f}(\auxi)\sqrt{1-\auxi}}(\gamma_{a}\bar{\lambda})^{\alpha}\!+\!\frac{1}{4f(\auxi)\sqrt{1-\auxi}}\bar{\auxii}_{c}(\overbrace{\gamma_{a}\gamma^{c}}^{\hspace{-1cm}-\gamma^{c}\gamma_{a}\lqn{{\scriptstyle +2\delta_{a}^{c}}}}\lambda)^{\alpha}\!\Bigr)\!\Bigl(\!\frac{1+\sqrt{1-\auxi}}{2f(\auxi)\sqrt{1-\auxi}}(\lambda\gamma^{b})_{\beta}\!+\!\frac{1}{4\bar{f}(\auxi)\sqrt{1-\auxi}}\auxii^{d}(\bar{\lambda}\overbrace{\gamma_{d}\gamma^{b}}^{\hspace{-1cm}-\gamma^{b}\gamma_{d}\lqn{{\scriptstyle +2\delta_{d}^{b}}}})_{\beta}\!\Bigr)\!\!}{\frac{\left(1+\sqrt{1-\auxi}\right)}{2\abs{f(\auxi)}^{2}\left(1-\auxi\right)}(\lambda\bar{\lambda})}+\nonumber \\
 &  & \hspace{-0.5cm}+\left(\tfrac{f\left(\auxi\right)}{2\sqrt{1-\auxi}\left(1+\sqrt{1-\auxi}\right)}+{\scriptstyle \left(1-\sqrt{1-\auxi}\right)}f'\left(\auxi\right)\right)\times\nonumber \\
 &  & \hspace{-0.5cm}\times\frac{\!\!\Bigl(\!\frac{1+\sqrt{1-\auxi}}{2\bar{f}(\auxi)\sqrt{1-\auxi}}\auxii^{a}(\gamma_{a}\bar{\lambda})^{\alpha}\!+\!\frac{1}{4f(\auxi)\sqrt{1-\auxi}}\bar{\auxii}_{c}\auxii^{a}(\overbrace{\gamma_{a}\gamma^{c}}^{-\gamma^{c}\gamma_{a}\lqn{{\scriptstyle +2\delta_{a}^{c}}}}\lambda)^{\alpha}\!\Bigr)\!\Bigl(\!\frac{1+\sqrt{1-\auxi}}{2\bar{f}(\auxi)\sqrt{1-\auxi}}\bar{\lambda}_{\beta}\!+\!\frac{1}{4f(\auxi)\sqrt{1-\auxi}}\bar{\auxii}_{d}(\lambda\gamma^{d})_{\beta}\!\Bigr)\!\!}{\frac{\left(1+\sqrt{1-\auxi}\right)}{2\abs{f(\auxi)}^{2}\left(1-\auxi\right)}(\lambda\bar{\lambda})}\nonumber \\
{}
\end{eqnarray}
}Changing the order of the gamma-matrices as indicated above the
curly brackets makes it for some of the terms possible to use Fierz
identities of the form $(\gamma^{a}\lambda)_{\alpha}(\gamma_{a}\lambda)_{\beta}\propto(\lambda\gamma_{a}\lambda)\gamma_{\alpha\beta}^{a}=0$
or $\auxii^{a}(\gamma_{a}\lambda)_{\alpha}=0$ (\ref{ProjZsquareZero})
and their complex conjugate counterparts. However, a few terms will
arise where these identities will not be applicable. Nevertheless
the change of order in the gamma-matrices is advantageous, as a multiplication
by a $\lambda$ from the right or a $\bar{\lambda}$ from the left
will trigger these identities. So after replacing the products of
$\gamma$-matrices by the expressions above the curly brackets, using
the just mentioned identities wherever possible, and sorting the terms
by collecting first those of the form $\lambda\otimes\lambda$, then
$\lambda\otimes\bar{\lambda}$ and so on, we obtain (before simplifying)
\begin{eqnarray}
\lqn{\Pi_{(f)\bot\beta}^{\alpha}(\rho,\bar{\rho})=f\left(\auxi\right)\delta_{\beta}^{\alpha}+}\nonumber \\
 &  & \hspace{-0.5cm}+\Biggl\{\!\!\tfrac{\left(1+\sqrt{1-\auxi}\right)f'\!\left(\auxi\right)\bar{f}(\auxi)}{2f(\auxi)}\!-\!\tfrac{\auxi f'\!\left(\auxi\right)\bar{f}(\auxi)}{2f(\auxi)}\!+\nonumber \\
 &  & \hspace{-0.5cm}-\tfrac{\bar{f}(\auxi)}{2f(\auxi)\left(1+\sqrt{1-\auxi}\right)}\!\left(\! f\left(\auxi\right)\!+\!\Bigl(\!\tfrac{f\left(\auxi\right)}{2\sqrt{1-\auxi}\left(1+\sqrt{1-\auxi}\right)}\!+\! f'\!\left(\auxi\right)\!\Bigr)\auxi\!\right)+\nonumber \\
 &  & \hspace{-0.5cm}+\tfrac{\auxi\bar{f}(\auxi)}{2\left(1+\sqrt{1-\auxi}\right)^{2}f(\auxi)}\left(\tfrac{f\left(\auxi\right)}{2\sqrt{1-\auxi}}+\auxi f'\left(\auxi\right)\right)\Biggr\}\frac{\lambda^{\alpha}\bar{\auxii}_{b}(\lambda\gamma^{b})_{\beta}}{(\lambda\bar{\lambda})}\nonumber \\
 &  & \hspace{-0.5cm}+\Biggl\{\!\auxi f'\!\left(\auxi\right)\!-\!\auxi f'\left(\auxi\right)\!{\scriptstyle \left(\!1+\sqrt{1-\auxi}\!\right)}\!-\!\tfrac{\auxi}{\!\!\left(1+\sqrt{1-\auxi}\right)^{\!2}\!}\!\left(\! f\left(\auxi\right)\!+\!\Bigl(\!\tfrac{f\left(\auxi\right)}{2\sqrt{1-\auxi}\left(1+\sqrt{1-\auxi}\right)}\!+\! f'\!\left(\auxi\right)\!\Bigr)\auxi\!\right)+\nonumber \\
 &  & \hspace{-0.5cm}+\tfrac{\auxi}{\left(1+\sqrt{1-\auxi}\right)}\left(\tfrac{f\left(\auxi\right)}{2\sqrt{1-\auxi}}+\auxi f'\left(\auxi\right)\right)\Biggr\}\frac{\lambda^{\alpha}\bar{\lambda}_{\beta}}{(\lambda\bar{\lambda})}\nonumber \\
 &  & \hspace{-0.5cm}-\tfrac{f\left(\auxi\right)}{8\left(1+\sqrt{1-\auxi}\right)^{2}}\frac{\bar{\auxii}_{c}(\gamma^{c}\gamma_{b}\lambda)^{\alpha}\auxii^{d}(\bar{\lambda}\gamma^{b}\gamma_{d})_{\beta}}{(\lambda\bar{\lambda})}+\nonumber \\
 &  & \hspace{-0.5cm}+\biggl\{\tfrac{1}{4}f'\left(\auxi\right)-\tfrac{\auxi f'\left(\auxi\right)}{4\left(1+\sqrt{1-\auxi}\right)}-\tfrac{1}{2}\left(\tfrac{f\left(\auxi\right)}{4\sqrt{1-\auxi}\left(1+\sqrt{1-\auxi}\right)}+\tfrac{1}{2}f'\left(\auxi\right)\right)+\nonumber \\
 &  & \hspace{-0.5cm}+\tfrac{1}{4\left(1+\sqrt{1-\auxi}\right)}\left(\tfrac{f\left(\auxi\right)}{2\sqrt{1-\auxi}}+\auxi f'\left(\auxi\right)\right)\biggr\}\frac{\auxii^{a}\left(\bar{\lambda}\gamma_{a}\right)^{\alpha}\bar{\auxii}_{b}(\lambda\gamma^{b})_{\beta}}{(\lambda\bar{\lambda})}\nonumber \\
 &  & \hspace{-0.5cm}-\tfrac{1}{2}f\left(\auxi\right)\frac{(\gamma_{a}\bar{\lambda})^{\alpha}(\lambda\gamma^{a})_{\beta}}{(\lambda\bar{\lambda})}+\nonumber \\
 &  & \hspace{-0.5cm}+\Biggl\{\tfrac{f'\left(\auxi\right)f(\auxi)\auxi}{2\bar{f}(\auxi)\left(1+\sqrt{1-\auxi}\right)}-\tfrac{\auxi f'\left(\auxi\right)f(\auxi)}{2\bar{f}(\auxi)}+\nonumber \\
 &  & \hspace{-0.5cm}-\tfrac{f(\auxi)}{2\left(1+\sqrt{1-\auxi}\right)\bar{f}(\auxi)}\left(f\left(\auxi\right)+\biggl(\tfrac{f\left(\auxi\right)}{4\sqrt{1-\auxi}\left(1+\sqrt{1-\auxi}\right)}+\tfrac{1}{2}f'\left(\auxi\right)\biggr)2\auxi\right)+\nonumber \\
 &  & \hspace{-0.5cm}+\tfrac{f(\auxi)}{2\bar{f}(\auxi)}\left(\tfrac{f\left(\auxi\right)}{2\sqrt{1-\auxi}}+\auxi f'\left(\auxi\right)\right)\Biggr\}\frac{\auxii^{a}\left(\gamma_{a}\bar{\lambda}\right)^{\alpha}\bar{\lambda}_{\beta}}{(\lambda\bar{\lambda})}
\end{eqnarray}
In a final effort we sort the terms in the curly brackets to terms
that have an $f'$ and those that have not and simplify the expressions.
It turns out that many terms cancel, in particular the curly bracket
before$\frac{\auxii^{a}\left(\gamma_{a}\bar{\lambda}\right)^{\alpha}\bar{\auxii}_{b}(\lambda\gamma^{b})_{\beta}}{(\lambda\bar{\lambda})}$
vanishes completely, as well as the one before $\frac{\auxii^{a}(\gamma_{a}\bar{\lambda})^{\alpha}\bar{\lambda}_{\beta}}{(\lambda\bar{\lambda})}$
and one ends up with the expression in (\ref{LinProjGenLam}) of the
proposition. \rembreak\rem{ Note that the somewhat unexpected (at
least when comparing to\\
 the toy-model) term $\frac{\bar{\auxii}_{c}(\gamma^{c}\gamma_{b}\lambda)^{\alpha}\auxii^{d}(\bar{\lambda}\gamma^{b}\gamma_{d})_{\beta}}{(\lambda\bar{\lambda})}$
can be rewritten as:
\begin{eqnarray*}
\lqn{\frac{\bar{\auxii}_{c}(\gamma^{c}\gamma_{b}\lambda)^{\alpha}\auxii^{d}(\bar{\lambda}\gamma^{b}\gamma_{d})_{\beta}}{(\lambda\bar{\lambda})}=}\\
 & = & -\frac{\bar{\auxii}_{c}(\gamma^{c}\gamma_{b}\lambda)^{\alpha}\auxii^{d}(\bar{\lambda}\gamma_{d}\gamma^{b})_{\beta}}{(\lambda\bar{\lambda})}=\\
 & \stackrel{Fierz}{=} & \frac{\bar{\auxii}_{c}(\gamma^{c}\gamma_{b})^{\alpha}\tief{\beta}\auxii^{d}(\bar{\lambda}\gamma_{d}\gamma^{b}\lambda)}{(\lambda\bar{\lambda})}+\frac{\bar{\auxii}_{c}(\overbrace{\gamma^{c}\gamma_{b}\gamma_{d}}^{\gamma_{b}\gamma_{d}\gamma^{c}+2\delta_{b}^{c}\gamma_{d}-2\gamma_{b}\delta_{d}^{c}}\bar{\lambda})^{\alpha}\auxii^{d}(\lambda\gamma^{b})_{\beta}}{(\lambda\bar{\lambda})}=\\
 & = & \frac{\bar{\auxii}_{c}(\gamma^{c}\gamma_{b})^{\alpha}\tief{\beta}\auxii^{d}(\bar{\lambda}\overbrace{\gamma_{d}\gamma^{b}}^{-\gamma^{b}\gamma_{d}\lqn{{\scriptstyle +2\delta_{d}^{b}}}}\lambda)}{(\lambda\bar{\lambda})}+2\frac{\auxii^{d}(\gamma_{d}\bar{\lambda})^{\alpha}\bar{\auxii}_{b}(\lambda\gamma^{b})_{\beta}}{(\lambda\bar{\lambda})}-\frac{4\auxi(\gamma_{b}\bar{\lambda})^{\alpha}(\lambda\gamma^{b})_{\beta}}{(\lambda\bar{\lambda})}=\\
 & = & 2\bar{\auxii}_{c}\auxii^{b}(\gamma^{c}\gamma_{b})^{\alpha}\tief{\beta}+2\frac{\auxii^{d}(\gamma_{d}\bar{\lambda})^{\alpha}\bar{\auxii}_{b}(\lambda\gamma^{b})_{\beta}}{(\lambda\bar{\lambda})}-\frac{4\auxi(\gamma_{b}\bar{\lambda})^{\alpha}(\lambda\gamma^{b})_{\beta}}{(\lambda\bar{\lambda})}
\end{eqnarray*}
}\rembreak

Next we plug (\ref{Proj-inv}) and (\ref{Proj-modulus-repeated})
into (\ref{linprojgen}):
\begin{eqnarray}
\lqn{\rest_{(f)\bot}^{\alpha\beta}(\rho,\bar{\rho})=-\tfrac{1}{2}\tfrac{f\left(\auxi\right)}{1+\sqrt{1-\auxi}}\auxii^{a}\gamma_{a}^{\alpha\beta}+}\nonumber \\
 &  & \hspace{-0.5cm}-2\auxi f'\left(\auxi\right)\tfrac{\left(\frac{1+\sqrt{1-\auxi}}{2f(\auxi)\sqrt{1-\auxi}}\lambda^{\alpha}+\frac{1}{4\os{}{\bar{f}}(\auxi)\sqrt{1-\auxi}}\auxii^{a}\left(\gamma_{a}\bar{\lambda}\right)^{\alpha}\right)\left(\frac{1+\sqrt{1-\auxi}}{2f(\auxi)\sqrt{1-\auxi}}\lambda^{\beta}+\frac{1}{4\os{}{\bar{f}}(\auxi)\sqrt{1-\auxi}}\auxii^{b}\left(\bar{\lambda}\gamma_{b}\right)^{\alpha}\right)}{\frac{\left(1+\sqrt{1-\auxi}\right)}{2\abs{f(\auxi)}^{2}\left(1-\auxi\right)}(\lambda\bar{\lambda})}+\nonumber \\
 &  & \hspace{-0.5cm}+f'\left(\auxi\right)\tfrac{\left(\frac{1+\sqrt{1-\auxi}}{2f(\auxi)\sqrt{1-\auxi}}\lambda^{\alpha}+\frac{1}{4\os{}{\bar{f}}(\auxi)\sqrt{1-\auxi}}\auxii^{a}\left(\gamma_{a}\bar{\lambda}\right)^{\alpha}\right)\auxii^{b}\bigl(\frac{1+\sqrt{1-\auxi}}{2\os{}{\bar{f}}(\auxi)\sqrt{1-\auxi}}(\bar{\lambda}\gamma_{b})^{\beta}+\frac{1}{4f(\auxi)\sqrt{1-\auxi}}\bar{\auxii}_{d}(\lambda\overbrace{{\scriptstyle \gamma^{d}\gamma_{b}}}^{\hspace{-0.5cm}-\gamma_{b}\gamma^{d}\lqn{{\scriptscriptstyle +2\delta_{b}^{d}}}})^{\beta}\bigr)}{\frac{\left(1+\sqrt{1-\auxi}\right)}{2\abs{f(\auxi)}^{2}\left(1-\auxi\right)}(\lambda\bar{\lambda})}+\nonumber \\
 &  & \hspace{-0.5cm}+\left(\tfrac{f\left(\auxi\right)}{2\sqrt{1-\auxi}\left(1+\sqrt{1-\auxi}\right)}+{\scriptstyle \left(1-\sqrt{1-\auxi}\right)}f'\left(\auxi\right)\right)\times\nonumber \\
 &  & \hspace{-0.5cm}\times\tfrac{\auxii^{a}\bigl(\frac{1+\sqrt{1-\auxi}}{2\os{}{\bar{f}}(\auxi)\sqrt{1-\auxi}}(\gamma_{a}\bar{\lambda})^{\alpha}+\frac{1}{4f(\auxi)\sqrt{1-\auxi}}\bar{\auxii}_{d}(\overbrace{{\scriptstyle \gamma_{a}\gamma^{d}}}^{\hspace{-0.5cm}-\gamma^{d}\gamma_{a}\lqn{{\scriptscriptstyle +2\delta_{a}^{d}}}}\lambda)^{\alpha}\bigr)\bigl(\frac{1+\sqrt{1-\auxi}}{2f(\auxi)\sqrt{1-\auxi}}\lambda^{\beta}+\frac{1}{4\os{}{\bar{f}}(\auxi)\sqrt{1-\auxi}}\auxii^{b}(\bar{\lambda}\gamma_{b})^{\alpha}\bigr)}{\frac{\left(1+\sqrt{1-\auxi}\right)}{2\abs{f(\auxi)}^{2}\left(1-\auxi\right)}(\lambda\bar{\lambda})}+\nonumber \\
 &  & \hspace{-0.5cm}-\tfrac{1}{2\left(1+\sqrt{1-\auxi}\right)}\left(\tfrac{f\left(\auxi\right)}{2\sqrt{1-\auxi}\left(1+\sqrt{1-\auxi}\right)}+f'\left(\auxi\right)\right)\times\nonumber \\
 &  & \hspace{-0.5cm}\times\tfrac{\!\!\auxii^{a}\!\bigl(\!\!\frac{1+\sqrt{1-\auxi}}{2\os{}{\bar{f}}(\auxi)\sqrt{1-\auxi}}(\gamma_{a}\bar{\lambda})^{\alpha}\!+\!\frac{1}{4f(\auxi)\sqrt{1-\auxi}}\bar{\auxii}_{d}(\overbrace{{\scriptstyle \gamma_{a}\gamma^{d}}}^{\hspace{-0.5cm}-\gamma^{d}\gamma_{a}\lqn{{\scriptscriptstyle +2\delta_{a}^{d}}}}\lambda)^{\alpha}\!\bigr)\auxii^{b}\!\bigl(\!\frac{1+\sqrt{1-\auxi}}{2\os{}{\bar{f}}(\auxi)\sqrt{1-\auxi}}\bar{(\lambda}\gamma_{b})^{\beta}\!+\!\frac{1}{4f(\auxi)\sqrt{1-\auxi}}\bar{\auxii}_{d}(\lambda\overbrace{{\scriptstyle \gamma^{d}\gamma_{b}}}^{\hspace{-0.5cm}-\gamma_{b}\gamma^{d}\lqn{{\scriptscriptstyle +2\delta_{b}^{d}}}})^{\beta}\!\bigr)\!\!}{\frac{\left(1+\sqrt{1-\auxi}\right)}{2\abs{f(\auxi)}^{2}\left(1-\auxi\right)}(\lambda\bar{\lambda})}\nonumber \\
 &  & {}
\end{eqnarray}
Again we change the order of the gamma-matrices as indicated above
the curly brackets for the same reason as before and use Fierz identities
of the form $(\gamma^{a}\lambda)_{\alpha}(\gamma_{a}\lambda)_{\beta}\propto(\lambda\gamma_{a}\lambda)\gamma_{\alpha\beta}^{a}=0$
(footnote \ref{fn:Fierz}) or $\auxii^{a}(\gamma_{a}\lambda)_{\alpha}=0$
(\ref{ProjZsquareZero}) and their complex conjugate counterparts.
Then we sort the terms by collecting first those of the form $\lambda\otimes\lambda$,
then $\lambda\otimes\bar{\lambda}$ and so on, and obtain (again before
simplifying)
\begin{eqnarray}
\lqn{\rest_{(f)\bot}^{\alpha\beta}(\rho,\bar{\rho})=}\nonumber \\
 & = & \biggl\{-\tfrac{\auxi f'\left(\auxi\right)\bar{f}(\auxi)\left(1+\sqrt{1-\auxi}\right)}{f(\auxi)}+\tfrac{f'\left(\auxi\right)\bar{f}(\auxi)\auxi}{f(\auxi)}+\tfrac{\auxi\bar{f}(\auxi)}{f(\auxi)\left(1+\sqrt{1-\auxi}\right)}\left(\auxi f'\left(\auxi\right)+\tfrac{f\left(\auxi\right)}{2\sqrt{1-\auxi}}\right)+\nonumber \\
 &  & -\tfrac{\auxi^{2}\bar{f}(\auxi)}{f(\auxi)\left(1+\sqrt{1-\auxi}\right)^{2}}\left(f'\left(\auxi\right)+\tfrac{f\left(\auxi\right)}{2\sqrt{1-\auxi}\left(1+\sqrt{1-\auxi}\right)}\right)\biggr\}\frac{\lambda^{\alpha}\lambda^{\beta}}{(\lambda\bar{\lambda})}+\nonumber \\
 &  & \biggl\{-\frac{1}{2}\auxi f'\left(\auxi\right)+\tfrac{f'\left(\auxi\right)\left(1+\sqrt{1-\auxi}\right)}{2}+\tfrac{\auxi}{2\left(1+\sqrt{1-\auxi}\right)^{2}}\left(\auxi f'\left(\auxi\right)+\tfrac{f\left(\auxi\right)}{2\sqrt{1-\auxi}}\right)+\nonumber \\
 &  & -\tfrac{\auxi}{2\left(1+\sqrt{1-\auxi}\right)}\left(f'\left(\auxi\right)+\tfrac{f\left(\auxi\right)}{2\sqrt{1-\auxi}\left(1+\sqrt{1-\auxi}\right)}\right)\biggr\}\frac{\lambda^{\alpha}\auxii^{b}\left(\bar{\lambda}\gamma_{b}\right)^{\alpha}}{(\lambda\bar{\lambda})}+\nonumber \\
 &  & \biggl\{-\frac{1}{2}\auxi f'\left(\auxi\right)+\tfrac{\auxi f'\left(\auxi\right)}{2\left(1+\sqrt{1-\auxi}\right)}+\tfrac{1}{2}\left(\auxi f'\left(\auxi\right)+\tfrac{f\left(\auxi\right)}{2\sqrt{1-\auxi}}\right)+\nonumber \\
 &  & -\tfrac{\auxi}{2\left(1+\sqrt{1-\auxi}\right)}\left(f'\left(\auxi\right)+\tfrac{f\left(\auxi\right)}{2\sqrt{1-\auxi}\left(1+\sqrt{1-\auxi}\right)}\right)\biggr\}\frac{\auxii^{a}\left(\gamma_{a}\bar{\lambda}\right)^{\alpha}\lambda^{\beta}}{(\lambda\bar{\lambda})}+\nonumber \\
 &  & \biggl\{-\tfrac{\auxi f'\left(\auxi\right)f(\auxi)}{4\os{}{\bar{f}}(\auxi)\left(1+\sqrt{1-\auxi}\right)}+\tfrac{f(\auxi)f'\left(\auxi\right)}{4\os{}{\bar{f}}(\auxi)}+\tfrac{f(\auxi)}{4\os{}{\bar{f}}(\auxi)\left(1+\sqrt{1-\auxi}\right)}\left(\auxi f'\left(\auxi\right)+\tfrac{f\left(\auxi\right)}{2\sqrt{1-\auxi}}\right)+\nonumber \\
 &  & -\tfrac{f(\auxi)}{4\os{}{\bar{f}}(\auxi)}\left(f'\left(\auxi\right)+\tfrac{f\left(\auxi\right)}{2\sqrt{1-\auxi}\left(1+\sqrt{1-\auxi}\right)}\right)\biggr\}\frac{\auxii^{a}\left(\gamma_{a}\bar{\lambda}\right)^{\alpha}\auxii^{b}\left(\bar{\lambda}\gamma_{b}\right)^{\alpha}}{(\lambda\bar{\lambda})}+\nonumber \\
 &  & -\tfrac{f\left(\auxi\right)}{2\left(1+\sqrt{1-\auxi}\right)}\auxii^{a}\gamma_{a}^{\alpha\beta}
\end{eqnarray}
Again, after collecting within the curly brackets the terms that contain
$f'$ and those that don't, and after simplifying, one arrives indeed
at (\ref{linprojgenlam}).$\hfill{\scriptstyle \square}$ \rembreak\rem{one
might do a backwards-consistency-check, starting with the \\
results (\ref{LinProjGenLam}) and (\ref{linprojgenlam}) and using
(\ref{ProjGeneral}) and (\ref{Proj-modulus}) to arrive back at (\ref{LinProjGen})
and (\ref{linprojgen}).} 

3. The fact that $\Pi_{(f)\bot}(\rho,\bar{\rho})$ and $\rest_{(f)\bot}(\rho,\bar{\rho})$
seen as linear endomorphisms map to spinors which are $\gamma$-orthogonal
to $\lambda^{\alpha}\equiv P_{(f)}^{\alpha}(\rho,\bar{\rho})$ as
claimed in (\ref{LinProjGamma-orth}) and (\ref{linprojgamma-orth})
follows directly from the fact that $P_{(f)}^{\alpha}(\rho,\bar{\rho})$
is a pure spinor according to (\ref{proj-prop1}). I.e. the variation
of (\ref{proj-prop1}) yields
\begin{eqnarray}
0 & = & \delta\left(P_{(f)}(\rho,\bar{\rho})\gamma^{m}P_{(f)}(\rho,\bar{\rho})\right)=\\
 & = & 2\delta\rho^{\beta}\partial_{\rho^{\beta}}P_{(f)}^{\alpha}(\rho,\bar{\rho})\left(\gamma^{m}P_{(f)}(\rho,\bar{\rho})\right)_{\alpha}+\nonumber \\
 &  & +2\delta\bar{\rho}_{\beta}\partial_{\bar{\rho}_{\beta}}P_{(f)}^{\alpha}(\rho,\bar{\rho})\left(\gamma^{m}P_{(f)}(\rho,\bar{\rho})\right)_{\alpha}
\end{eqnarray}
For this to hold for all variations $\delta\rho,\delta\bar{\rho}$,
we necessarily need (\ref{LinProjGamma-orth}) and (\ref{linprojgamma-orth}).
One can also easily verify these equations directly by using the form
(\ref{LinProjGenLam}) and (\ref{linprojgenlam}) of the projector-matrices.$\hfill{\scriptstyle \square}$ 

4. On the constraint surface of spinors $\rho^{\alpha}=\lambda^{\alpha}$
with $\lambda\gamma^{m}\lambda=0$ we have $\auxii^{a}=0$ and $\auxi=0$
(\ref{xzvanishforpure}). Plugging this into either (\ref{LinProjGen}),
(\ref{linprojgen}) or (\ref{LinProjGenLam}) and (\ref{linprojgenlam})
and assuming that $f$ and $f'$ are non-singular at $\auxi=0$, one
obtains immediately the claimed results (\ref{Proj-matrix-onshell})
and (\ref{proj-matrix-on-shell}). The remaining matrix on the constraint
surface $\Pi_{\bot}\equiv\one-\tfrac{1}{2}\frac{(\gamma_{a}\bar{\lambda})\otimes(\lambda\gamma^{a})}{(\lambda\bar{\lambda})}$
finally is indeed idempotent:
\begin{eqnarray}
\Pi_{\bot}^{2} & = & \one-\frac{(\gamma_{a}\bar{\lambda})\otimes(\lambda\gamma^{a})}{(\lambda\bar{\lambda})}+\tfrac{1}{4}\frac{(\lambda\overbrace{\gamma^{a}\gamma_{b}}^{-\gamma_{b}\gamma^{a}\lqn{{\scriptstyle +2\delta_{b}^{a}}}}\bar{\lambda})(\gamma_{a}\bar{\lambda})\otimes(\lambda\gamma^{b})}{(\lambda\bar{\lambda})^{2}}=\\
 & \stackrel{{\rm Fierz}}{=} & \one-\tfrac{1}{2}\frac{(\gamma_{a}\bar{\lambda})\otimes(\lambda\gamma^{a})}{(\lambda\bar{\lambda})}=\Pi_{\bot}\label{PiSurfIdem}
\end{eqnarray}
This last property is certainly already well-known, as $\Pi_{\bot}$
is the transposed of the projection matrix that maps the antighost
$\omega_{z\alpha}$ to its gauge invariant part (\cite{Oda:2001zm,Oda:2004bg,Berkovits:2010zz}).
$\hfill{\scriptstyle \square}$

5. The projection-properties (\ref{LinProjProp}) and (\ref{linprojprop})
directly follow from derivatives of the projection property $P_{(f)}=P_{(f)}\circ P_{(f)}$
(\ref{proj-prop2}) :
\begin{eqnarray}
\partial_{\rho^{\beta}}P_{(f)}^{\alpha}(\rho,\bar{\rho}) & = & \partial_{\rho^{\beta}}P_{(f)}^{\alpha}\left(P_{(f)}(\rho,\bar{\rho}),\bar{P}_{(f)}(\rho,\bar{\rho})\right)=\\
 & = & \partial_{\rho^{\beta}}P_{(f)}^{\gamma}(\rho,\bar{\rho})\bei{\partial_{\tilde{\rho}^{\gamma}}P_{(f)}^{\alpha}(\tilde{\rho},\bar{\tilde{\rho}})}{\tilde{\rho}=P_{(f)}(\rho,\bar{\rho})}+\nonumber \\
 &  & +\partial_{\rho^{\beta}}\bar{P}_{(f)\bot\gamma}(\rho,\bar{\rho})\bei{\partial_{\bar{\tilde{\rho}}_{\gamma}}P_{(f)}^{\alpha}(\tilde{\rho},\bar{\tilde{\rho}})}{\tilde{\rho}=P_{(f)}(\rho,\bar{\rho})}\\
 & = & \Pi_{(f)\bot\gamma}^{\alpha}\left(P_{(f)}(\rho,\bar{\rho}),\bar{P}_{(f)}(\rho,\bar{\rho})\right)\Pi_{(f)\bot\beta}^{\gamma}(\rho,\bar{\rho})+\nonumber \\
 &  & +\underbrace{\rest_{(f)\bot}^{\alpha\gamma}\left(P_{(f)}(\rho,\bar{\rho}),\bar{P}_{(f)}(\rho,\bar{\rho})\right)}_{=0\quad\mbox{\tiny\eqref{proj-matrix-on-shell}}}\bar{\rest}_{(f)\bot\gamma\beta}(\rho,\bar{\rho})
\end{eqnarray}
The last term vanishes (as indicated) according to $(\ref{proj-matrix-on-shell})$.
This then proves (\ref{LinProjProp}). In the same way we can write
\begin{eqnarray}
\partial_{\bar{\rho}_{\beta}}P_{(f)}^{\alpha}(\rho,\bar{\rho}) & = & \partial_{\bar{\rho}_{\beta}}P_{(f)}^{\alpha}\left(P_{(f)}(\rho,\bar{\rho}),\bar{P}_{(f)}(\rho,\bar{\rho})\right)=\\
 & = & \partial_{\bar{\rho}_{\beta}}P_{(f)}^{\gamma}(\rho,\bar{\rho})\bei{\partial_{\tilde{\rho}^{\gamma}}P_{(f)}^{\alpha}(\tilde{\rho},\bar{\tilde{\rho}})}{\tilde{\rho}=P_{(f)}(\rho,\bar{\rho})}+\nonumber \\
 &  & +\partial_{\bar{\rho}_{\beta}}\bar{P}_{(f)\bot\gamma}(\rho,\bar{\rho})\bei{\partial_{\bar{\tilde{\rho}}_{\gamma}}P_{(f)}^{\alpha}(\tilde{\rho},\bar{\tilde{\rho}})}{\tilde{\rho}=P_{(f)}(\rho,\bar{\rho})}=\\
 & = & \Pi_{(f)\bot\gamma}^{\alpha}\left(P_{(f)}(\rho,\bar{\rho}),\bar{P}_{(f)}(\rho,\bar{\rho})\right)\rest_{(f)\bot}^{\gamma\beta}(\rho,\bar{\rho})+\nonumber \\
 &  & +\underbrace{\rest_{(f)\bot}^{\alpha\gamma}\left(P_{(f)}(\rho,\bar{\rho}),\bar{P}_{(f)}(\rho,\bar{\rho})\right)}_{=0\quad\mbox{\tiny\eqref{proj-matrix-on-shell}}}\bar{\Pi}_{(f)\bot\gamma}\hoch{\beta}(\rho,\bar{\rho})
\end{eqnarray}
 which proves (\ref{linprojprop}). $\hfill{\scriptstyle \square}$

6. The trace of (\ref{LinProjGen}) is given by: 
\begin{eqnarray}
\lqn{\Pi_{(f)\bot\alpha}^{\alpha}(\rho,\bar{\rho})=}\nonumber \\
 & = & 16f\left(\auxi\right)+f'\left(\auxi\right)\frac{(\rho\gamma^{b}\rho)\bar{\auxii}_{b}}{\rho\bar{\rho}}-2\auxi f'\left(\auxi\right)+\nonumber \\
 &  & -\tfrac{1}{1+\sqrt{1-\auxi}}\!\!\left(\!\! f\left(\auxi\right)\delta_{b}^{a}\!+\!\tfrac{1}{2}\biggl(\tfrac{f\left(\auxi\right)}{2\sqrt{1-\auxi}\left(1+\sqrt{1-\auxi}\right)}+f'\left(\auxi\right)\biggr)\auxii^{a}\bar{\auxii}_{b}\!\!\right)\!\!\frac{(\rho\overbrace{\gamma^{b}\gamma_{a}}^{-\gamma_{a}\lqn{{\scriptstyle \gamma^{b}+2\delta_{a}^{b}}}}\bar{\rho})}{\rho\bar{\rho}}+\nonumber \\
 &  & +\tfrac{1}{1+\sqrt{1-\auxi}}\left(\tfrac{f\left(\auxi\right)}{2\sqrt{1-\auxi}}+\auxi f'\left(\auxi\right)\right)\frac{\auxii^{a}(\bar{\rho}\gamma_{a}\bar{\rho})}{\rho\bar{\rho}}=\\
 & \stackrel{\mbox{\small\eqref{zsquarezero}}}{=} & 16f\left(\auxi\right)+\nonumber \\
 &  & \negthickspace\negthickspace-\tfrac{1}{1+\sqrt{1-\auxi}}\!\!\left(\!\! f\!\left(\auxi\right)\!\Bigl(\!10+\tfrac{\auxi}{\sqrt{1-\auxi}\left(1+\sqrt{1-\auxi}\right)}-\tfrac{\auxi}{\sqrt{1-\auxi}}\Bigr)-2\auxi(\auxi-1)f'\!\left(\auxi\right)\!\!\right)\!\!=\qquad\quad\\
 & = & \left(16-\tfrac{9+\sqrt{1-\auxi}}{1+\sqrt{1-\auxi}}\right)f\left(\auxi\right)+\tfrac{2\auxi(\auxi-1)}{1+\sqrt{1-\auxi}}f'\left(\auxi\right)=\\
 & = & \left(11-\tfrac{4(1-\sqrt{1-\auxi})}{1+\sqrt{1-\auxi}}\right)f\left(\auxi\right)-2(1-\auxi)\left(1-\sqrt{1-\auxi}\right)f'\left(\auxi\right)
\end{eqnarray}
It as a nice consistency check that the same result is obtained by
taking the trace of (\ref{LinProjGenLam}). Note that this result
is in general%
\footnote{\label{fn:differentialEqForf} For the toy model in the appendix,
the choice $f=1$ yields a proper projection matrix also off the constraint
surface (see (\ref{Pi1quad}) on page \pageref{Pi1quad})), whose
trace is the same as on the surface. This is here not the case for
$f=1$. In order to find an $f$ such that $\Pi_{(f)\bot}^{2}=\Pi_{(f)\bot}$
even off the surface, we would need at least that $\tr\Pi_{(f)\bot}=11$
off the surface which gives a differential equation 
\[
\left(11-\tfrac{4(1-\sqrt{1-\auxi})}{1+\sqrt{1-\auxi}}\right)f\left(\auxi\right)-2(1-\auxi)\left(1-\sqrt{1-\auxi}\right)f'\left(\auxi\right)\stackrel{!}{=}11
\]
This differential equation is slightly easier to solve in the parametrization
$\Auxi$ of equation (\ref{newx}) where it turns into 
\[
(11-4\Auxi)\fa(\Auxi)-(1-\Auxi)\Auxi\fa'(\Auxi)\stackrel{!}{=}11
\]
The homogeneous equation $\frac{\fa'(\Auxi)}{\fa(\Auxi)}=\frac{(11-4\Auxi)}{(1-\Auxi)\Auxi}$
is solved by $\fa(\Auxi)=C\cdot\frac{\Auxi^{11}}{(1-\Auxi)^{7}}$
with some constant $C$. In order to solve the inhomogeneous equation,
one can promote $C$ to a function $C(\Auxi)$. Plugging the ansatz
\[
\fa(\Auxi)=C(\Auxi)\cdot\tfrac{\Auxi^{11}}{(1-\Auxi)^{7}}
\]
back into the differential equation yields $C'(\Auxi)=-\frac{11(1-\Auxi)^{6}}{\Auxi^{12}}$
which can be integrated to 
\[
C(\Auxi)=\frac{1}{210\Auxi^{5}}-\frac{1-\Auxi}{42\Auxi^{6}}+\frac{(1-\Auxi)^{2}}{14\Auxi^{7}}-\frac{(1-\Auxi)^{3}}{6\Auxi^{8}}+\frac{(1-\Auxi)^{4}}{3\Auxi^{9}}-\frac{3\,(1-\Auxi)^{5}}{5\Auxi^{10}}+\frac{(1-\Auxi)^{6}}{\Auxi^{11}}+{\rm const\qquad\fussend}
\]
\frem{In the original parametrization the homogeneous equation reads
$\frac{f'\left(\auxi\right)}{f\left(\auxi\right)}=\frac{7+15\sqrt{1-\auxi}}{2(1-\auxi)\auxi}$
and is solved by $f(\auxi)=C\cdot(\frac{\auxi}{1-\auxi})^{\frac{7}{2}}(\frac{1-\sqrt{1-\auxi}}{1+\sqrt{1-\auxi}})^{\frac{15}{2}}$
(using $\frac{\partial}{\partial x}\log(\frac{x}{1-x})=\frac{1}{x\left(1-x\right)}$).
Plugging the ansatz 
\[
f(\auxi)=C(\auxi)\cdot(\frac{\auxi}{1-\auxi})^{\frac{7}{2}}(\frac{1-\sqrt{1-\auxi}}{1+\sqrt{1-\auxi}})^{\frac{15}{2}}
\]
back into the differential equation yields
\begin{eqnarray*}
\tilde{C}' & = & -\frac{11(1-\auxi)^{\frac{5}{2}}\left(1+\sqrt{1-\auxi}\right)^{16}}{2\auxi^{12}}
\end{eqnarray*}
}%
} not equal to $11$ for $\auxi\neq0$. For $\auxi=0$ it clearly reduces
to $11$ for all $f$. This completes the proof of proposition \ref{prop:linProj}.\end{proof}

\subsection{Proof of proposition 3}

\label{app:proof3}\begin{proof}[Proof of proposition \ref{prop:hermiticity}]1.
From the explicit form of $\Pi_{(h)\bot}$ and $\rest_{(h)\bot}$
in (\ref{HermLinProjLam}) and (\ref{hermlinprojlam}), it is obvious
that $\Pi_{(h)\bot}$ is Hermitian and $\rest_{(h)\bot}$ is symmetric.
So what remains to show that $h$ is the only function for which this
is the case. 

To do so, let us assume that $f$ is any function defined and differentiable
at least on $I=[0,b[$ (with $b\leq1$) for which we have Hermiticity
in this region (i.e. for all $\rho^{\alpha}$ for which $\auxi\in I$).
So starting from (\ref{LinProjGenLam}), we build the difference of
$\Pi_{(f)\bot}$ and its Hermitian conjugate and require it to vanish
\begin{eqnarray}
0 & \stackrel{!}{=} & \Pi_{(f)\bot}(\rho,\bar{\rho})-\Pi_{(f)\bot}^{\dagger}(\rho,\bar{\rho})=\\
 & = & \left(f(\auxi)-\bar{f}(\auxi)\right)\underbrace{\left(\one-\tfrac{1}{2}\frac{(\gamma^{a}\bar{\lambda})\otimes(\lambda\gamma_{a})}{(\lambda\bar{\lambda})}\right)}_{\Pi_{\bot}^{\dagger}=\Pi_{\bot}}+\nonumber \\
 &  & -\left(\tfrac{\bar{f}(\auxi)}{2\left(1+\sqrt{1-\auxi}\right)}-\tfrac{\bar{f}(\auxi)f'\left(\auxi\right)\left(1-\auxi\right)}{f(\auxi)}\right)\frac{\lambda\otimes\bar{\auxii}_{c}\left(\lambda\gamma^{c}\right)}{(\lambda\bar{\lambda})}+\nonumber \\
 &  & +\left(\tfrac{f(\auxi)}{2\left(1+\sqrt{1-\auxi}\right)}-\tfrac{f(\auxi)\bar{f}'\left(\auxi\right)\left(1-\auxi\right)}{\bar{f}(\auxi)}\right)\frac{\auxii^{c}\left(\gamma_{c}\bar{\lambda}\right)\otimes\bar{\lambda}}{(\lambda\bar{\lambda})}+\nonumber \\
 &  & -{\scriptstyle 2\left(1-\auxi\right)\left(1-\sqrt{1-\auxi}\right)}\left(f'\left(\auxi\right)-\bar{f}'\left(\auxi\right)\right)\underbrace{\frac{\lambda\otimes\bar{\lambda}}{(\lambda\bar{\lambda})}}_{\Pi_{\lambda}^{\dagger}=\Pi_{\lambda}}+\nonumber \\
 &  & -\tfrac{\left(f\left(\auxi\right)-\bar{f}\left(\auxi\right)\right)}{8\left(1+\sqrt{1-\auxi}\right)^{2}}\underbrace{\frac{\bar{\auxii}_{c}(\gamma^{c}\gamma_{b}\lambda)\otimes\auxii^{d}(\bar{\lambda}\gamma^{b}\gamma_{d})}{(\lambda\bar{\lambda})}}_{\bar{\auxii}_{c}(\gamma^{c}\gamma_{b}\Pi_{\lambda}\gamma^{b}\gamma_{d})\auxii^{d}}\label{hermproof}
\end{eqnarray}
A priori this has to vanish for all $\rho^{\alpha}$ for which the
projection matrices are defined (so for which $\auxi\in I$). This
is according to (\ref{Proj-inv}) and (\ref{ProjZsquareZero}) equivalent
to saying that it has to vanish for all pure spinors $\lambda^{\alpha}$
and for all $\auxii^{a}$ with $\auxii^{a}(\gamma_{a}\lambda)_{\alpha}=0$
and with $\auxi\in I$ and $f(\auxi)\neq0$ (the last one is in order
for (\ref{Proj-inv}) to be well-defined). For simplicity let us first
assume that 
\begin{equation}
f(\auxi)\neq0\quad\forall\auxi\in I\quad({\rm assumption})\label{h-has-no-zeros-assump}
\end{equation}
We will at the end of this proof relax this assumption. Note that
the constraint $\auxii^{a}(\gamma_{a}\lambda)_{\alpha}=0$ is no constraint
on the modulus $\auxi$, because to this constraint always exist solutions
of the form $\auxii^{a}=(\lambda\gamma^{a}\chi)$ with some arbitrary
spinor $\chi^{\alpha}$, and these solutions can be rescaled to obtain
any $\auxi$ one wants. It is therefore necessary (not sufficient)
that (\ref{hermproof}) has to hold for any pure spinor $\lambda^{\alpha}$
and for all $\auxi\in I$. 

Although it seems very plausible for a generic $\lambda^{\alpha}$,
it is not completely obvious if the 5 matrices $\one-\frac{(\gamma^{a}\bar{\lambda})\otimes(\lambda\gamma_{a})}{2(\lambda\bar{\lambda})},\frac{\lambda\otimes\bar{\auxii}_{c}\left(\lambda\gamma^{c}\right)}{(\lambda\bar{\lambda})}$,$\frac{\auxii^{c}\left(\gamma_{c}\bar{\lambda}\right)\otimes\bar{\lambda}}{(\lambda\bar{\lambda})}$,
$\frac{\lambda\otimes\bar{\lambda}}{(\lambda\bar{\lambda})}$ and
$\frac{\bar{\auxii}_{c}(\gamma^{c}\gamma_{b}\lambda)\otimes\auxii^{d}(\bar{\lambda}\gamma^{b}\gamma_{d})}{(\lambda\bar{\lambda})}$
are all linearly independent, because there are many ways of rewriting
them using Fierz identities. If they are independent, its coefficients
have to vanish separately. From the first line of (\ref{hermproof})
this would already imply reality of the function $f$, i.e. $\bar{f}\stackrel{!}{=}f$,
which in turn also would make vanish the last two lines of (\ref{hermproof})
and we would be left with a differential equation from the second
(or equivalently from the third) line. 

In order to be on the safe side, however, we could expand (\ref{hermproof})
in a basis from which we know about linear Independence, namely $\{\one,\gamma^{e_{1}e_{2}},\gamma^{e_{1}e_{2}e_{3}e_{4}}\}$.
An alternative way is to perform all kind of contractions of (\ref{hermproof})
until we have enough conditions on $h$ to really make the full equation
(\ref{hermproof}) hold.  We will follow the latter path. Note that
in the following we will use the identities $(\lambda\gamma^{c}\lambda)=(\bar{\lambda}\gamma^{c}\bar{\lambda})=0$
and $\auxii^{c}(\gamma_{c}\lambda)=\bar{\auxii}_{c}(\gamma^{c}\bar{\lambda})=0$
wherever possible without always mentioning this. Let us start with
the most obvious contraction, namely taking the trace of (\ref{hermproof}):
\begin{eqnarray}
0 & \stackrel{!}{=} & 11\left(f(\auxi)-\bar{f}(\auxi)\right)-{\scriptstyle 2\left(1-\auxi\right)\left(1-\sqrt{1-\auxi}\right)}\left(f'\left(\auxi\right)-\bar{f}'\left(\auxi\right)\right)+\nonumber \\
 &  & -\tfrac{\left(f\left(\auxi\right)-\bar{f}\left(\auxi\right)\right)}{8\left(1+\sqrt{1-\auxi}\right)^{2}}\frac{\bar{\auxii}_{c}(\bar{\lambda}\gamma^{b}\overbrace{\gamma_{d}\gamma^{c}}^{\gamma_{d}\hoch c+\delta_{d}^{c}}\gamma_{b}\lambda)\auxii^{d}}{(\lambda\bar{\lambda})}
\end{eqnarray}
Using the identity 
\begin{equation}
\gamma^{b}\gamma^{e_{1}\ldots e_{l}}\gamma_{b}=(-)^{l}(10-2l)\gamma^{e_{1}\ldots e_{l}}\quad\forall l\in\{0,\ldots,5\}\label{sandwich}
\end{equation}
we can rewrite $\gamma^{b}(\gamma_{d}\hoch c+\delta_{d}^{c})\gamma_{b}=6\gamma_{d}\hoch c+10\delta_{d}^{c}=-6\gamma^{c}\gamma_{d}+16\delta_{d}^{c}$.
Having in mind that $\auxii^{c}\bar{\auxii}_{c}=2\auxi$ and $\frac{\auxi}{\left(1+\sqrt{1-\auxi}\right)^{2}}=\frac{1-\sqrt{1-\auxi}}{1+\sqrt{1-\auxi}}$,
we arrive at 
\begin{eqnarray}
\underline{\tr(\ref{hermproof}):}\quad0 & \stackrel{!}{=} & \left(f(\auxi)-\bar{f}(\auxi)\right)\left(11-4\tfrac{1-\sqrt{1-\auxi}}{1+\sqrt{1-\auxi}}\right)+\nonumber \\
 &  & -2\left(1-\auxi\right)\left(1-\sqrt{1-\auxi}\right)\left(f'\left(\auxi\right)-\bar{f}'\left(\auxi\right)\right)\label{hermproof-trace}
\end{eqnarray}
This is the first condition on $f(\auxi$). 

Next let us contract (\ref{hermproof}) $\bar{\lambda}$ from the
left and $\lambda$ from the right, which yields
\begin{equation}
\underline{\left(\bar{\lambda}(\ref{hermproof})\lambda\right):}\quad0\stackrel{!}{=}\left(f(\auxi)-\bar{f}(\auxi)\right)(\lambda\bar{\lambda})-{\scriptstyle 2\left(1-\auxi\right)\left(1-\sqrt{1-\auxi}\right)}\left(f'\left(\auxi\right)-\bar{f}'\left(\auxi\right)\right)(\lambda\bar{\lambda})
\end{equation}
Using (\ref{hermproof-trace}) one can eliminate $f'$ and obtains
\begin{equation}
0\stackrel{!}{=}-\left(f(\auxi)-\bar{f}(\auxi)\right)\Bigl(10-4\underbrace{\tfrac{1-\sqrt{1-\auxi}}{1+\sqrt{1-\auxi}}}_{\leq1}\Bigr)(\lambda\bar{\lambda})
\end{equation}
Because of the indicated inequality, the factor $10-4\tfrac{1-\sqrt{1-\auxi}}{1+\sqrt{1-\auxi}}$
can never become zero. Further more the above condition has to hold
for all pure spinors $\lambda^{\alpha}$, in particular those with
$(\lambda\bar{\lambda})\neq0$. This finally shows 
\begin{equation}
\bar{f}(\auxi)\stackrel{!}{=}f(\auxi)\quad\forall\auxi\in I
\end{equation}
This was what we suspected from the beginning (right after (\ref{hermproof})),
but now we are sure. 

Finally let us plug the above reality result into (\ref{hermproof})
and contract the equation from the right with $\lambda$ and with
$\bar{\auxii}_{d}(\lambda\gamma^{d})$ from the left. This yields
\begin{eqnarray}
0 & \stackrel{!}{=} & \left(\tfrac{f(\auxi)}{2\left(1+\sqrt{1-\auxi}\right)}-f'\left(\auxi\right)\left(1-\auxi\right)\right)\auxii^{c}\bar{\auxii}_{d}(\lambda\overbrace{\gamma^{d}\gamma_{c}}^{-\gamma_{c}\gamma^{d}\lqn{{\scriptstyle +2\delta_{c}^{d}}}}\bar{\lambda})=\\
 & = & \left(\tfrac{f(\auxi)}{2\left(1+\sqrt{1-\auxi}\right)}-f'\left(\auxi\right)\left(1-\auxi\right)\right)4\auxi(\lambda\bar{\lambda})
\end{eqnarray}
Again this should hold for all pure spinors $\lambda^{\alpha}$, so
including those with $(\lambda\bar{\lambda})\neq0$. This is the case
if and only if either $\auxi=0$  or
\begin{eqnarray}
\tfrac{f'\left(\auxi\right)}{f(\auxi)} & = & \tfrac{1}{2\left(1-\auxi\right)\left(1+\sqrt{1-\auxi}\right)}\quad\forall\auxi\in I/\{0\}\label{h-DGL}
\end{eqnarray}
As we have assumed continuity of $f'$ at $\auxi=0$ in the proposition,
the equation will also hold for $\auxi=0$. We now have already extracted
all information from (\ref{hermproof}), which is clear when plugging
the above equation (together with $\bar{f}=f$) back into (\ref{hermproof})
and arriving at $0\stackrel{!}{=}0.$ Using now that $f(0)=1$, this
can be finally uniquely integrated to $f(\auxi)=h(\auxi)=\frac{1+\sqrt{1-\auxi}}{2\sqrt{1-\auxi}}$,
as it was defined in (\ref{h-def}). Together with its derivative
$h'(\auxi)=\frac{1}{4\sqrt{1-\auxi}^{3}}$ (\ref{hprime}), this function
indeed obeys the differential equation (\ref{h-DGL}) and the boundary
condition $h(0)=1$ is met. So already Hermiticity of $\Pi_{(f)\bot}$fixes
$f$ uniquely to be $h$, so that no further conditions may come from
symmetry of $\rest_{(f)\bot}$. And indeed, $\rest_{(h)\bot}$ with
$h$ of (\ref{h-def}) is obviously symmetric in (\ref{hermlinprojlam}),
as was noted already at the beginning of the proof. 

In (\ref{h-has-no-zeros-assump}) we had made the assumption that
$f(\auxi)\neq0$ for all $\auxi\in I=[0,b[$. Now assume that there
is a zero of $f$ at $\auxi_{0}\in I/\{0\}$ (at $\auxi=0$ we necessarily
have $f(0)=1$). Then the above proof of $f=h$ holds only for the
interval $[0,\auxi_{0}[$, but because of continuity it has to hold
also for $\auxi_{0}$. And $h=\frac{1+\sqrt{1-\auxi}}{2\sqrt{1-\auxi}}$
does not have any zeros on $[0,1[$ which disproves the assumption
that $f$ had a zero at $\auxi_{0}\in I.\qquad$ \hfill~$\quad{\scriptstyle \square}$ 

2. When $\Pi_{(h)\bot}$ and $\rest_{(h)\bot}$ are blocks of a Hermitian
matrix, it means that $\partial_{\rho^{\alpha}}P_{(h)}^{\beta}-\partial_{\bar{\rho}_{\beta}}\bar{P}_{(h)\alpha}=0$
as well as $\bar{\partial}^{[\alpha}P_{(h)}^{\beta]}=0$ (with $\bar{\partial}^{\alpha}\equiv\partl{\bar{\rho}_{\alpha}}$).
This means that $P_{(h)}$ regarded as a 1-form $\bs P\equiv P_{(h)}^{\alpha}\de\bar{\rho}_{\alpha}+\bar{P}_{(h)\alpha}\de\rho^{\alpha}$
is closed and thus locally exact, i.e. $\de\bs P=0$. As this is a
flat vector space, there are no global obstructions and one can find
a potential for the projector. And indeed the potential $\Phi\equiv\tfrac{(\rho\bar{\rho})}{2}(1+\sqrt{1-\auxi})$
of equation (\ref{Potential}) has the following derivatives
\begin{eqnarray}
\partial_{\bar{\rho}_{\alpha}}\Phi(\rho,\bar{\rho}) & = & \tfrac{\rho^{\alpha}}{2}(1+\sqrt{1-\auxi})-\tfrac{(\rho\bar{\rho})}{2}\partial_{\bar{\rho}_{\alpha}}\auxi\cdot\tfrac{1}{2\sqrt{1-\auxi}}=\\
 & \stackrel{\mbox{\tiny\eqref{xbyrhoderivative}}}{=} & \tfrac{\rho^{\alpha}}{2}\underbrace{\Bigl(1+\sqrt{1-\auxi}+\tfrac{\auxi}{\sqrt{1-\auxi}}\Bigr)}_{\frac{1+\sqrt{1-\auxi}}{\sqrt{1-\auxi}}}-\tfrac{\auxii^{a}(\gamma_{a}\bar{\rho})^{\alpha}}{4\sqrt{1-\auxi}}
\end{eqnarray}
This result indeed agrees with $P_{(h)}^{\alpha}(\rho,\bar{\rho})$
of equation (\ref{hermProjector}). Via complex conjugation we finally
obtain also $\bar{P}_{(h)\alpha}=\partial_{\rho^{\alpha}}\Phi$.\hfill $\quad{\scriptstyle \square}$ 

3. The absolute value squared of $P_{(h)}^{\alpha}(\rho,\bar{\rho})$
is given by\global\long\def\smallref#1{{\scriptstyle #1}}
 
\begin{eqnarray}
\lqn{P_{(h)}^{\alpha}(\rho,\bar{\rho})\bar{P}_{(h)\alpha}(\rho,\bar{\rho})=}\nonumber \\
 & \stackrel{\mbox{\tiny\eqref{hermProjector}}}{=} & \left(\tfrac{1+\sqrt{1-\auxi}}{2\sqrt{1-\auxi}}\rho^{\alpha}-\tfrac{\auxii^{a}(\bar{\rho}\gamma_{a})^{\alpha}}{4\sqrt{1-\auxi}}\right)\left(\tfrac{1+\sqrt{1-\auxi}}{2\sqrt{1-\auxi}}\bar{\rho}_{\alpha}-\tfrac{(\gamma^{b}\rho)_{\alpha}\bar{\auxii}_{b}}{4\sqrt{1-\auxi}}\right)=\\
 & = & \tfrac{\left(1+\sqrt{1-\auxi}\right)^{2}}{4(1-\auxi)}(\rho\bar{\rho})-\tfrac{1+\sqrt{1-\auxi}}{8\left(1-\auxi\right)}\underbrace{(\rho\gamma^{b}\rho)}_{(\rho\bar{\rho})\auxii^{b}\:\mbox{\tiny\eqref{xzDefRep}}}\bar{\auxii}_{b}+\nonumber \\
 &  & -\tfrac{1+\sqrt{1-\auxi}}{8\left(1-\auxi\right)}\auxii^{a}\underbrace{(\bar{\rho}\gamma_{a}\bar{\rho})}_{(\rho\bar{\rho})\bar{\auxii}_{a}\:\mbox{\tiny\eqref{xzDefRep}}}+\tfrac{1}{16\left(1-\auxi\right)}\auxii^{a}(\bar{\rho}\underbrace{\gamma_{a}\gamma^{b}}_{-\gamma^{b}\gamma_{a}+2\delta_{a}^{b}}\rho)\bar{\auxii}_{b}=\\
 & \stackrel{\mbox{\tiny\eqref{zsquarezero}}}{=} & (\rho\bar{\rho})\Bigl\{\!\tfrac{\left(1+\sqrt{1-\auxi}\right)^{2}}{4(1-\auxi)}\!-\!\tfrac{1+\sqrt{1-\auxi}}{8\left(1-\auxi\right)}\auxii^{b}\bar{\auxii}_{b}\!-\!\tfrac{1+\sqrt{1-\auxi}}{8\left(1-\auxi\right)}\auxii^{a}\bar{\auxii}_{a}\!+\!\tfrac{1}{8\left(1-\auxi\right)}\auxii^{a}\bar{\auxii}_{a}\!\Bigr\}\qquad
\end{eqnarray}
Using $\auxii^{a}\bar{\auxii}_{a}=2\auxi$ \eqref{xzDefRep} and simplifying
a bit, the expression in the curly bracket becomes $\tfrac{1}{2}(1+\sqrt{1-\auxi})$
and thus (via the definition \eqref{Potential} of $\Phi$) indeed
yields $P_{(h)}^{\alpha}(\rho,\bar{\rho})\bar{P}_{(h)\alpha}(\rho,\bar{\rho})=\Phi(\rho,\bar{\rho})$.\rem{This
resembles very much a Kähler potential. Then $\Pi_{(h)\bot\beta}^{\alpha}\de\bar{\rho}_{\alpha}\wedge\de\rho^{\beta}$
would be the corresponding Kähler form.:$\bs{\partial}\bar{\bs{\partial}}\Phi=\bs{\partial}\left(P_{(h)}^{\alpha}\de\bar{\rho}_{\alpha}\right)=\Pi_{(h)\bot\beta}^{\alpha}\de\rho^{\beta}\de\bar{\rho}_{\alpha}$
But degenerate?}\rembreak\rem{Here hidden in a \LyX{}-note the efforts
of describing the behaviour of the complete projection for $\auxi\to1$.
Probably no limit exists.} \end{proof}\providecommand{\href}[2]{#2}\begingroup\raggedright

\endgroup 
\end{document}